\documentclass[acmsmall,screen,nonacm,final,author]{acmart} 
\pdfoutput=1
\usepackage[T1]{fontenc}
\usepackage{graphicx}
\usepackage{tikz}
\usetikzlibrary{calc}
\usetikzlibrary{positioning}
\usetikzlibrary{arrows,automata}
\usetikzlibrary{patterns,patterns.meta}
\usepackage{xspace}
\usepackage{xcolor}
\usepackage{amsmath}
\usepackage{thmtools}
\usepackage{thm-restate}
\usepackage{mathrsfs}
\usepackage{ifdraft}
\usepackage{wrapfig}
\usepackage{subcaption}
\usepackage{verbatim}
\usepackage{booktabs}
\usepackage{multirow}
\usepackage{paralist}
\usepackage{lipsum}
\usepackage{amsbsy}
\usepackage[normalem]{ulem}
\usepackage{textcomp}
\usepackage{stmaryrd}
\usepackage{mathpartir}
\usepackage{keyval}
\usepackage{algorithm}
\usepackage{algorithmicx}
\usepackage[noend]{algpseudocode}
\usepackage{makecell}
\usepackage{array}
\input{style/notation.sty} 
\usepackage{newunicodechar}
\newunicodechar{∃}{\ensuremath{\exists}}
\newunicodechar{∀}{\ensuremath{\forall}}
\newunicodechar{θ}{\ensuremath{\theta}}
\newunicodechar{τ}{\ensuremath{\tau}}
\newunicodechar{φ}{\ensuremath{\varphi}}
\newunicodechar{ξ}{\ensuremath{\xi}}
\newunicodechar{ζ}{\ensuremath{\zeta}}
\newunicodechar{ψ}{\ensuremath{\psi}}
\newunicodechar{π}{\ensuremath{\pi}}
\newunicodechar{α}{\ensuremath{\alpha}}
\newunicodechar{β}{\ensuremath{\beta}}
\newunicodechar{γ}{\ensuremath{\gamma}}
\newunicodechar{δ}{\ensuremath{\delta}}
\newunicodechar{ε}{\ensuremath{\varepsilon}}
\newunicodechar{κ}{\ensuremath{\kappa}}
\newunicodechar{λ}{\ensuremath{\lambda}}
\newunicodechar{μ}{\ensuremath{\mu}}
\newunicodechar{ρ}{\ensuremath{\rho}}
\newunicodechar{σ}{\ensuremath{\sigma}}
\newunicodechar{ω}{\ensuremath{\omega}}
\newunicodechar{Γ}{\ensuremath{\Gamma}}
\newunicodechar{Φ}{\ensuremath{\Phi}}
\newunicodechar{Δ}{\ensuremath{\Delta}}
\newunicodechar{Σ}{\ensuremath{\Sigma}}
\newunicodechar{Π}{\ensuremath{\Pi}}
\newunicodechar{∑}{\ensuremath{\Sigma}}
\newunicodechar{∏}{\ensuremath{\Pi}}
\newunicodechar{Θ}{\ensuremath{\Theta}}
\newunicodechar{Ω}{\ensuremath{\Omega}}
\newunicodechar{⇒}{\ensuremath{\Rightarrow}}
\newunicodechar{⇐}{\ensuremath{\Leftarrow}}
\newunicodechar{⇔}{\ensuremath{\Leftrightarrow}}
\newunicodechar{→}{\ensuremath{\rightarrow}}
\newunicodechar{←}{\ensuremath{\leftarrow}}
\newunicodechar{↔}{\ensuremath{\leftrightarrow}}
\newunicodechar{¬}{\ensuremath{\neg}}
\newunicodechar{∧}{\ensuremath{\land}}
\newunicodechar{∨}{\ensuremath{\lor}}
\newunicodechar{≠}{\ensuremath{\neq}}
\newunicodechar{≡}{\ensuremath{\equiv}}
\newunicodechar{∼}{\ensuremath{\sim}}
\newunicodechar{≈}{\ensuremath{\approx}}
\newunicodechar{≥}{\ensuremath{\geq}}
\newunicodechar{≤}{\ensuremath{\leq}}
\newunicodechar{≫}{\ensuremath{\gg}}
\newunicodechar{≪}{\ensuremath{\ll}}
\newunicodechar{∅}{\ensuremath{\emptyset}}
\newunicodechar{⊆}{\ensuremath{\subseteq}}
\newunicodechar{⊂}{\ensuremath{\subset}}
\newunicodechar{∩}{\ensuremath{\cap}}
\newunicodechar{⋂}{\ensuremath{\cap}}
\newunicodechar{∪}{\ensuremath{\cup}}
\newunicodechar{⋃}{\ensuremath{\cup}}
\newunicodechar{⊎}{\ensuremath{\uplus}}
\newunicodechar{∈}{\ensuremath{\in}}
\newunicodechar{∉}{\ensuremath{\not\in}}
\newunicodechar{⊤}{\ensuremath{\top}}
\newunicodechar{⊥}{\ensuremath{\bot}}
\newunicodechar{₀}{\ensuremath{_0}}
\newunicodechar{₁}{\ensuremath{_1}}
\newunicodechar{₂}{\ensuremath{_2}}
\newunicodechar{₃}{\ensuremath{_3}}
\newunicodechar{₄}{\ensuremath{_4}}
\newunicodechar{₅}{\ensuremath{_5}}
\newunicodechar{₆}{\ensuremath{_6}}
\newunicodechar{₇}{\ensuremath{_7}}
\newunicodechar{₈}{\ensuremath{_8}}
\newunicodechar{₉}{\ensuremath{_9}}
\newunicodechar{⁰}{\ensuremath{^0}}
\newunicodechar{¹}{\ensuremath{^1}}
\newunicodechar{²}{\ensuremath{^2}}
\newunicodechar{³}{\ensuremath{^3}}
\newunicodechar{⁴}{\ensuremath{^4}}
\newunicodechar{⁵}{\ensuremath{^5}}
\newunicodechar{⁶}{\ensuremath{^6}}
\newunicodechar{⁷}{\ensuremath{^7}}
\newunicodechar{⁸}{\ensuremath{^8}}
\newunicodechar{⁹}{\ensuremath{^9}}
\newunicodechar{𝔹}{\ensuremath{\mathbb{B}}}
\newunicodechar{ℝ}{\ensuremath{\mathbb{R}}}
\newunicodechar{ℕ}{\ensuremath{\mathbb{N}}}
\newunicodechar{ℂ}{\ensuremath{\mathbb{C}}}
\newunicodechar{ℚ}{\ensuremath{\mathbb{Q}}}
\newunicodechar{𝕋}{\ensuremath{\mathbb{T}}}
\newunicodechar{𝕏}{\ensuremath{\mathbb{X}}}
\newunicodechar{ℤ}{\ensuremath{\mathbb{Z}}}
\newunicodechar{✓}{\ding{51}}
\newunicodechar{✗}{\ding{55}}
\newunicodechar{◊}{\ensuremath{\lozenge}}
\newunicodechar{□}{\ensuremath{\square}}
\newunicodechar{𝓐}{\ensuremath{\mathcal{A}}}
\newunicodechar{𝓑}{\ensuremath{\mathcal{B}}}
\newunicodechar{𝓒}{\ensuremath{\mathcal{C}}}
\newunicodechar{𝓓}{\ensuremath{\mathcal{D}}}
\newunicodechar{𝓔}{\ensuremath{\mathcal{E}}}
\newunicodechar{𝓕}{\ensuremath{\mathcal{F}}}
\newunicodechar{𝓖}{\ensuremath{\mathcal{G}}}
\newunicodechar{𝓗}{\ensuremath{\mathcal{H}}}
\newunicodechar{𝓘}{\ensuremath{\mathcal{I}}}
\newunicodechar{𝓙}{\ensuremath{\mathcal{J}}}
\newunicodechar{𝓚}{\ensuremath{\mathcal{K}}}
\newunicodechar{𝓛}{\ensuremath{\mathcal{L}}}
\newunicodechar{𝓜}{\ensuremath{\mathcal{M}}}
\newunicodechar{𝓝}{\ensuremath{\mathcal{N}}}
\newunicodechar{𝓞}{\ensuremath{\mathcal{O}}}
\newunicodechar{𝓟}{\ensuremath{\mathcal{P}}}
\newunicodechar{𝓠}{\ensuremath{\mathcal{Q}}}
\newunicodechar{𝓡}{\ensuremath{\mathcal{R}}}
\newunicodechar{𝓢}{\ensuremath{\mathcal{S}}}
\newunicodechar{𝓣}{\ensuremath{\mathcal{T}}}
\newunicodechar{𝓤}{\ensuremath{\mathcal{U}}}
\newunicodechar{𝓥}{\ensuremath{\mathcal{V}}}
\newunicodechar{𝓦}{\ensuremath{\mathcal{W}}}
\newunicodechar{𝓧}{\ensuremath{\mathcal{X}}}
\newunicodechar{𝓨}{\ensuremath{\mathcal{Y}}}
\newunicodechar{𝓩}{\ensuremath{\mathcal{Z}}}
\newunicodechar{…}{\ensuremath{\ldots}}
\newunicodechar{∗}{\ensuremath{\ast}}
\newunicodechar{⊢}{\ensuremath{\vdash}}
\newunicodechar{⊧}{\ensuremath{\models}}
\newunicodechar{′}{\ensuremath{'}}
\newunicodechar{″}{\ensuremath{''}}
\newunicodechar{‴}{\ensuremath{'''}}
\newunicodechar{∥}{\ensuremath{\|}}
\newunicodechar{⊕}{\ensuremath{\oplus}}
\newunicodechar{⁺}{\ensuremath{^+}}
\newunicodechar{⊇}{\ensuremath{\supseteq}}
\newunicodechar{∘}{\ensuremath{\circ}}
\newunicodechar{∙}{\ensuremath{\cdot}}
\newunicodechar{⋅}{\ensuremath{\cdot}}
\newunicodechar{≈}{\ensuremath{\approx}}
\newunicodechar{×}{\ensuremath{\times}}
\newunicodechar{∞}{\ensuremath{\infty}}
\newunicodechar{⊑}{\ensuremath{\sqsubseteq}}

\input{style/hmsc-tikz.sty}

\ifoptionfinal
{\usepackage[disable]{todonotes}}
{\usepackage{todonotes}}

\usepackage{hyperref}
\usepackage[capitalise]{cleveref}
\usepackage[shortlabels]{enumitem}

\crefformat{section}{\S#2#1#3} % § for sections
\crefmultiformat{section}{\S#2#1#3}{ and~\S#2#1#3}{, \S#2#1#3}{, and~\S#2#1#3} % § section ranges
\crefrangeformat{section}{\S#3#1#4 to~\S#5#2#6} % § section ranges
\AtBeginDocument{%
  }

\newcommand{\toolname}{\textsc{Sprout}\xspace} 
\newcommand{\newtoolname}{\toolname$\!(\netarch)$\xspace}
\newcommand{\sessionstar}{\textsc{Session*}\xspace}

\newcommand{\muCLP}{$\mu$CLP\xspace}
\newcommand{\muval}{\textsc{MuVal}\xspace}

\newcommand{\bagavail}{\mathsf{bagavail}}
\newcommand{\bag}{\texttt{bag}\xspace}
\newcommand{\ptpbox}{\texttt{p2p}\xspace}
\newcommand{\mailbox}{\texttt{mb}\xspace}
\newcommand{\senderbox}{\texttt{sb}\xspace}
\newcommand{\monobox}{\texttt{monob}\xspace}

\begin{document}


\title{Implementability of Global Distributed Protocols modulo Network Architectures}

\author{Elaine Li}
%\authornote{corresponding author}
\orcid{0000-0003-0173-4498}
\affiliation{%
	\institution{New York University}
	\city{New York}
	\country{USA}
}
\email{efl9013@nyu.edu}

\author{Thomas Wies}
\orcid{0000-0003-4051-5968}
\affiliation{%
	\institution{New York University}
	\city{New York}
	\country{USA}
}
\email{wies@cs.nyu.edu}

%%\renewcommand{\shortauthors}{Anonymous}

\begin{abstract}
Global protocols specify distributed, message-passing protocols from a birds-eye view, and are used as a specification for synthesizing local implementations. Implementability asks whether a given global protocol admits a distributed implementation. 
We present the first comprehensive investigation of global protocol implementability modulo network architectures. 
We propose a set of network-parametric Coherence Conditions, and exhibit sufficient assumptions under which it precisely characterizes implementability. 
We further reduce these assumptions
% , stated as algebraic properties of formal languages,
to a minimal set of operational axioms describing insert and remove behavior of individual message buffers. 
Our reduction immediately establishes that five commonly studied asynchronous network architectures, namely peer-to-peer FIFO, mailbox, senderbox, monobox and bag, are instances of our network-parametric result. 
%Additionally, our reduction yields that heterogeneous network architectures that mix FIFO queues and unordered multisets are also instances of our network-parametric result. 
We use our characterization to derive optimal complexity results for implementability modulo networks, relationships between classes of implementable global protocols, and symbolic algorithms for deciding implementability modulo networks. 
We implement the latter in the first network-parametric tool \newtoolname, and show that it achieves network generality without sacrificing performance and modularity. 

%%% Local Variables:
%%% mode: latex
%%% TeX-master: "main"
%%% End:

\end{abstract}

%\begin{CCSXML}
%<ccs2012>
%  <concept_id>00000000.0000000.0000000</concept_id>
%  <concept_desc>Do Not Use This Code, Generate the Correct Terms for Your Paper</concept_desc>
%  <concept_significance>500</concept_significance>
%  <concept_id>00000000.00000000.00000000</concept_id>
%  <concept_desc>Do Not Use This Code, Generate the Correct Terms for Your Paper</concept_desc>
%  <concept_significance>300</concept_significance>
%  <concept_id>00000000.00000000.00000000</concept_id>
%  <concept_desc>Do Not Use This Code, Generate the Correct Terms for Your Paper</concept_desc>
%  <concept_significance>100</concept_significance>
%  <concept_id>00000000.00000000.00000000</concept_id>
%  <concept_desc>Do Not Use This Code, Generate the Correct Terms for Your Paper</concept_desc>
%  <concept_significance>100</concept_significance>
%</ccs2012>
%\end{CCSXML}

%\ccsdesc[500]{Do Not Use This Code~Generate the Correct Terms for Your Paper}
%\ccsdesc[300]{Do Not Use This Code~Generate the Correct Terms for Your Paper}
%\ccsdesc{Do Not Use This Code~Generate the Correct Terms for Your Paper}
%\ccsdesc[100]{Do Not Use This Code~Generate the Correct Terms for Your Paper}

\keywords{
	Asynchronous communication, network semantics, global protocol verification, communicating state machines, multiparty session types
}

\maketitle
        
\section{Introduction}
\label{sec:intro}

Distributed, message-passing protocols are notoriously difficult to implement correctly. 
Asynchronous network interleavings of independent events make the detection of communication errors such as orphan messages, unspecified receptions and deadlocks especially challenging. 
Global protocol specifications enjoy the illusion of synchrony, specifying send and receive events atomically from the perspective of an omniscient observer, thereby ruling out large classes of communication errors by construction.
Global protocols are then used as a blueprint from which to synthesize correct-by-construction distributed implementations. 
This approach inspires a top-down verification methodology, whose central decision problem is realizability, also called implementability. 
Implementability asks whether a given global specification admits a distributed implementation, in other words, whether each participant's local view is sufficient for all participants to collectively follow the protocol designer's omniscient intent without relying on covert coordination. 

Global protocols as a specification mechanism do not a priori commit to a particular network model. This choice is instead deferred to the assignment of semantics to global protocols, which in turn depends on the target distributed implementation model. 
By far the most common network model studied in the literature is peer-to-peer (p2p) FIFO communication, in which protocol participants are pairwise connected by ordered channels. 
Realizability of global protocols targeting p2p FIFO communcation has been thoroughly studied in the domain of high-level message sequence charts (HMSCs)~\cite{DBLP:conf/sdl/MauwR97,
	DBLP:conf/ac/GenestMP03,%
	DBLP:conf/acsd/GenestM05,%
	DBLP:conf/concur/GazagnaireGHTY07,%
	DBLP:journals/tosem/RoychoudhuryGS12, 
	DBLP:journals/tse/AlurEY03, 
	DBLP:journals/tcs/Lohrey03, 
	DBLP:conf/concur/AlurY99,DBLP:conf/mfcs/MuschollP99, 
	DBLP:conf/stacs/Morin02,DBLP:journals/jcss/GenestMSZ06}, and p2p FIFO is the predominant network model for multiparty session types~\cite{DBLP:conf/popl/HondaYC08, DBLP:conf/concur/BocchiHTY10, DBLP:conf/tgc/BocchiDY12, DBLP:journals/jlp/ToninhoY17, DBLP:journals/pacmpl/00020HNY20, DBLP:conf/cav/LiSWZ23, DBLP:journals/pacmpl/HinrichsenBK20, DBLP:journals/lmcs/HinrichsenBK22, DBLP:journals/pacmpl/JacobsHK23, Multris, DBLP:conf/cav/LiSWZ23} and choreographic programming frameworks~\cite{DBLP:journals/tcs/Cruz-FilipeM20,DBLP:conf/ecoop/GiallorenzoMPRS21,DBLP:journals/pacmpl/HirschG22,DBLP:journals/corr/abs-2111-03484,DBLP:journals/corr/abs-2303-00924}. Implementations of such frameworks can be found in more than a dozen programming languages. 
While p2p FIFO communication is ubiquitous in practice, other network models also hold practical and theoretical interest. For example, Erlang and Go implement mailbox communication, and many correctness proofs of classic distributed algorithms such as leader election~\cite{DBLP:journals/toplas/GallagerHS83} and clock synchronization~\cite{DBLP:books/acm/19/Lamport19b} rely on bag semantics.

Naturally, the question arises whether the top-down methodology of global protocols can be extended to target different network models. 
The diversity of communication models and ways to define them semantically has been studied as a topic of independent interest~\cite{DBLP:journals/fac/ChevrouHQ16, DBLP:conf/fm/ChevrouH0Q19, DBLP:journals/pacmpl/GiustoFLL23}, but the more salient research question is whether existing results targeting a particular network model can be generalized to the others. 
Decision problems that have been studied in a multi-network setting include synchronizability and reachability~\cite{DBLP:journals/tcs/Lohrey03,DBLP:journals/lmcs/FinkelL23, DBLP:conf/concur/BolligGFLLS21,DBLP:conf/concur/DelpyMS24} of communicating finite-state machines. 
Communication errors in one network model do not necessarily arise in another, making the design of truly network-parametric algorithms very challenging.  
Consequently, the vast majority of existing results target a fixed network model, and the question of if and how they generalize is left unanswered.

We provide a positive answer to this question through the first comprehensive investigation of the implementability problem for global specifications modulo network architectures. We represent network architectures as formal languages, and propose a set of semantic conditions, called Generalized Coherence Conditions, parametric in a choice of network architecture. We exhibit sufficient conditions under which our network-parametric conditions are sound and complete with respect to implementability. These sufficient conditions take the shape of a set of abstract, algebraic properties over formal languages defining the behavior of network architectures, and are the key to enabling a unified, network-parametric proof of soundness and completeness. This first set of contributions constitutes a network-parametric, precise solution to global protocol implementability.

We take network parametricity one step further by reducing our formal language-theoretic sufficiency conditions to a set of simple, operational axioms defined over buffer data structures. These buffer axioms greatly simplify the task of determining whether a network architecture is an instance of our result, by boiling it down to proving a few specifications about insert and remove operations. As an immediate consequence of this reduction, we conclude that five commonly considered network models, namely peer-to-peer FIFO, mailbox (n-1), senderbox (1-n), monobox (1-1) and bag, are all instances of our parametric result. A perhaps more surprising result is that our axioms leave the communication topology completely unconstrained and accommodate even heterogeneous network architectures featuring multiple buffer data structures.

As a final set of contributions, we derive concrete algorithms for checking implementability modulo network architectures and obtain decidability and complexity results for various protocol fragments.
For finite global protocols, we show that implementability is co-NP-complete for all aforementioned network architectures, even for the special case when global protocols belong to the fragment of directed choice multiparty session types~\cite{DBLP:conf/popl/HondaYC08}.
For symbolic global protocols with first-order logic transition constraints, we show that our network-parametric implementability characterization admits an encoding in the first-order fixpoint logic \muCLP~\cite{DBLP:journals/pacmpl/UnnoTGK23}. 

Thus, our characterization yields a symbolic algorithm for checking implementability of protocols with infinitely many states and data. We implement our derived symbolic algorithm in a tool \newtoolname, extending an existing tool that decides implementability for global protocols assuming p2p FIFO networks~\cite{DBLP:conf/cav/LiSWZ25}. Our evaluation shows that \newtoolname extends modularly to different architectures without sacrificing performance. 

\smallskip
\noindent\emph{Contributions.} In summary, our contributions are: 
\begin{itemize}
\item We introduce the implementability problem for global protocol specifications modulo network architectures.
\item We give conditions that characterize implementability in a network-parametric fashion and provide sufficient conditions under which this characterization is sound and complete.
\item We reduce these sufficient conditions to a simple axiomatic model of buffer data structures that can be easily checked for any given network architecture. 
\item  We derive decidability and complexity results for finite-state protocol implementability modulo five commonly considered network architectures as well as algorithms for checking implementability of symbolic global protocols.
\item We present \newtoolname, the first network-parametric implementability checker for symbolic global protocols. 
\end{itemize}

%%% Local Variables:
%%% mode: latex
%%% TeX-master: "main"
%%% End:

\section{Motivating Example} 
\label{sec:overview}

\begin{figure}[t]
	\centering
	\resizebox{1.01\textwidth}{!}{
	\begin{tikzpicture}[hmsc,baseline,msg/.append style={tight,above=2pt,pos=.2},node distance=8ex]

%\node at (s-1) {};

\begin{scope}[msc=s1]
  \begin{scope}[participant=s]
    \node[head];
    \node[event] {};
  \end{scope}
  \node[yshift=.7pt] at (s-1) {$\server$};

  \begin{scope}[participant=wa]
    \node[head];
    \node[event] {};
  \end{scope}
  \node at (wa-1) {$\workerOne$};

  \begin{scope}[participant=wb]
    \node[head];
    \node[no event] {};
  \end{scope}
  \node at (wb-1) {$\workerTwo$};

  \draw[messages]
    (s-2) edge[blue] node[msg,midway]{$\fullMsg$} (wa-2)
  ;
\end{scope}

\begin{scope}[msc=s2,xshift=24ex]
  \begin{scope}[participant=s]
    \node[head];
    \node[event] {};
    \node[event] {};
  \end{scope}
  \node[yshift=.7pt] at (s-1) {$\server$};

  \begin{scope}[participant=wa]
    \node[head];
    \node[no event] {};
    \node[no event] {};
  \end{scope}
  \node at (wa-1) {$\workerOne$};

  \begin{scope}[participant=wb]
    \node[head];
    \node[event] {};
    \node[event] {};
  \end{scope}
  \node at (wb-1) {$\workerTwo$};

  \draw[messages]
    (s-2) edge[blue] node[msg,pos=.25]{$\fullMsg$} (wb-2)
    (wb-3) edge[blue] node[msg,pos=.25]{$\fullMsg$} (s-3)
    %(s-4) edge node[msg,midway]{$z \cond{z > 0}$} (ba-4)
  ;
\end{scope}

\begin{scope}[msc=s3,xshift=48ex]
  \begin{scope}[participant=s]
    \node[head];
    \node[event] {};
    \node[event] {};
    \node[event] {};
  \end{scope}
  \node[yshift=.7pt] at (s-1) {$\server$};

  \begin{scope}[participant=wa]
    \node[head];
    \node[event] {};
    \node[no event] {};
    \node[no event] {};
  \end{scope}
  \node at (wa-1) {$\workerOne$};

  \begin{scope}[participant=wb]
    \node[head];
    \node[no event] {};
    \node[event] {};
    \node[event] {};
  \end{scope}
  \node at (wb-1) {$\workerTwo$};

  \draw[messages]
    (s-2) edge[red] node[msg,midway]{$\halfMsg$} (wa-2)
    (s-3) edge[red] node[msg,pos=.25]{$\halfMsg$} (wb-3)
    (wb-4) edge[red] node[msg,pos=.25]{$\halfMsg$} (s-4)
    %(s-4) edge node[msg,midway]{$z \cond{z > 0}$} (ba-4)
  ;
\end{scope}

\begin{scope}[msc=s4,yshift=-16ex,xshift=-12ex]
  \begin{scope}[participant=s]
    \node[head];
    \node[no event] {};
    \node[event] {};
    \node[event] {};
  \end{scope}
  \node[yshift=.7pt] at (s-1) {$\server$};

  \begin{scope}[participant=wa]
    \node[head];
    \node[event] {};
    \node[event] {};
    \node[no event] {};
  \end{scope}
  \node at (wa-1) {$\workerOne$};

  \begin{scope}[participant=wb]
    \node[head];
    \node[event] {};
    \node[no event] {};
    \node[event] {};
  \end{scope}
  \node at (wb-1) {$\workerTwo$};

  \draw[messages]
    (wa-2) edge[red] node[msg,midway]{$\halfMsg$} (wb-2)
    (wa-3) edge[red] node[msg,midway]{$\halfMsg$} (s-3)
    (wb-4) edge[red] node[msg,pos=.25]{$\halfMsg$} (s-4)
    %(s-4) edge node[msg,midway]{$z \cond{z > 0}$} (ba-4)
  ;
\end{scope}

\begin{scope}[msc=s5, yshift=-16ex,xshift=12ex]
  \begin{scope}[participant=s]
    \node[head];
    \node[event] {};
  \end{scope}
  \node[yshift=.7pt] at (s-1) {$\server$};

  \begin{scope}[participant=wa]
    \node[head];
    \node[event] {};
  \end{scope}
  \node at (wa-1) {$\workerOne$};

  \begin{scope}[participant=wb]
    \node[head];
    \node[no event] {};
  \end{scope}
  \node at (wb-1) {$\workerTwo$};

  \draw[messages]
    (wa-2) edge[blue] node[msg,midway]{$\fullMsg$} (s-2)
  ;
\end{scope}

\begin{scope}[msc=s6, yshift=-18ex,xshift=36ex]
  \begin{scope}[participant=s]
    \node[head];
    \node[no event] {};
    \node[event] {};
  \end{scope}
  \node[yshift=.7pt] at (s-1) {$\server$};

  \begin{scope}[participant=wa]
    \node[head];
    \node[event] {};
    \node[no event] {};
  \end{scope}
  \node at (wa-1) {$\workerOne$};

  \begin{scope}[participant=wb]
    \node[head];
    \node[event] {};
    \node[event] {};
  \end{scope}
  \node at (wb-1) {$\workerTwo$};

  \draw[messages]
    (wa-2) edge[red] node[msg,midway]{$\halfMsg$} (wb-2)
    (wb-3) edge[red] node[msg,pos=.25]{$\halfMsg$} (s-3)
  ;
\end{scope}

\begin{scope}[msc=s7, yshift=-18ex,xshift=59ex]
  \begin{scope}[participant=s]
    \node[head];
    \node[event] {};
  \end{scope}
  \node[yshift=.7pt] at (s-1) {$\server$};

  \begin{scope}[participant=wa]
    \node[head];
    \node[event] {};
  \end{scope}
  \node at (wa-1) {$\workerOne$};

  \begin{scope}[participant=wb]
    \node[head];
    \node[no event] {};
  \end{scope}
  \node at (wb-1) {$\workerTwo$};

  \draw[messages]
    (wa-2) edge[red] node[msg,midway]{$\halfMsg$} (s-2)
  ;
\end{scope}

% \begin{scope}[msc=s2,final]
%   \begin{scope}[participant=e]
%     \node[event] {};
%     \node[event] {};
%     \node[event] {};
%   \end{scope}

%   \begin{scope}[participant=o]
%     \node[head];
%     \node[event] {};
%     \node[event] {};
%     \node[event] {};
%   \end{scope}

%   \draw[messages]
%     (e-2) edge node[msg]{1} (o-3)
%     (o-2) edge node[msg]{1} (e-3)
%     (o-4) edge node[msg,midway]{win} (e-4)
%   ;
% \end{scope}

% \begin{scope}[msc=s3,final]
%   \begin{scope}[participant=e]
%     \node[head,right=of s2o];
%     \node[event] {};
%     \node[event] {};
%     \node[event] {};
%   \end{scope}

%   \begin{scope}[participant=o]
%     \node[head];
%     \node[event] {};
%     \node[event] {};
%     \node[event] {};
%   \end{scope}

%   \draw[messages]
%     (e-2) edge node[msg]{1} (o-3)
%     (o-2) edge node[msg]{0} (e-3)
%     (e-4) edge node[msg,midway]{win} (o-4)
%   ;
% \end{scope}

% \begin{scope}[msc=s4,final]
%   \begin{scope}[participant=e]
%     \node[head,right=of s3o];
%     \node[event] {};
%     \node[event] {};
%     \node[event] {};
%   \end{scope}

%   \begin{scope}[participant=o]
%     \node[head];
%     \node[event] {};
%     \node[event] {};
%     \node[event] {};
%   \end{scope}

%   \draw[messages]
%     (e-2) edge node[msg]{0} (o-3)
%     (o-2) edge node[msg]{1} (e-3)
%     (e-4) edge node[msg,midway]{win} (o-4)
%   ;
% \end{scope}

% \useasboundingbox +(0,.3cm);
% \scoped[exclude from bounding box]
% \path (s2.north) -- node[initial,empty block,name=s0,above=1em]{} (s3.north);

\node[above=10ex of s2,fill] (s0) {};

\path[trans]
  (s0) edge node[midway,yshift=1.2ex]{$1$} (s1.north)
  (s0) edge node[midway,xshift=-1ex,yshift=.4ex]{$2$} (s2.north)
  (s0) edge node[midway,yshift=1.2ex]{$3$} (s3.north)
  (s1.south) edge node[midway,yshift=1ex]{$a$} (s4.north)
  (s1.south) edge node[midway,yshift=1.2ex]{$b$} (s5.north)
  (s3.south) edge node[midway,yshift=1ex]{$c$} (s6.north)
  (s3.south) to node[midway,yshift=1.2ex]{$d$} (s7.north)
  %(s5) edge (s4)
  %(s2) edge (s5)
  %(s5) edge (s6)
  ;
\iffalse
  \draw[trans]
  (s6.south) -- node {} ++(0,-0.15) |- ++(2.40,0) |- 
  node[yshift=1.5ex,xshift=-7ex]{\scriptsize$\initUpdLoop{z_1,z_2}{b_1,b_2}$} (s2.east)
  ;
\fi
\end{tikzpicture}

%%% Local Variables:
%%% mode: latex
%%% TeX-master: "main"
%%% End:
	}
	\caption{Task scheduler with task delegation.}
        \label{fig:task-scheduling}
\end{figure}
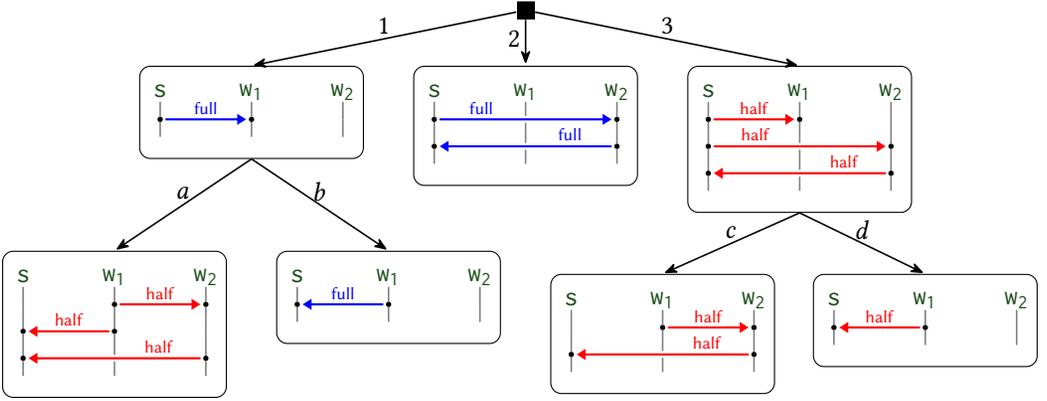

We use a task scheduling protocol to motivate the subtle differences across network architectures. The global specification of the protocol is shown in \cref{fig:task-scheduling}, using high-level message sequence chart (HMSC) notation. The protocol involves three participants: a scheduler $\server$ and two workers $\workerOne$ and $\workerTwo$. Initially, $\server$ chooses to schedule the entire task to either only $\workerOne$ or only $\workerTwo$, or it decides to split the workload between the two workers. The first two cases are depicted in branches $1$ and $2$, where $\server$ sends a $\fullMsg$ message to the respective worker. The third case is branch $3$ where $\server$ sends a $\halfMsg$ message to both workers. Worker $\workerTwo$ always completes its assigned task immediately and sends the result back to $\server$ by echoing the message it has received from $\server$. However, whenever $\workerOne$ is assigned a task, it has the option to behave like $\workerTwo$ (branches $b$ and $d$) or delegate some or all of its work to $\workerTwo$ (branches $a$ and $c$). The protocol operates in a loop, but we omit the back-edges to the initial state in \cref{fig:task-scheduling} for readability.

\paragraph{Implementability modulo network architectures}
Implementability asks whether there exist local implementations for the three participants that behave according to the global protocol specification when executed concurrently on an asynchronous network architecture. In particular, implementations should never deadlock and all participants should behave consistently according to each locally chosen branch, executing send and receive actions exactly in the prescribed order. The latter property is known as \emph{protocol fidelity}. The network architecture is left a parameter of the problem statement. 
\paragraph{Determining implementability}
A local implementation can only gain information about the global protocol state by making branching decisions and by receiving messages. Non-implementability arises when a participant's local information, comprising the decisions and observations it has made so far, is insufficient for determining its next action. 

Let us start by analyzing the implementability question for the task scheduling protocol, assuming a standard peer-to-peer FIFO network architecture (also referred to as \emph{peer-to-peer box} semantics in this paper). Throughout the paper, we use $\snd{\procA}{\procB}{\val}$ to denote a send event where participant $\procA$ sends $\val$ to $\procB$. Likewise, we use $\rcv{\procA}{\procB}{\val}$ to denote an event where $\procB$ receives message $\val$ that was previously sent from $\procA$.

Under a peer-to-peer box network, the protocol is not implementable: $\workerTwo$ cannot distinguish between a run that follows branches $1$ and $a$ and a run that follows branches $3$ and $c$. In both cases, $\workerOne$ may find itself in the same local state $q$ where only the $\halfMsg$ message from $\workerOne$ is available in its associated channel buffer (i.e., in the $3c$ run, the $\halfMsg$ message from $\server$ to $\workerOne$ may be delayed). If the protocol is following the $1a$ run, $\workerTwo$'s next action should be to send a reply to $\server$. However, in the $3c$ run it should first wait for the arrival of $\halfMsg$ from $\server$. If $\workerTwo$ were to wait in state $q$, this would lead to a deadlock in the $1a$ run and if it were to send to $\server$, it would violate protocol fidelity in the $3c$ run.

Perhaps surprisingly, replacing the peer-to-peer box network by another asynchronous network architecture does not resolve this problem. The reason for non-implementability solely depends on the asynchronous nature of communication and the fact that the two send events $\snd{\server}{\workerTwo}{\halfMsg}$ and $\snd{\workerOne}{\workerTwo}{\halfMsg}$ in the $3a$ run do not causally depend on each other. They can therefore happen concurrently, causing the two messages to arrive at $\workerTwo$ in any order. Thus, the protocol is non-implementable for any asynchronous network architecture.

However, in general, implementability depends on the specific network architecture. For example, consider a possible repair of the global specification that replaces the message value $\halfMsg$ of the send from $\workerOne$ to $\workerTwo$ on branch $1a$ with $\delegMsg$. Now, $\workerTwo$ can tell the two branches apart: it can wait until either $\halfMsg$ is available in its buffer from $\server$, indicating that the protocol follows branch $3$, or $\delegMsg$ is available in the buffer from $\workerTwo$, indicating that the other participants chose to follow branch $1a$. Since the two cases are exclusive, $\workerTwo$ can make its decision as soon as it observes one of the two. This change renders the protocol implementable under peer-to-peer box semantics.

On the other hand, this repair does not help for the mailbox network architecture, in which all messages sent to the same recipient are collected in one FIFO buffer. The issue with mailbox is that in the $3c$ branch, the network may still asynchronously reorder the two messages sent to $\workerTwo$ by delaying the message from $\server$. Since messages are buffered in FIFO order of arrival, this would force $\workerTwo$ to first receive the message from $\workerOne$ before being able to retrieve the one from $\server$ in the buffer. The resulting sequence of events would violate protocol fidelity. Note that this time, the ensuing violation of protocol fidelity has nothing to do with incomplete information by any of the protocol participants about what branch the protocol is following. Instead, it is solely due to the ability of the network architecture to reorder independent events in executions of individual protocol runs. A possible repair of the protocol for this architecture is to introduce a causal dependency between $\snd{\server}{\workerTwo}{\halfMsg}$ and $\snd{\workerOne}{\workerTwo}{\halfMsg}$, e.g., by inserting an additional message exchange $\snd{\server}{\workerOne}{\doneMsg}$ between the two exchanges, forcing $\workerOne$ to wait until $\server$ has sent its message to $\workerTwo$.

Finally, if we swap the network architecture from mailbox to \emph{mailbag}, where the single FIFO queue per recipient is replaced by an unordered multiset, the protocol becomes implementable again even without the proposed second repair for mailbox.

%%% Local Variables:
%%% mode: latex
%%% TeX-master: "main"
%%% End:

\section{Preliminaries}
\label{sec:prelim}
\paragraph{Words}
Let $\AlphAsync$ be an alphabet.
$\AlphAsync^*$ denotes the set of finite words over $\AlphAsync$, $\AlphAsync^\omega$ the set of infinite words, and $\AlphAsync^\infty$ their union $\AlphAsync^* \cup \AlphAsync^\omega$.
A word $u \in \AlphAsync^*$ is a \emph{prefix} of word $v \in \AlphAsync^\infty$, denoted $u \leq v$, if there exists $w \in \AlphAsync^\infty$ with $u \cdot w = v$;
we denote all prefixes of $u$ with $\pref(u)$. 
We sometimes omit the concatenation symbol $\cdot$, instead writing $uw = v$. 
Given a word $w = w_0 \ldots w_n$, we use $w[i]$ to denote the $i$-th symbol $w_i \in \AlphAsync$, and $w[0..i]$ to denote the subword between and including $w_0$ and $w_i$, i.e.\ $w_0 \ldots w_i$.

\paragraph{Message Alphabets}
Let $\Procs$ be a (possibly infinite) set of participants and $\MsgVals$ be a (possibly infinite) data domain. 
We define the set of \emph{synchronous events} $\AlphSyncSubscript \is \set{ \msgFromTo{\procA}{\procB}{\val} \mid \procA,\procB ∈ \Procs \text{ and } \val ∈ \MsgVals}$ where $\msgFromTo{\procA}{\procB}{\val}$ denotes a message exchange of $\val$ from sender $\procA$ to receiver $\procB$.
For a participant $\procA\in \Procs$, we define the alphabet $\AlphSync_\procA = \set{\msgFromTo{\procA}{\procB}{\val} \mid \procB \in \Procs,\; \val \in \MsgVals} \cup \set{\msgFromTo{\procB}{\procA}{\val} \mid \procB \in \Procs,\; \val \in \MsgVals}$, and a homomorphism~$\wproj_{\AlphSync_\procA}$, where $x \wproj_{\AlphSync_\procA} = x$ if $x \in \AlphSync_\procA$ and $\emptystring$ otherwise. 
For a participant $\procA \in\Procs$, we define the alphabet
$\AlphAsync_{\procA,!} = \set{\snd{\procA}{\procB}{\val} \mid \procB \in \Procs,\; \val \in \MsgVals }$ of \emph{send} events and the alphabet
$\AlphAsync_{\procA,?} = \set{\rcv{\procB}{\procA}{\val} \mid \procB \in \Procs,\; \val \in \MsgVals }$ of \emph{receive} events.
The event $\snd{\procA}{\procB}{\val}$ denotes participant $\procA$ sending a message $\val$ to $\procB$,
and $\rcv{\procB}{\procA}{\val}$ denotes participant $\procA$ receiving a message $\val$ from $\procB$.
We write $\AlphAsync_{\procA} = \AlphAsync_{\procA,!} \union \AlphAsync_{\procA,?}$,
$\AlphAsync_! = \Union_{\procA \in \Procs} \AlphAsync_{\procA,!}$, and
$\AlphAsync_? = \Union_{\procA \in \Procs} \AlphAsync_{\procA,?}$.
Finally, the set of \emph{asynchronous} events is $\AlphAsyncSubscript = \AlphAsync_! \cup \AlphAsync_?$.

\paragraph{Projections.}
We say that $\procA$ is \emph{active} in $x \in \AlphAsyncSubscript$ if $x \in \AlphAsync_{\procA}$.
For each participant $\procA\in \Procs$, we define a homomorphism~$\wproj_{\AlphAsync_\procA}$, where $x \wproj_{\AlphAsync_\procA} = x$ if $x \in \AlphAsync_\procA$ and $\emptystring$ otherwise.
We define a class of projections based on pattern-matching of alphabet symbols, denoted $\wproj_{\hole}$. The result of the projection is determined by the unspecified parts of the pattern. 
For example, $\wproj_{\snd{\procA}{\hole}{\hole}}$ projects the symbol $\snd{\procA}{\procB}{\val}$ onto $(\procB,\val)$, and non-send symbols and send-symbols that do not have $\procA$ as the sender onto $\emptystring$. 
The function $\wproj_{\rcv{\procA}{\procB}{\hole}}$ projects receive events of $\procA$ from $\procB$ of any message value onto the message value, and all other events to $\emptystring$. 

We adopt labeled transition systems over the synchronous alphabet $\AlphSyncSubscript$ as our starting point for specifying global protocols. 

\paragraph{Labeled Transition Systems}
A \emph{labeled transition system} (LTS) is a tuple
$\mathcal{S} = (S, \AlphSync, T, s_0, F)$
where
$S$ is a set of states, $\AlphSync$ is a set of labels,
$T$ is a set of transitions from $S \times \AlphSync \times S$,
$F \subseteq S$ is a set of final states,
and $s_0 \in S$ is the initial state. 
We use $p \xrightarrow{\alpha} q$ to denote the transition $(p, \alpha, q) \in T$.
Runs and traces of an LTS are defined in the expected way.
A run is $\emph{maximal}$ if it is either finite and ends in a final state, or is infinite.
The language of an LTS $\mathcal{S}$, denoted $\lang(\mathcal{S})$, is defined as the set of maximal traces.
A state $s \in S$ is a \emph{deadlock} if it is not final and has no outgoing transitions.
An LTS is \emph{deadlock-free} if no reachable state is a deadlock.

In \cite{Li25OopslaOfficial}, LTS over $\AlphSyncSubscript$ are constrained with three additional conditions, to yield a fragment called \emph{global communicating labeled transition systems}, hereafter GCLTS. 
The three GCLTS assumptions are sink finality, sender-driven choice, and deadlock freedom. 
Sink finality is a purely syntactic condition enforcing that final states have no outgoing transitions. 
Sender-driven choice states that from any state, all outgoing transitions share a unique sender, and moreover are deterministic. That is, at any branching point in the protocol, there is a unique sender that decides which branch is taken. 
It was shown in \cite{DBLP:journals/tcs/Lohrey03,DBLP:phd/dnb/Stutz24} that forgoing sender-driven choice leads to undecidability of implementability, though there has been recent interest in mixed choice for synchronous protocol fragments~\cite{DBLP:conf/lics/PetersY24}. 
Deadlock freedom requires that all reachable states are either final, or have outgoing transitions.

Our task scheduling example (\cref{fig:task-scheduling}) satisfies all GCLTS assumptions. In particular, to see that it satisfies sender-driven choice, observe that $\server$ chooses among the branches $\{1,2,3\}$ and $\workerOne$ chooses among $\{a,b\}$ and $\{c,d\}$, respectively. 

%%% Local Variables:
%%% mode: latex
%%% TeX-master: "main"
%%% End:

\section{Implementability modulo network architectures} 
\label{sec:implementability}

In this section, we present network-parametric definitions of distributed implementations, global protocol semantics, and finally the implementability problem. 

Our implementation model is based on communicating state machines (CSMs)~\cite{DBLP:journals/jacm/BrandZ83}. CSMs consist of a collection of finite state machines, one for each participant, that communicate via pairwise FIFO channels. 
We generalize CSMs along two key dimensions: the communication topology, and the data structure for message buffers. We also lift the restriction imposed by CSMs that the number of participants and the state spaces of the local state machines must be finite, following~\cite{Li25OopslaOfficial,DBLP:conf/itp/LiW25}.

\begin{definition}[Network architecture] 
	A \emph{network architecture} over a set of participants $\Procs$ and message values $\MsgVals$ is a tuple $\netarch = (\channels, \buffers, \buffermap, \bufinsert, \bufremove, \empbuffer)$ where $\channels$ is a set of channels, $\buffers$ a set of \emph{channel contents} and $\buffermap : \Procs \times \Procs \rightarrow \channels$ a map that associates each sender and receiver with a channel. Intuitively, $\buffermap(\procA,\procB)$ denotes the message buffer to which messages sent from $\procA$ to $\procB$ are deposited. We refer to $\buffermap$ as the \emph{communication topology} of $\netarch$. The set of \emph{channel states} is $\channelstates = \channels \to \buffers$.

        Messages $\msg \in \msgs$ in channel contents are tagged with their sender and receiver, i.e., $\msgs = \Procs \times \Procs \times \MsgVals$. Channel contents are equipped with partial insert and remove operations, $\bufinsert,\bufremove : \msgs \to \buffers \to \buffers$, where $\bufinsert(\msg)(\buffer)$ being undefined indicates that $\buffer$ blocks on inserting $\msg$ and $\bufremove(\msg)(\buffer)$ is only defined when $\msg$ is available for removal in $\buffer$. 
       Finally, $\empbuffer \in \buffers$ is the empty channel contents.
\end{definition} 

\begin{example}
  \label{ex:netarchs}
We define eight concrete network architectures that we will revisit later.
We consider four communication topologies: n-to-n, in which all senders and receivers share the same channel, one-to-n, in which receivers share the same channel to receive from a single sender, n-to-one, in which senders share the same channel to send to a single receiver, and one-to-one, in which each sender and receiver pair have a unique channel. 
We consider two message buffer data structures: ordered FIFO queues, and unordered multisets. 
The aforementioned network architectures often appear under the names of global bus (n-to-n FIFO), message soup (one-to-one multiset), and mailbox (n-to-one FIFO). Message soups are commonly found in leader election protocols such as Paxos and Raft, and one-to-n FIFO is found in work stealing patterns in parallel programming. Finally, the one-to-one or peer-to-peer FIFO is the standard network architecture for CSMs, and widely assumed in the theory and practice of message-passing concurrency. 

The four communication topologies are defined as follows (with our naming conventions given in parenthesis, where ``B'' refers to the name of one of the buffer types below): 
\begin{itemize} 
	\item n-to-n (peer-to-peer B): $\channels = \Procs \times \Procs, \buffermap(\procA,\procB) = (\procA,\procB)$,
	\item one-to-n (mailB): $\channels = \Procs, \buffermap(\procA,\procB) = \procA$,
	\item n-to-one (senderB): $\channels = \Procs, \buffermap(\procA,\procB) = \procB$,
	\item one-to-one (monoB): $\channels = \{0\}, \buffermap(\procA,\procB) = 0$,
\end{itemize} 
and the two buffer types are
FIFO queues (B=box), $\buffers = \msgs^*$, and multisets (B=bag), $\buffers = \msgs \to \Nat$.
In the case of FIFO queues, $\textit{insert}$ corresponds to appending at the end of the queue, $\textit{remove}$ corresponds to removing from the head; in the case of multisets, $\textit{insert}$ is multiset addition, and $\textit{remove}$ multiset deletion. 
The empty buffer contents $b_0$ is $\emptystring$ in the case of FIFO queues and $\emptyset$ in the case of multiset buffers.

We note that the four network architectures with homogeneous bag channels are operationally equivalent and collapse to the monobag case: since messages are all labeled with their sender and receiver, one can ``on demand'' separate the message soup into $\Procs^2$ multisets, or $\Procs$ multisets by sender or receiver whenever messages are sent or received, and thus simulate the other network architectures. This leaves us with the four FIFO network architectures, which we refer to as \ptpbox, \mailbox, \senderbox, and \monobox in the rest of the paper, and the collapsed case for bag channels (\bag).
\end{example}

We next present communicating labeled transition systems parametric in a choice of network architecture $\netarch$. For the following definitions, we fix a network architecture $\netarch = (\channels, \buffers, \buffermap, \bufinsert, \bufremove, \empbuffer)$. We lift the \textit{insert} and \textit{remove} operations to channel states $\chstate \in \channelstates$ of $\netarch$ by defining $\textit{insert}(\chstate, \procA, \procB, \val) = \chstate'$ where $\chstate'(\buffermap(\procA,\procB)) = \bufinsert(\procA,\procB,\val)(\chstate(\buffermap(\procA,\procB)))$ and all other channel contents remain unchanged; 
$\textit{remove}(\chstate, \procA, \procB, \val)$ is defined analogously.

\begin{definition}[Network-parametric CLTS]
$\mathcal{T_\netarch} = \CLTS{T}$ is a \emph{communicating labeled transition system} (CLTS) over $\Procs$, $\MsgVals$ and $\netarch$ if
${T}_\procA$
is a deterministic LTS
over~$\AlphAsync_\procA$ for every $\procA\in\Procs$, denoted by
$(Q_\procA, \AlphAsync_\procA, \delta_\procA, q_{0, \procA}, F_\procA)$.
Let 
$\prod_{\procA \in \Procs} Q_\procA$ 
denote the set of global states. 
A~\emph{configuration} of $\mathcal{T_\netarch}$ is a pair $(\vec{s}, \chstate)$, where $\vec{s}\,$ is a global state and
$\chstate \in \channelstates$ is  a channel state. 
We use $\vec{s}_\procA$ to denote the state of $\procA$ in $\vec{s}$.
The CLTS transition relation, denoted $\rightarrow$, is defined as follows. 
\begin{itemize}
	\item
	$(\vec{s},\chstate) \xrightarrow{\snd{\procA}{\procB}{\val}} (\pvec{s}',\chstate')$ if
	$(\vec{s}_\procA, \snd{\procA}{\procB}{\val}, \pvec{s}'_\procA)\in\delta_\procA$,
	$\vec{s}_\procC = \pvec{s}'_\procC$ for every participant $\procC \neq \procA$,
	$\chstate' =  \textit{insert}(\chstate, \procA, \procB, \val)$. 
	
	\item
	$(\vec{s},\chstate) \xrightarrow{\rcv{\procA}{\procB}{\val}} (\pvec{s}',\chstate')$ if
	$(\vec{s}_\procB, \rcv{\procA}{\procB}{\val}, \pvec{s}'_\procB)\in\delta_\procB$,
	$\vec{s}_\procC = \pvec{s}'_\procC$ for every participant $\procC \neq \procB$,
	$\chstate' = \textit{remove}(\chstate,\procA,\procB,\val)$. 

\end{itemize}
In the initial configuration $(\vec{s}_0, \chstate_0)$, each participant's state in $\vec{s}_0$ is the initial state $q_{0,\procA}$ of $A_\procA$, and $\chstate_0$ maps each channel to the empty buffer $b_0$.
A configuration $(\vec{s}, \chstate)$ is \emph{final} iff $\vec{s}_\procA$ is final for every $\procA$ and $\chstate=\chstate_0$.
Runs and traces are defined in the expected way. 
A run is \emph{maximal} if either it is finite and ends in a final configuration, or it is infinite. 
The language $\lang(\mathcal{T}_\netarch)$ of the CLTS $\mathcal{T}_\netarch$ is defined as the set of maximal traces, and $\text{pref}(\lang(\mathcal{T}_\netarch))$ is defined as the set of prefixes of maximal traces. 
A configuration $(\vec{s}, \chstate)$ is a \emph{deadlock} if it is not final and has no outgoing transitions.
A CLTS is \emph{deadlock-free} if no reachable configuration is a~deadlock.
\end{definition}

The local implementations for participants $\server$, $\workerOne$ and $\workerTwo$ of the repaired task scheduling protocol from \cref{sec:overview} are depicted in \cref{fig:local-implementations-s}, \cref{fig:local-implementations-w1} and \cref{fig:local-implementations-w2} respectively. Note that the active participants are omitted from transition labels for clarity. 

\begin{figure}[t]
	\begin{minipage}[b]{.55\textwidth}
		\centering 
	\begin{tikzpicture}[sem] % node distance=1cm and 2cm,>=stealth', line width=0.25mm]
	% Leftmost states 
	\node[state, initial](q0){};

	% Upper branch
	\node[state, above right=1.2cm and 1.5cm of q0](q1){};
	\node[state, above right=0.3cm and 1.5cm of q1](q2){};
	\node[state, right=1.5cm of q2](q2a){};
	\node[state, below right=0.3cm and 1.5cm of q1](q2b){};
	
	% Middle branch
	\node[state, right=1cm and 1.5cm of q0](q3){};
	\node[state, right=1.5cm of q3](q4){};
	
	% Lower branch
	\node[state, below right=1.2cm and 1.5cm of q0](q5){};
	\node[state, right=1.5cm of q5](q6){};
	\node[state, right=1.5cm of q6](q7){};
	\node[state, right=1.5cm of q7](q8){};

	% First branch
	\path(q0) edge node[sloped, pos=0.5, above] 
	{$\ssnd{\workerOne}{\fullMsg}$} (q1);
	\path(q1) edge node[sloped, pos=0.5, above] 
	{$\srcv{\workerOne}{\halfMsg}$} (q2);
	\path(q1) edge node[sloped, pos=0.5, below] 
	{$\srcv{\workerOne}{\fullMsg}$} (q2b);
	\path(q2) edge node[sloped, pos=0.5, above] 
	{$\srcv{\workerTwo}{\halfMsg}$} (q2a);
	
	% Second branch 
	\path(q0) edge node[sloped, pos=0.5, below] 
	{$\ssnd{\workerOne}{\halfMsg}$} (q5);

	% Third branch 
	\path(q0) edge node[sloped, pos=0.5, above] 
	{$\ssnd{\workerTwo}{\fullMsg}$} (q3);
	\path(q3) edge node[sloped, pos=0.5, above] 
	{$\srcv{\workerTwo}{\fullMsg}$} (q4);
	\path(q5) edge node[sloped, pos=0.5, below] 
	{$\ssnd{\workerTwo}{\halfMsg}$} (q6);
	\path(q6) edge node[sloped, pos=0.5, below] 
	{$\srcv{\workerTwo}{\halfMsg}$} (q7);
	\path(q7) edge node[sloped, pos=0.5, above] 
	{$\srcv{\workerOne}{\halfMsg}$} (q8);
	\path(q7) edge node[sloped, pos=0.5, below] 
	{$\srcv{\workerTwo}{\halfMsg}$} (q8);

\end{tikzpicture}
	\caption{Local implementation for $\server$
	\label{fig:local-implementations-s}}
	\end{minipage}%
	\hfill
        \begin{minipage}[b]{.4\textwidth}
        %\begin{minipage}[b]{.45\textwidth}
		\centering 
		\begin{tikzpicture}[sem]
	% States
	\node[state, initial] (q0) {};
	\node[state] (q1) [above right=0.5cm and 1.5cm of q0] {};
	\node[state] (q2) [right=1.5cm of q1] {};
	\node[state] (q3) [below right=0.5cm and 1.5cm of q0] {};
	\node[state] (q4) [right=1.5cm of q3] {};
	
	% Transitions
	\path (q0) edge node[sloped, pos=0.5, above] {$\srcv{\server}{\fullMsg}$} (q1);
	\path (q1) edge node[pos=0.5, above] {$\ssnd{\server}{\fullMsg}$} (q2);
	\path (q0) edge node[sloped, pos=0.5, above] {$\srcv{\workerOne}{\halfMsg}$} (q3);
	\path (q0) edge node[sloped, pos=0.5, below] {$\srcv{\workerOne}{\delegMsg}$} (q3);
	\path (q3) edge node[sloped, pos=0.5, below] {$\ssnd{\server}{\halfMsg}$} (q4);
\end{tikzpicture}
		\caption{Local implementation for $\workerTwo$
		\label{fig:local-implementations-w2}}
	%\end{minipage}
	%\begin{minipage}[b]{.45\textwidth}
		\centering 
		\begin{tikzpicture}[sem] % node distance=1cm and 2cm,>=stealth', line width=0.25mm]
	\node[state, initial](q0){};
	% Upper branch
	\node[state, above right=0.3cm and 1.5cm of q0](q1){};
	\node[state, above right = 0.05cm and 1.5cm of q1](q2){};
	\node[state, below right=0.05cm and 1.5cm of q1](q3){};
	\node[state, right=1.5cm of q2](q4){};

	% Upper branch
	\path(q0) edge node[sloped, pos=0.5, above] 
	{$\srcv{\server}{\fullMsg}$} (q1);
	\path(q1) edge node[sloped, pos=0.5, above] 
	{$\ssnd{\workerTwo}{\delegMsg}$} (q2);
	\path(q2) edge node[sloped, pos=0.5, above] 
	{$\ssnd{\server}{\halfMsg}$} (q4);
	\path(q1) edge node[sloped, pos=0.5, below] 
	{$\ssnd{\server}{\fullMsg}$} (q3);
	
	% Lower branch
	\node[state, below right=0.3cm and 1.5cm of q0](q6){};
	\node[state, right=1.5cm of q6](q7){};
	
	% Lower branch
	\path(q0) edge node[sloped, pos=0.5, below] 
	{$\srcv{\server}{\halfMsg}$} (q6);
	\path(q6) edge node[sloped, pos=0.5, below] 
	{$\ssnd{\server}{\halfMsg}$} (q7);
	\path(q6) edge node[sloped, pos=0.5, above] 
	{$\ssnd{\workerTwo}{\halfMsg}$} (q7);
	
\end{tikzpicture}
                \vspace{-2em}
		\caption{Local implementation for $\workerOne$
		\label{fig:local-implementations-w1}}
	\end{minipage} 
\end{figure}

\paragraph{Channel compliance.}
Before we define the semantics of global protocols, it is useful to disentangle the behavior of a given network architecture $\netarch$ from that of any particular CLTS $\mathcal{T}_\netarch$. To this end, we introduce the semantic notion of channel compliance that describes all words $w$ that are consistent with $\netarch$. The channel-compliant words are those words that are traces of the universal CLTS, which lifts all constraints on traces imposed by participants' local states.
Formally, let $\mathcal{U}_\netarch$ be the \emph{universal CLTS} for $\netarch$, with $\mathcal{U}_\netarch = \CLTS{U}$ such that $\lang(U_\procA) = \AlphAsync_\procA^*$ for every $\procA \in \Procs$. We say that $w \in \AlphAsyncSubscript^*$ is $\netarch$-channel-compliant if $w$ is a trace of $\mathcal{U}_\netarch$. If $\netarch$ is understood, we just say $w$ is channel-compliant. Moreover, if in fact $w \in \lang(\mathcal{U}_\netarch)$ holds, then we call $w$ \emph{channel-matched}. Intuitively, if one can execute a channel-matched word in any given CLTS for $\netarch$, then in the reached configuration all buffers will be empty.
We use $\lang(\netarch) \subseteq \AlphAsyncSubscript^*$ to denote all channel-compliant words. 

\paragraph{Equality under local projection}
We say that $w_1$ and $w_2$ are equal under local projection, denoted $w_1 \equiv_\Procs w_2$, if for all $\procA$, $w_1 \wproj_{\AlphAsync_\procA} = w_2 \wproj_{\AlphAsync_\procA}$. 
We use $[w]_{\equiv_\Procs}$ to denote the equivalence class under local projection with representative $w$. We lift this to sets $W \subseteq \AlphAsyncSubscript^\infty$, by defining $[W]_{\equiv_\Procs} = \bigcup_{w \in W} [w]_{\equiv_\Procs}$.

\paragraph{Global protocol semantics}
We define the asynchronous semantics of a global protocol $\mathcal{S}$, denoted $\interswaplang_{\mathbb{A}}(\mathcal{S}) \subseteq \AlphAsyncSubscript^\infty$, modulo a choice of network architecture $\netarch$. 
Recall that $\mathcal{S}$ is an LTS over the alphabet $\AlphSyncSubscript$.
The starting point for the semantics $\interswaplang_{\mathbb{A}}(\mathcal{S})$ is the synchronous language $\lang(\mathcal{S})$. 
From $\lang(\mathcal{S})$ we can obtain a set of 1-synchronous asynchronous words through $\SyncToAsync$, which simply splits each atomic send and receive event into its two counterparts, denoted $\SyncToAsync(\lang(\mathcal{S}))$. 
We want to include all asynchronous words that are equal to these 1-synchronous words under local projection and the given network architecture $\netarch$.

We handle the finite and infinite words separately to define the global protocol semantics as the union of its finite and infinite semantics:
\[
\interswaplang_{\mathbb{A}}(\mathcal{S}) = \interswaplang_\netarch^{\fin}(\mathcal{S}) \cup \interswaplang^{\inf}_\netarch(\mathcal{S})\enspace
\]
Following the above recipe and restricting $\SyncToAsync(\lang(\mathcal{S}))$ to finite words yields the finite semantics: 
\[
	\interswaplang^{\fin}_{\netarch}(\mathcal{S}) = [\AlphAsyncSubscript^* \cap \SyncToAsync(\mathcal{L}(\mathcal{S}))]_{\equiv \Procs} \inters
	\lang(\netarch)\enspace. 
\]
The infinite semantics are those words whose prefixes are extensible to some word in $\lang(\mathcal{S})$ modulo equality under local projection and the network semantics:
\[
\interswaplang^{\inf}_{\netarch}(\mathcal{S}) = 
\set{w \in \AlphAsyncSubscript^\infty~\mid~\forall u \leq w.~u \in \pref([\SyncToAsync(\mathcal{L}(\mathcal{S}))]_{\equiv \Procs} \inters
	\lang(\netarch))}\enspace. 
\] 
For disambiguation, we refer to $\lang(\mathcal{S}) \subseteq \AlphSyncSubscript^\omega$ as the \emph{LTS semantics} of $\mathcal{S}$, and refer to $\interswaplang_{\mathbb{A}}(\mathcal{S}) \subseteq \AlphAsyncSubscript^\omega$ as the \emph{protocol semantics} of $\mathcal{S}$. 

\begin{example}
To illustrate our protocol semantics, consider the following protocol, which contains both finite and infinite words in its semantics: $\msgFromTo{\procA}{\procB}{\val}^\infty + (\msgFromTo{\procA}{\procB}{\val})^* \cdot \msgFromTo{\procC}{\procB}{\val}$.
The synchronous runs of the protocol are either of the form $\msgFromTo{\procA}{\procB}{\val}^\infty$, or of the form $(\msgFromTo{\procA}{\procB}{\val})^n \cdot \msgFromTo{\procC}{\procB}{\val}$. 
The $\SyncToAsync$ runs are subsequently of the form $(\snd{\procA}{\procB}{\val} \cdot \rcv{\procA}{\procB}{\val})^\infty$ or 
$(\snd{\procA}{\procB}{\val} \cdot \rcv{\procA}{\procB}{\val})^n \cdot \snd{\procC}{\procB}{\val} \cdot \rcv{\procC}{\procB}{\val}$. Because our infinite word semantics do not impose any fairness assumptions, the unfairly scheduled word $\snd{\procA}{\procB}{\val}^\infty$ is part of the protocol's infinite semantics. 
The word $\snd{\procC}{\procB}{\val} \cdot \snd{\procA}{\procB}{\val} \cdot \rcv{\procA}{\procB}{\val} \cdot \rcv{\procC}{\procB}{\val}$ is part of the protocol's finite semantics under \ptpbox, where the network reorders the send events from $\procA$ and $\procC$, but $\procB$ receives in the specified protocol order, first from $\procA$ and then from $\procC$. 
\end{example}

We are now ready to define the network-parametric implementability problem: 
\begin{definition}[Network-parametric Protocol Implementability]
	A protocol $\mathcal{S}$ is \emph{implementable} under network architecture $\netarch$ if there exists a CLTS $\mathcal{T}_{\netarch} = \CLTS{T}$ such that the following two properties hold:
	\begin{inparaenum}[(i)]
		\item \label{def:lts-implementability-protocol-fidelity}
		\emph{protocol fidelity}: $\lang(\CLTS{T}) = \interswaplang_{\mathbb{A}}(\mathcal{S})$, and
		\item \label{def:lts-implementability-deadlock-freedom}
		\emph{deadlock freedom}: $\CLTS{T}$ is deadlock-free.
	\end{inparaenum}
	We say that $\CLTS{T}$ implements $\mathcal{S}$ under $\netarch$.
\end{definition}

%%% Local Variables:
%%% mode: latex
%%% TeX-master: "main"
%%% End:

\section{Characterization of generalized implementability} 
\label{sec:characterization}

In this section, we present our sound and complete network-parametric implementability characterization. 
We take as a starting point a recently proposed precise characterization for \ptpbox implementability~\cite{Li25OopslaOfficial}. 
We illuminate key abstract assumptions about the protocol semantics and implementation model made by the characterization in \cite{Li25OopslaOfficial} that enable its soundness and completeness proofs. 
In a process analogous to computing weakest pre-conditions, we distill, and in some cases weaken these abstract assumptions, in tandem with developing our network-parametric characterization. 
Ultimately, we obtain a network-parametric set of conditions, that we call Generalized Coherence Conditions, along with a set of abstract assumptions that, when satisfied, render our characterization sound and complete with respect to implementability. 

The implementability characterization of \cite{Li25OopslaOfficial} takes the form of three Coherence Conditions (\Characterization) that a global protocol must satisfy. 
The Coherence Conditions are 2-hyperproperties that scrutinize pairs of global protocol states from which a participant can perform different actions, but whose distinction may not be locally observable to the participant. 
\Characterization describes the kinds of local actions that are safe to perform in this state of unawareness. 
\emph{Send Coherence} says that if a participant has the option to perform a send action from one state, it must have the option to perform the same send action from any indistinguishable state. 
\emph{Receive Coherence} says that if a participant has the option to perform a receive action from one state, then this same receive action could not possibly be performed from any other indistinguishable state. 
\emph{No Mixed Choice} says that a participant cannot equivocate between performing a send and receive action. 

Formally, a global protocol state $s$ is reachable for $\procA$ on $u \in \AlphSync_{\procA}^*$
when there exists $w \in \AlphSyncSubscript^*$ such that $s$ is reachable with trace $w$ and $w \wproj_{\AlphSync_{\procA}} = u$. 
If two global protocol states $s_1$ and $s_2$ are both reachable for $\procA$ on the same $u \in \AlphSync_{\procA}^*$, we call them \emph{simultaneously reachable}. 
Simultaneous reachability captures pairs of states that are locally indistinguishable to a participant. 

\citet{Li25OopslaOfficial} show that \Characterization is sound by invoking a \emph{canonical implementation}, which serves as the witness to implementability. 

\begin{definition}[Canonical implementations~\cite{Li25OopslaOfficial}]
	\label{def:local-language-property}
	A CLTS $\CLTS{T}$ is a \emph{canonical implementation} for a protocol $\mathcal{S} = (S, \AlphSyncSubscript, T, s_0, F)$
	if for every $\procA \in \Procs$, $T_\procA$ satisfies: \\
	\begin{inparaenum}[(i)]
		\item $\forall w \in \Alphabet_{\procA}^*.~w \in \lang(T_\procA) \Leftrightarrow w \in \lang(\mathcal{S}) \wproj_{\Alphabet_\procA}$, and
		\item $\pref(\lang(T_\procA)) = \pref(\lang(\mathcal{S}) \wproj_{\Alphabet_\procA})$.
	\end{inparaenum}
\end{definition}

As observed in \cite{DBLP:conf/itp/LiW25}, canonicity can be defined directly in terms of a global protocol's LTS semantics, since local projections are oblivious to asynchronous reorderings. As discussed in \cite{Li25OopslaOfficial} and formalized in \cite{DBLP:conf/itp/LiW25}, canonical implementations can always be constructed using a generalized subset construction, and synthesis is separate from considerations of network architecture for implementability. We discuss synthesis further in \S\ref{sec:discussion}. 

The key technical argument for soundness lies in showing that the canonical implementation's language is a subset of global protocol semantics, and is deadlock-free. 
This requires showing that every canonical implementation trace can be associated with a run in the global protocol that each participant has partially completed the prescribed actions of. 
This in turn is shown by induction on canonical implementation traces, appealing to \Characterization to argue that the extension by either a send or receive event retains the existence of a global run that can be associated with the resulting trace. 

Completeness of \Characterization is established in \cite{Li25OopslaOfficial} via modus tollens: from the negation of each Coherence Condition a trace is constructed that is compliant with no protocol run, yet must be admitted by any candidate implementation of the protocol. This suffices for an implementability violation, because either the trace leads to a deadlock, or to a maximal word not in the global protocol semantics. 

\paragraph{Revisiting Send Coherence and No Mixed Choice}
These two Coherence Conditions remain sound and complete for network-parametric implementability in their original form. However, their soundness and completeness proofs still need to be generalized to the network-parametric case.

Towards generalizing the completeness proofs of Send Coherence and No Mixed Choice, we observe that the witness constructed by the completeness proof in \cite{Li25OopslaOfficial} can be adapted to one that is the prefix of a 1-synchronous trace. Concretely, the witnesses for both conditions assume the form of $\SyncToAsync(\alpha) \cdot x$, where $\alpha$ is a synchronous word and $x$ is a send event. 
The proof then shows that $\SyncToAsync(\alpha) \cdot x$ is channel-compliant, and together with the fact that it must be executable by any candidate implementation, thus consistutes a counterexample to implementability. 
The generalization thus first requires the assumption that send transitions are always enabled in the network architecture under consideration. We formalize this assumption as follows: 
\begin{equation}
	\label{asm:cc-send}
	\forall w \in \AlphAsyncSubscript^*, x \in \AlphAsync_!.~w \text{ is channel-compliant } \implies wx \text{ is channel-compliant}.
	\tag{\ref*{ca:send-enabled}}
\end{equation}
To complete the network-parametric proof, we show that $\SyncToAsync(\alpha) \cdot x$ is channel-compliant. This fact follows from~\ref{asm:cc-send}, provided $\SyncToAsync(\alpha)$ is channel-compliant. This yields our second assumption:
\begin{equation}
  \label{asm:cc-sync}
  \forall \alpha \in \AlphSyncSubscript^*.~\SyncToAsync(\alpha) \text{ is channel-compliant}.
  \tag{\ref*{ca:sync}}
\end{equation}
Intuitively, a reasonable network should allow sends and their receives to be executed consecutively.

The soundness of Send Coherence seeks to establish that any send extension by the canonical implementation remains a protocol prefix. It turns out that this fact depends on both No Mixed Choice and Receive Coherence. Next, we discuss how to generalize Receive Coherence. 

\paragraph{Revisiting Receive Coherence}
In contrast to send events, receive events are conditioned on transitions as well as the availability of the message in question to be received.  
Message availability in turn highly depends on the network architecture, and requires considering possible asynchronous reorderings, as evidenced by participant $\workerTwo$'s plight in the task delegation protocol from \S\ref{sec:overview}. 

Consider the protocols $\mathcal{S}_a$ and $\mathcal{S}_b$, depicted in \cref{fig:bag-neq-p2p} and \cref{fig:p2p-neq-senderbox}, both from the perspective of the receiver $\procB$.
In $\mathcal{S}_a$ under a monobag network architecture, $\procB$'s bag can contain any number of $\val$ messages from $\procA$, in addition to a $\bot$ message. Participant $\procB$ can never know when it is safe to receive the $\bot$ message, and thus the protocol is non-implementable. 
If we replace the monobag network with a peer-to-peer box network, the final $\bot$ message is ordered after all $\val$ messages have been sent, and thus participant $\procB$ can always receive the message at the head of its FIFO queue from $\procA$. 
In $\mathcal{S}_b$ under a peer-to-peer box network, when message $\val$ from $\procB$ is available to receive, the same message $\val$ could also be available from $\procA$ simultaneously, leading $\procB$ to a protocol violation. 
If we replace the peer-to-peer box network with a senderbox network, this ambiguity again goes away: $\procC$'s message to $\procD$ effectively blocks the message to $\procB$ from being available to receive, thus a unique message is available for $\procB$ to receive regardless of the branch selected by $\procA$. 

\begin{figure}[t]
	\centering
	\begin{minipage}[b]{0.37\textwidth}\vspace{0pt}%
		\centering 
		\begin{tikzpicture}[sem]
	% States
	\node[state, initial] (q0) {};
	\node[state, accepting, right=2cm of q0] (q1) {};
	
	% Transitions
	\path (q0) edge [loop above] node {$\msgFromTo{\procA}{\procB}{\val}$} (q0);
	\path (q0) edge node[sloped, pos=0.5, above] {$\msgFromTo{\procA}{\procB}{\bot}$} (q1);
\end{tikzpicture} 
		\caption{A protocol $\mathcal{S}_a$ that is not implementable on a bag network but implementable on a peer-to-peer box network.\label{fig:bag-neq-p2p} }
	\end{minipage}%
	\hfill%
	\begin{minipage}[b]{0.59\textwidth}\vspace{0pt}%
		\centering
		\begin{comment}
	$G_b \is 
	+ 
	\begin{cases}
		\msgFromTo{\procA}{\procC}{\msgO}. \msgFromTo{\procC}{\procB}{\val}. 0  \\
		\msgFromTo{\procA}{\procC}{\msgB}. \msgFromTo{\procA}{\procB}{\val}. \msgFromTo{\procC}{\procD}{\val}. \msgFromTo{\procC}{\procB}{\val}.0
	\end{cases}
	$
\end{comment}

\begin{tikzpicture}[sem] % node distance=1cm and 2cm,>=stealth', line width=0.25mm]
	\node[state,initial](q0){};
	% Upper branch
	\node[state, above right=0.1cm and 1.5cm of q0](q1){};
	\node[state, accepting, right=1.2cm of q1](q2){};

	% Upper branch
	\path(q0) edge node[sloped, pos=0.5, above] {$\msgFromTo{\procA}{\procC}{\msgO}$} (q1);
	\path(q1) edge node[sloped, pos=0.5, above] {$\msgFromTo{\procC}{\procB}{\val}$} (q2);
	
	% Lower branch
	\node[state, below right=0.1cm and 1.5cm of q0](q6){};
	\node[state, right=1.2cm of q6](q7){};
	\node[state, right=1.2cm of q7](q78){};
	\node[state, right=1.2cm of q78](q8){};
	\node[state, accepting, right=1.2cm of q8](q9){};
	
	% Lower branch
	\path(q0) edge node[sloped, pos=0.5, below] {$\msgFromTo{\procA}{\procC}{\msgB}$} (q6);
	\path(q6) edge node[sloped, pos=0.5, below] {$\msgFromTo{\procA}{\procB}{\val}$} (q7);
	\path(q7) edge node[sloped, pos=0.5, below] {$\msgFromTo{\procB}{\procD}{\val}$} (q78);	
	\path(q78) edge node[sloped, pos=0.5, below] {$\msgFromTo{\procC}{\procD}{\val}$} (q8);
	\path(q8) edge node[sloped, pos=0.5, below] {$\msgFromTo{\procC}{\procB}{\val}$} (q9);
	
\end{tikzpicture}
                \vspace*{-1em}
		\caption{A protocol $\mathcal{S}_b$ that is not implementable on a peer-to-peer box network but implementable on a senderbox network.
			\label{fig:p2p-neq-senderbox} }
	\end{minipage}
\end{figure}

To generalize Receive Coherence, we first make the following key observation about all implementations of global protocols. 
If at any point during the implementation's execution, a participant has the choice between receiving two different messages, then one of these choices must constitute a protocol violation.
Because protocol runs totally order the events of all participants, there exists no protocol run that is compliant with both choices of receptions. 
Assume that for each choice of reception, there exists a protocol run that is compliant with the resulting execution trace. 
Due to the sender-driven nature of global protocols, for any two runs of a global protocol, the first event along which they differ must share a common sender. 
The fact that the receiver can choose between two receptions means it is oblivious to the difference between the corresponding compliant runs, and their choice thus forces the sender to commit to one of two branches, a coordination between participants that is impossible since either the send event causally precedes the receive event, or the send event is independent of the receive event.
Indeed, the soundness proof of Receive Coherence in \cite{Li25OopslaOfficial} argues that all receive extensions to the canonical implementation's traces are unique, and the completeness proof of Receive Coherence constructs a counterexample to implementability from the existence of an implementation prefix that admits two different receive extensions.  

It follows from this key observation that irrespective of network architecture, it is sufficient and necessary to guarantee that no two distinct receptions are enabled for a given participant from any configuration. 
We then capture this simpler property in a network-parametric way. 
The notion of $\val$ being receivable by $\procB$ from $\procA$ in some configuration after executing a word $w$ is captured simply as the channel compliance of $w \cdot \rcv{\procA}{\procB}{\val}$. 
We desire two more constraints on $w$: first, the sending of $\val$ by $\procA$ should not depend on any events by $\procB$, and thus we require that $w \wproj_{\AlphAsync_{\procB}} = \emptystring$, and second, a different message must also have been sent to $\procB$ in $w$: either the sender or the message value differ. 

This leads to our Generalized Receive Coherence condition, defined as follows:

\begin{definition}[Generalized Receive Coherence]
	A protocol $\mathcal{S}$ % = (S, \AlphSyncSubscript, T, s_0, F)
	satisfies Generalized Receive Coherence when for every two simultaneously reachable by $\procB$ states $s_1, s_1'$ with transitions 
	$s_1 \xrightarrow{\msgFromToNS{\procA}{\procB}{\val}} s_2$ and $s_1' \xrightarrow{\msgFromToNS{\procC}{\procB}{\val'}} s_2'$, if  
	$\procC \neq \procA$ or $\val' \neq \val$, then
        for any $w \in \pref (\lang(\mathcal{S}_{s_1'}))$ such that
	$w \wproj_{\Alphabet_{\procB}} = \emptystring$ and 
	$\snd{\procC}{\procB}{\val'} \leq w \wproj_{\AlphAsync_\procC}$, 
        the extension $w\cdot\rcv{\procA}{\procB}{\val}$ is not channel-compliant.
\end{definition}

We highlight two differences from Receive Coherence for \ptpbox from \cite{Li25OopslaOfficial}: $w$ is required to be a protocol prefix for the protocol reinitialized at $s_1'$ instead of $s_2'$, and we require an additional conjunct that $\snd{\procC}{\procB}{\val'}$ is the first event for participant $\procC$ in $w$. 
It is easy to see that for \ptpbox channel compliance, the two RHSs are equivalent, because messages from $\procC$ to $\procB$ do not share a channel with messages from a different sender, and thus unmatched send events can always be prepended to traces without affecting channel compliance. 
However, imposing this fact as an assumption would a priori rule out \mailbox and \monobox network architectures. 
Thus, the RHS of Generalized Receive Coherence captures that a send from $\procC$ to $\procB$ exists in $w$, without committing to where in $w$ it resides, and we formulate the following assumption on channel compliance instead, which we call \emph{history insensitivity}: for all $w \in \AlphAsyncSubscript^*$, $\alpha,\beta \in \AlphSyncSubscript^*$, $\procA \neq \procB \in \Procs$, $\val \in \MsgVals$, then
\begin{align*}
\label{asm:cc-history-insensitive}
(\forall \procA \in \Procs.~w \wproj_{\AlphAsync_{\procA}} \leq \SyncToAsync(\alpha\beta) \wproj_{\AlphAsync_{\procA}}) \land w \wproj_{\AlphAsync_{\procB}} \!=\! \SyncToAsync(\alpha)  \wproj_{\AlphAsync_{\procB}} \land w \!\cdot\! \rcv{\procA}{\procB}{\val} \text{ is channel-compliant } \\
\implies
\exists w'.(\forall \procA \in \Procs.~w' \wproj_{\AlphAsync_{\procA}} \leq \SyncToAsync(\beta) \wproj_{\AlphAsync_{\procB}}) \land w' \wproj_{\AlphAsync_{\procB}} \!=\! \emptystring \land w' \!\cdot\! \rcv{\procA}{\procB}{\val} \text{ is channel-compliant } & \tag{\ref*{ca:history-insensitive}}
\end{align*}

Intuitively, \ref{ca:history-insensitive} allows an asynchronous trace that is compliant with a a synchronous trace $\alpha\beta$ to selectively forget events from $\alpha$, resulting in a trace $w'$ that is compliant with $\beta$, such that if $\procB$ completed all events in $\alpha$ in $w$, then all messages that were receivable in $w$ remain receivable in $w'$. 
\ref{asm:cc-history-insensitive} is thus sufficient to establish the soundness of Generalized Receive Coherence: from an execution trace extensible with a receive event, we argue that this receive event must be uniquely determined by all compliant protocol runs of the trace thus far, otherwise we use \ref{asm:cc-history-insensitive} to construct a new trace that contradicts Generalized Receive Coherence. 

\paragraph{Prefix Extensibility.}
The three Coherence Conditions for \ptpbox  alone are no longer sufficient to guarantee implementability, even with Generalized Receive Coherence. It turns out that 
there is a fourth source of non-implementability that never occurs for \ptpbox networks but does occur for other architectures. We capture this source of non-implementability in a new fourth Coherence Condition, coined Prefix Extensibility.

To motivate the new condition, consider the simple straight-line global specification consisting only of the synchronous word $\msgFromTo{\procA_1}{\procB}{\val} \cdot \msgFromTo{\procA_2}{\procB}{\val}$. 
Despite its simplicity, there does not exist a deadlock-free \mailbox (or \monobox) CLTS that implements this specification. 
Any candidate CLTS must exhibit the following deadlocking trace: 
$\snd{\procA_2}{\procB}{\val} \cdot \snd{\procA_1}{\procB}{\val}$.
Because the network can reorder $\procA_2$'s send event before $\procA_1$'s send event, yet the local actions of $\procB$ tell it to receive in the opposite order, the only way for the CLTS to not deadlock is for $T_\procB$ to admit the local trace $\rcv{\procA_2}{\procB}{\val} \cdot \rcv{\procA_1}{\procB}{\val}$.

To rule out such cases of non-implementability, Prefix Extensibility requires that protocol prefixes are closed under per-participant equality.

\begin{definition}[Prefix Extensibility] 
	A protocol $\mathcal{S}$ %= (S, \AlphSyncSubscript, T, s_0, F)  
	satisfies Prefix extensibility when for every protocol trace $\rho \in \text{pref}(\lang(\mathcal{S}))$ and $w \in \AlphAsyncSubscript^*$ such that $w$ is $\netarch$-channel-compliant and agrees with $\rho$, there exists $u \in \AlphAsyncSubscript^*$ such that $wu$ is channel compliant, and $wu \equiv_\Procs \SyncToAsync(\rho)$. 
\end{definition}

Note that the protocol above violates Prefix Extensibility: let $\rho = \msgFromTo{\procA_1}{\procB}{\val} \cdot \msgFromTo{\procA_2}{\procB}{\val}$ and $w = \snd{\procA_2}{\procB}{\val}$, then $w$ agrees with $\rho$ and is \mailbox channel-compliant. However, all extensions of $w$ that are per-participant equal to $\SyncToAsync(\rho)$ are not \mailbox channel-compliant because they deadlock. 

\paragraph{Generalized Coherence Conditions}
Our generalized characterization of implementability is thus defined as the conjunction of four conditions. 
\begin{definition}[Generalized Coherence Conditions]
	A protocol $\mathcal{S}$ % = (S, \AlphSyncSubscript, T, s_0, F)
	satisfies Generalized Coherence Conditions under network architecture $\netarch$ when it satisfies Send Coherence, Generalized Receive Coherence, No Mixed Choice, and Prefix Extensibility. 
\end{definition}

We present the three remaining assumptions we impose on the network architecture that are required for our preciseness proof.
We refer to the collection of six assumptions \ccfact{1} through \ccfact{6} altogether as \emph{channel compliance facts}, and these together constitute the sufficient conditions under which our network-parametric characterization is sound and complete. 
	\label{def:cc-facts} 
	Let $w \in \AlphAsyncSubscript^*$ be channel-compliant. 
	\begin{enumerate}[\ccfact{\arabic*}]
		%\item For all $x \in \AlphAsync_!$, $wx$ is channel-compliant. \label{ca:send-enabled}
		\item[\ref*{ca:send-before-receive}] For all $\procA \neq \procB \in \Procs$ and $\val \in \MsgVals$, $\card {w \wproj_{\snd{\procA}{\procB}{\val}}} \geq \card {w \wproj_{\rcv{\procA}{\procB}{\val}}}$.
		\item[\ref*{ca:sync-matched}] For all $\rho \in \AlphSyncSubscript^*$, if 
		$w \equiv_\Procs \SyncToAsync(\rho)$, then $w$ is channel-matched.
		\item[\ref*{ca:intro-send}] For all $x \in \AlphAsync_!$, $y \in \AlphAsync_?$, if $wy$ is channel-compliant then $wxy$ is channel-compliant. 
	\end{enumerate}
These basic assumptions are used throughout the proofs. For example, \ref*{ca:send-before-receive} states that messages cannot be received before they were sent.

The following theorem states that our Generalized Coherence Conditions are sound and complete with respect to implementability, assuming a network architecture that satisfies our channel compliance facts. 
The proof of \cref{thm:preciseness} is fully mechanized in Rocq, and builds on the Rocq mechanization for \ptpbox networks in \cite{DBLP:conf/itp/LiW25}.

\begin{theorem}[Preciseness of Generalized Coherence Conditions]
  \label{thm:preciseness}
   Let $\netarch$ be a network architecture that satisfies \ccfact{1} through \ccfact{6} and let $\mathcal{S}$ be a global protocol. 
	Then, $\mathcal{S}$ is implementable under $\netarch$ if and only if it satisfies Generalized Coherence Conditions under $\netarch$. 
\end{theorem}

%%% Local Variables:
%%% mode: latex
%%% TeX-master: "main"
%%% End:

\section{Reducing channel compliance facts to buffer axioms}
\label{sec:channel-compliance}

With a network-parametric implementability characterization and preciseness proof in place, what remains is determining whether a given network architecture $\netarch$ is a suitable instance of the framework. 
Based on our technical development thus far, this amounts to showing that $\netarch$-channel compliance satisfies the six channel compliance facts assumed in the theorem. However, proving these facts for every network architecture individually is tedious and repetitive at best, and non-trivial at worst: in particular, history insensitivity (\ref{asm:cc-history-insensitive}) requires reasoning globally across events in a word that are not causally ordered either by participants or by channels. As empirical evidence of the proof effort required to prove the channel compliant facts, the Rocq mechanization of \ccfact{1} through \ccfact{6} for \senderbox comprises some 5000 lines of code!
 
 We thus take a further step in reducing channel compliance facts to a set of simple, operational buffer data structure specifications that only describe the behavior of operations per individual channel in $\netarch$. 
 We call these \emph{buffer axioms}, and the set of buffer axioms constitute our axiomatic model of network architectures to which our implementability characterization applies.

In the following, for a network architecture $\netarch$, we denote by $\ops$ the set of all operations $\bufinsert(\msg),\bufremove(\msg): \buffers \pto \buffers$ for $\msg \in \msgs$. By slight abuse of notation, we identify finite sequences $\opseq =\op_1\dots\op_n \in \ops^*$ with the partial function $\op_n \circ \dots \circ \op_1$ obtained by composing $\op_1,\dots,\op_n$ in reverse order, and the empty sequence $\emptystring$ with the identity function on $\buffers$. In other words, $\opseq$ is identified with the composite operation that captures the cumulative effect of $\opseq$'s constituent operations. We further write $\revapp{f}{x}$ for reverse function application, i.e., $\revapp{x}{f} = f(x)$.

\begin{definition}[Axiomatic network model]
	\label{def:buffer-axioms}
  We denote by $\netarchs$ the set of all network architectures $\netarch$ that satisfy the following properties, for all $\buffer \in \buffers$, $\msg,\msgb \in \msgs$, and $\opseq \in \ops^*$:
  \begin{enumerate}[\baxiom{\arabic*}]
  \item $\revapp{\buffer}{\bufinsert(\msg)}$ is defined. \label{ba:insert-total}
  \item $\revapp{\revapp{\empbuffer}{\bufinsert(\msg)}}{\bufremove(\msg)}$ is defined.\label{ba:emp-insert-remove}
  \item $\revapp{\empbuffer}{\bufremove(\msg)}$ is undefined. \label{ba:emp-remove}
  \item If $\revapp{\buffer}{\bufremove(\msg)}$ is defined, then $\revapp{\revapp{\buffer}{\bufinsert(\msgb)}}{\bufremove(\msg)}$ is defined. \label{ba:intro-insert}
  \item If $\revapp{\revapp{\buffer}{\bufinsert(\msg)}}{\opseq}$ is defined and $\bufremove(\msg)$ does not occur in $\opseq$, then $\revapp{\buffer}{\opseq}$ is defined. \label{ba:elim-insert}
  \item If $\revapp{\buffer}{\bufremove(\msg)}$ is defined, then  $\revapp{\revapp{\buffer}{\bufinsert(\msg)}}{\bufremove(\msg)} = \revapp{\revapp{\buffer}{\bufremove(\msg)}}{\bufinsert(\msg)}$. \label{ba:left-com-remove-eq}
  \item If $\msg \neq \msgb$, then $\revapp{\revapp{\buffer}{\bufinsert(\msg)}}{\bufremove(\msgb)} = \revapp{\revapp{\buffer}{\bufremove(\msgb)}}{\bufinsert(\msg)}$. \label{ba:left-com-remove-neq}
  \item If $\revapp{\buffer}{\bufremove(\msg)}$ is undefined, $\revapp{\revapp{\revapp{\buffer}{\bufinsert(\msg)}}{\opseq}}{\bufremove(\msg)}$ is defined, and $\bufremove(\msg)$ does not occur in $\opseq$, then $\revapp{\revapp{\revapp{\buffer}{\bufinsert(\msg)}}{\opseq}}{\bufremove(\msg)} = \revapp{\buffer}{\opseq}$. \label{ba:cancel-insert-remove}
  \end{enumerate}
\end{definition}

Intuitively, \ref{ba:insert-total} says that $\bufinsert$ is total, \ref{ba:emp-insert-remove} says that one can always insert and immediately remove the same message from an empty channel, \ref{ba:emp-remove} says that one cannot remove from the empty channel, \ref{ba:intro-insert} says that adding inserts does not disable subsequent removes, \ref{ba:elim-insert} says that unmatched inserts can be omitted, \ref{ba:left-com-remove-eq} says that removes left-commute with inserts of the same message, provided the message has already been inserted earlier, \ref{ba:left-com-remove-neq} says that removes left-commute with inserts of different messages, and \ref{ba:cancel-insert-remove} says that matching inserts and removes cancel each other out.

It is now trivial to observe that FIFO queues and multisets satisfy these axioms. Thus all network architectures of \cref{ex:netarchs} are in $\netarchs$. However, some common channel data structures are ruled out. In particular, \ref{ba:insert-total} rules out bounded channels and \ref{ba:cancel-insert-remove} rules out channels that allow message duplication.

The following lemma states that our network model implies the precondition of \cref{thm:preciseness}.

\begin{lemma}
  \label{lem:model-entails-ccfacts}
  All $\netarch \in \netarchs$ satisfy \ccfact{1} through \ccfact{6}.
\end{lemma}

\begin{corollary}
  Let $\netarch \in \netarchs$. Then for all protocols $\mathcal{S}$, $\mathcal{S}$ is implementable under $\netarch$ if and only if it satisfies Generalized Coherence Conditions. 
\end{corollary}

The proof of \cref{lem:model-entails-ccfacts} is in \cref{sec:lem:model-entails-ccfacts-proof}. For example, to show that \ref{asm:cc-sync} holds one proves the stronger property that $\SyncToAsync(\alpha)$ is channel-matched by induction on the length of $\alpha$. In the induction step, one uses \ref{ba:emp-insert-remove} to show that the extended word remains channel-compliant, and then \ref{ba:cancel-insert-remove} to show that it leaves all channels empty, i.e., is channel-matched.

\newcommand\ccarch[0]{\mathcal{C}_\netarch\xspace}

%%% Local Variables:
%%% mode: latex
%%% TeX-master: "main"
%%% End:

\section{Derivatives of Generalized \Characterization}
\label{sec:instantiations}

In this section, we derive a series of results from our sound and complete characterization of network-agnostic implementability. 
First, we introduce symbolic protocols featuring dependent refinements. 
Following the blueprint in~\cite{Li25OopslaOfficial}, we show that Generalized \Characterization similarly admits an encoding as least and greatest fixpoints over the symbolic protocol's transition relation. 
This allows us to define a sound and relatively complete algorithm for deciding implementability of symbolic protocols for the five network architectures considered in \cref{ex:netarchs}. 
As an interlude, we use our encoding to establish the relationship between classes of implementable protocols for each network architecture. 
Finally, we present complexity results for finite fragments under the considered architectures. 

\subsection{Checking Implementability of Symbolic Protocols}
\label{sec:symbolic}

Symbolic protocols~\cite{Li25OopslaOfficial} are defined over a fixed but unspecified first-order background theory of message values (\eg linear integer arithmetic).
We assume standard syntax and semantics of first-order formulas and denote by $\formulas$ the set of first-order formulas with free variables drawn from an infinite set $X$.
We assume that these variables are interpreted over the set of message values $\MsgVals$.

\begin{definition}[Symbolic protocol~\cite{Li25OopslaOfficial}]
	A symbolic protocol is a tuple 
	$\SymProt = (S, R, \Delta, s_{0}, \rho_0, F)$ 
	where
	\begin{itemize}
		\item $S$ is a finite set of control states,
		\item $R$ is a finite set of register variables,
		\item $\Delta \subseteq S \times \Procs \times X \times \Procs \times \formulas \times S$ is a finite set that consists of symbolic transitions of the form $s \xrightarrow{\msgFromToNS{\procA}{\procB}{x \set{\varphi}}} s'$ where the formula $\varphi$ with free variables $R \uplus R' \uplus \set{x}$ expresses a transition constraint that relates the old and new register values ($R$ and $R'$), and the sent value $x$,
		\item $s_0 \in S$ is the initial control state,
		\item $\rho_0 \from R \rightarrow \MsgVals$ is the initial register assignment, and
		\item $F \subseteq S$ is a set of final states.
	\end{itemize}
\end{definition}

In the remainder of the section, we fix a symbolic protocol $\SymProt = (S, R, \Delta, s_{0}, \rho_0, F)$ whose concretization satisfies the GCLTS assumptions.

First, we show how to encode Generalized Receive Coherence (GRC) into \muCLP. 
We denote the four conjuncts on the right-hand side of the implication of GRC by $\text{RHS} \is 
	w \in \pref(\lang(\mathcal{S}_{s_1'})) \land 
	w \wproj_{\Alphabet_{\procB}} = \emptystring 
	\land
	\snd{\procC}{\procB}{\val'} \leq w \wproj_{\AlphAsync_\procC}
	\land
	w \cdot \rcv{\procA}{\procB}{\val} \text{ is channel-compliant}$. 
Let $s_1 \xrightarrow{\msgFromToNS{\procA}{\procB}{\val}} s_2, s_1' \xrightarrow{\msgFromToNS{\procC}{\procB}{\val'}} s_2' \in T$ such that $s_1, s_2$ are simultaneously reachable by $\procB$. 
We first factor GRC into two conditions by doing case analysis on its precondition. 
Let $\textsc{GRC(a)}$ denote $(\procC \neq \procA \!\implies\! \neg~ \exists~w.~\text{RHS}(w))$, and let $\textsc{GRC(b)}$ denote $(\procC = \procA \land \val' \neq \val \implies \neg~ \exists~w.~\text{RHS}(w))$. 
It is clear that $\textsc{GRC(b)}$ is unsatisfiable for network models with FIFO channels. 
In a FIFO channel, reception order follows send order. 
By the semantics of global protocols, per-participant events are totally ordered, and thus $\snd{\procA}{\procB}{\val'}$ by $\procA$ is ordered before $\snd{\procA}{\procB}{\val}$ by $\procA$, which must exist in $w$ because channel compliant words satisfy send-before-receive order. 
The only way in which $\val$ is receivable by $\procB$ in $w$ is if $\procB$ has already received $\val'$, but the conjunct $w \wproj_{\Alphabet_{\procB}} = \emptystring$ says that $\procB$ has no events in $w$. 
In conclusion, $\textsc{GRC(b)}$ need only be checked only for $\netarch = \bag$. 

Next, we show how to encode $\textsc{RHS}(w)$ as a least fixpoint. 
\citet{Li25OopslaOfficial} define a family of predicates $\avail_{\procA, \procB, \blockedset}(x_1, s_2, \boldsymbol{r_2})$ that captures whether $x_1$ may be available as the first message from $\procB$ to $\procA$ in a \ptpbox network, while tracking causal dependencies using $\blockedset$. 
$\blockedset$ tracks the set of participants that are blocked from sending a message because their send action causally depends on $\procB$ first receiving from~$\procC$. 
However, their $\avail$ predicate is specific to the \ptpbox network architecture. 

We present a network-parametric version of $\avail$ below, parametric in a choice of sender $\procA$ and receiver $\procB$ whose message exchange we are searching for, a blocked set $\blockedset$ that tracks causal dependencies between participants, and finally a network architecture $\netarch \in \{\ptpbox,\senderbox,\mailbox,\monobox,\bag\}$. 

\newcommand{\isbag}{\mathit{isbag}}

\begin{definition}[Symbolic Availability]
	\begin{align*}
		\small
		\avail_{\netarch, \procA, \procB, \blockedset}(x_1, s, \boldsymbol{r})
		\is_\mu % \\
		&\phantom{\lor j}(
		\bigvee_{\substack{(s, \,\msgFromToNS{\procC}{\procE}{x \set{\varphi}}, \,s') \in \Delta \\ \procC \in \blockedset %\\ 
				%\buffermap(\procC,\procE) \neq \buffermap(\procA,\procB)
				%\top 
		}}
		\hspace{-2ex}
		\exists x~\boldsymbol{r'}.\, \avail_{\netarch, \procA, \procB, \blockedset \cup \set{\procE}} (x_1, s', \boldsymbol{r'})
		\land 
		\varphi
		\, )  \\
		& \lor(
		\bigvee_{\substack{(s, \,\msgFromToNS{\procC}{\procE}{x \set{\varphi}}, \,s') \in \Delta \\ \procC \notin \blockedset \\ 
				\buffermap(\procC,\procE) \neq \buffermap(\procA,\procB)
				\lor 
				\procE \notin \blockedset \lor \netarch=\bag
		}}
		\hspace{-2ex}
		\exists x~\boldsymbol{r'}.\, \avail_{\netarch, \procA, \procB, \blockedset}(x_1, s', \boldsymbol{r'})
		\land 
		\varphi
		\, ) \\
		&\lor(
		\bigvee_{\substack{(s, \,\msgFromToNS{\procA}{\procB}{x \set{\varphi}}, \,s') \in \Delta \\ \procA \notin \blockedset}}
		\hspace{-2ex}
		\varphi[x_1/x]
		\, ) \enspace.
	\end{align*}
\end{definition}

The last disjunct in the definition handles the cases where the message $x_1$ from $\procA$ is immediately available to be received by $\procB$ in symbolic state $(s,\boldsymbol{r})$ and $\procA$ has not been blocked from sending. The first disjunct skips all transitions where the sender is blocked, and adds the receiver to $\blockedset$ in the recursive call to $\avail$. 
The second disjunct skips transitions where sender is not blocked, but the message exchange does not interfere with the message $x_1$ from $\procA$ to $\procB$. 
That is, one of three scenarios is true: the message exchange occurs at a different channel, or the message exchange occurs at the same channel but the receiver is unblocked, or $\netarch$ is \bag (i.e., $x$ and $x_1$ can be reordered). 

The key observation underpinning our generalized $\avail$ predicate is that the conditions for each disjunct determining how each event encountered along a protocol run ought to be handled can be made fully parametric in the communication topology. 
To demonstrate, observe that when we instantiate $\buffermap(\procC,\procE) \neq \buffermap(\procA, \procB)$ with $\procC \neq \procA \lor \procE \neq \procB$ and $\netarch=\bag$ with $\bot$ for a p2p box network, and remove the assumption that $\procB$ is always an element of $\blockedset$, we obtain Definition 5.5 from \cite{Li25OopslaOfficial}. 

The following equivalences show how to instantiate $\avail$ to decide Generalized Receive Coherence for each homogeneous network architecture. Note that $\textsc{RHS}$ takes as implicit arguments $\procA, \procB, \procC, \val, \val'$ and $s_1'$. 
\begin{proposition}
	For $\netarch \in \set{\ptpbox, \bag}$, $\exists~w.~\textsc{RHS}_\netarch(w) \iff \exists~x_1, s_2', \boldsymbol{r}.~\avail_{\netarch,\procA, \procB, \set{\procB}}(x_1, s_2', \boldsymbol{r})$. 
	For $\netarch \in \set{\senderbox, \monobox, \mailbox}$, $\exists~w.~\textsc{RHS}_\netarch(w) \iff \exists~x_1, s_2', \boldsymbol{r}.~\avail_{\netarch,\procA, \procB, \set{\procB,\procC}}(x_1, s_2', \boldsymbol{r})$. 
\end{proposition}

Using the equivalence above and our definitions of $\avail$, we can then define Generalized Receive Coherence as a predicate over symbolic protocols. 
We give the predicate for $\netarch = \senderbox$ as an example. 

\begin{definition}[Symbolic Generalized Receive Coherence for $\senderbox$] 
	\label{cond:sym-receive-coherence}
	A symbolic protocol $\SymProt$ satisfies Symbolic Generalized Receive Coherence for $\senderbox$ when for every pair of transitions $s_1 \xrightarrow{\msgFromToNS{\procA}{\procB}{x_1 \set{\varphi_1}}} s_1' \in \Delta
	$ and $
	s_2 \xrightarrow{\msgFromToNS{\procC}{\procB}{x_2 \set{\varphi_2}}} s_2' \in \Delta$ with $\procA \neq \procC$:
	\[
	\prodreach_\procB(s_1, \boldsymbol{r_1}, s_2, \boldsymbol{r_2})~\land~\varphi_1[\boldsymbol{r_1}\boldsymbol{r_1}'/\boldsymbol{r}\boldsymbol{r'}]~\land~\varphi_2[\boldsymbol{r_2}\boldsymbol{r_2}'/\boldsymbol{r}\boldsymbol{r}']%
	~\land~
	\avail_{\senderbox, \procA, \procB, \set{\procB, \procC}}(x_1, s_2', \boldsymbol{r_2'})
	\implies 
	\bot,
      \]
\end{definition}

Finally, we describe how to check Symbolic Prefix Extensibility. 
First, we establish that a stronger form of prefix extensibility holds for \ptpbox, \bag and \senderbox networks. 
\begin{restatable}{lemma}{StrongPrefixExtensibility}
  For $\netarch \in \set{\ptpbox, \senderbox, \bag}$, for every $\rho \in \AlphSyncSubscript^*$, $w \in \AlphAsyncSubscript^*$ such that $w$ is $\netarch$-channel compliant and agrees with $\rho$, there exists $u \in \AlphAsyncSubscript^*$ such that $wu$ is channel compliant, and $wu \equiv_\Procs \SyncToAsync(\rho)$. 
\end{restatable} 
  
The proof is in \cref{app:symbolic-proofs}. Prefix Extensibility constrains the runs of a global protocol, and requires that every channel-compliant word that partially completes a run can be extended to complete some run.
Lemma 7.7 states that for the aforementioned network architectures, every channel-compliant word that partially completes any synchronous word can be extended to complete the same synchronous word. 
Clearly, Lemma 7.7 implies Prefix Extensibility. 
The key fact enabling Lemma 7.7 lies in the ability to fastforward the first event $z$ with active participant $\procA$ ahead of any word $u_1$ while maintaining channel compliance, under the condition that $u_1$ contains no events with $\procA$ as active participant. 
The examples given in \cref{sec:characterization} serve as easy counterexamples to Lemma 7.7 for \mailbox and \monobox channel compliance. 
Thus, for these two networks, we must check prefix extensibility on runs of the protocol under consideration explicitly. 

To motivate our encoding of prefix extensibility, we first observe that any \monobox or \mailbox channel compliant word is also \ptpbox channel-compliant. Thus, given a word \monobox or \mailbox channel compliant word $w$, let $u$ be its completion such that $wu$ is \ptpbox channel-compliant. 
On the contrary, not every \ptpbox channel-compliant word is \monobox or \mailbox channel-compliant. 
We focus our attention on events that \monobox or \mailbox channel compliance can order, but that \ptpbox cannot. 
Both \monobox and \mailbox networks can order independent sends to a \emph{single} receiver, whose total order is fixed by $\rho$. Thus, the construction of $u$ fails when it encounters such independent sends. 
Monobox networks can additionally order independent sends to different receivers. 
In a \ptpbox network, only independent sends to different receivers that are not causally dependent can be reordered.
A key observation here is that introducing a causal dependency between independent sends to different receivers necessarily involves ordering independent sends to a single receiver, and from this we conclude the \monobox and \mailbox prefix extensibility are equivalent.

Thus, to check prefix extensibility for both \mailbox and \monobox, one must check that when a message is sent, there cannot exist a different message sent to the same \emph{receiver} that is available. 
The recursive predicate required is similar in spirit to the $\avail$ predicate for the various network models. 
Formally, it is sufficient to ensure that from any transition $s_1 \xrightarrow{\msgFromToNS{\procA}{\procB}{x\set{\varphi}}} s_2 \in \Delta$, there does not exist a run passing through $s_2$ that subsequently passes through a transition $s_3 \xrightarrow{ \msgFromToNS{\procC}{\procB}{x\set{\varphi}}} s_4$ such that $\procC \neq \procA$'s send does not depend on $\procA$ or $\procB$. We capture this in the predicate $\avail'$. 

\newcommand{\othermbavail}{\mathsf{avail'}}
\begin{definition}[Symbolic Availability]
	\begin{align*}
		\small
		\othermbavail_{\procB, \blockedset}(s, \boldsymbol{r})
		\is_\mu % \\
		&\phantom{\lor j}(
		\bigvee_{\substack{(s, \,\msgFromToNS{\procC}{\procE}{x \set{\varphi}}, \,s') \in \Delta \\ \procC \in \blockedset }}
		\hspace{-2ex}
		\exists x~\boldsymbol{r'}.\, \othermbavail_{\procB, \blockedset \cup \set{\procE}} (s', \boldsymbol{r'})
		\land 
		\varphi
		\, )  \\
		& \lor(
		\bigvee_{\substack{(s, \,\msgFromToNS{\procC}{\procE}{x \set{\varphi}}, \,s') \in \Delta \\ \procC \notin \blockedset \\ \procE \neq \procB \lor \procE \notin \blockedset}}
		\hspace{-2ex}
		\exists x~\boldsymbol{r'}.\, \othermbavail_{\procB, \blockedset}(s', \boldsymbol{r'})
		\land 
		\varphi
		\, )  \;\lor\; (
		\bigvee_{\substack{(s, \,\msgFromToNS{\procC}{\procB}{x \set{\varphi}}, \,s') \in \Delta \\ \procC \notin \blockedset}}
		\hspace{-2ex}
		\exists x~\boldsymbol{r'}.~\varphi
		\, ) \enspace.
	\end{align*}
\end{definition}

We additionally require a predicate capturing all reachable register and variable assignments to a given control state, from \cite{DBLP:conf/cav/LiSWZ25}. 
\begin{definition}[Reachability in symbolic protocol]
	\label{def:reach}
	Let $s \in S$. Then, 
	\vspace{-1ex}
	{ \small 
		\begin{align*}
			& 
			\reach(s', \boldsymbol{r'}) \is_\mu 
			\quad
			(\,
			s' = s_0 \land \boldsymbol{r'} = \rho_0 
			\,)
			% & 
			\; \lor \;
			(
			\bigvee_{\substack{
					(s, \,\msgFromToNS{\procA}{\procB}{x \set{\varphi}}, \,s') \in \Delta }}
			% &
			\hspace{-2ex}\exists x~\boldsymbol{r}.\, \reach(s,r) 
			\land 
			\varphi
			\,)
			\enspace.
		\end{align*}
	}
\end{definition} 

With this at hand, we present Symbolic Prefix Extensibility for $\netarch \in \set{\mailbox, \monobox}$. 

\begin{definition}[Prefix extensibility for mailbox, monobox]
	A symbolic protocol $\SymProt$ satisfies $\netarch$-prefix extensibility for $\netarch \in \{\mailbox,\monobox\}$ when for every transition $s \xrightarrow{\msgFromToNS{\procA}{\procB}{x \set{\varphi}}} s' \in \Delta$:
	\[
	\reach(s', \boldsymbol{r'})
	~\land~
	\othermbavail_{\procB, \set{\procA,\procB}}(s', \boldsymbol{r'})
	\implies 
	\bot
	\enspace. 
	\]
\end{definition}

\subsection{Implementability relationships}
\label{sec:impl-relationships}
\label{sec:separation-examples}
From the network-parametric definition of $\avail$, it is clear that \senderbox-$\avail$ implies \ptpbox-$\avail$ implies \bag-$\avail$. 
Because $\avail$ appears in a negative position in Generalized Receive Coherence, and the other two coherence conditions are network-agnostic, we obtain the following.

\begin{lemma}[Implementability relationships]
	\label{lm:implementability-relationships}
	Any \bag-implementable global protocol is \ptpbox-implementable, and any \ptpbox-implementable global protocol is \senderbox-implementable. 
\end{lemma} 

The strictness of these inclusions is witnessed by examples $G_a$ and $G_b$ in \cref{fig:bag-neq-p2p} and \cref{fig:p2p-neq-senderbox}. 

As a consequence of the equivalence of \mailbox and \monobox prefix extensibility, we obtain the following: 
\begin{corollary}
	A global protocol is \mailbox-implementable if and only if it is \monobox-implementable. 
\end{corollary}

We further conjecture that the \mailbox/\monobox class is contained within the \ptpbox class. 

We note that the relationship between the \senderbox, \ptpbox and \bag network architectures induced by \cref{lm:implementability-relationships} coincide with the relationship between their semantics, defined as sets of MSCs or sets of executable traces~\cite{DBLP:journals/fac/ChevrouHQ16,DBLP:journals/pacmpl/GiustoFLL23}. 
We revisit this point in further detail in \S\ref{sec:related}. 

\subsection{Complexity of finite fragments}
\label{sec:complexity}
We generalize the co-NP-completeness complexity result from \cite{Li25OopslaOfficial} to the other network architectures, and include a detailed discussion in the appendix~\cref{app:complexity}. 
We additionally show that the full generality of GCLTS is not required for the co-NP-hardness lower bound; in particular, the reduction holds even for global types with directed choice~\cite{DBLP:conf/popl/HondaYC08}. 

\begin{theorem}
	For $\netarch \in \set{\senderbox, \mailbox, \monobox, \bag}$, implementability of finite protocols is co-NP complete. 
\end{theorem}

\begin{corollary}
  For $\netarch \in \set{\ptpbox, \senderbox, \mailbox, \monobox, \bag}$, implementability of global multiparty session types with directed choice is co-NP complete. 
\end{corollary}

%%% Local Variables:
%%% mode: latex
%%% TeX-master: "main"
%%% End:

\section{Implementation and evaluation}
\label{sec:evaluation}

We extend the \toolname~\cite{DBLP:conf/cav/LiSWZ25} protocol implementability checker and verification tool so that it is parametric in the network architecture.
The original \toolname takes a symbolic global protocol $\SymProt$ as introduced in \cref{sec:symbolic} as input. It then checks whether $\SymProt$ represents a GCLTS and if so, whether it is implementable. The implementability check uses an encoding of the Coherence Conditions of~\cite{Li25OopslaOfficial} to \muCLP instances. \muval is used as the backend solver for the generated muCLP instances. \toolname also provides rudimentary support for verifying protocol-specific properties.

Our extension of \toolname replaces the encoding of \Characterization from~\cite{Li25OopslaOfficial} with an encoding of our new network-parametric Generalized \Characterization, following the discussion in \cref{sec:symbolic}. We instantiate the encoding for the five concrete network architectures $\netarch \in \{\ptpbox,\senderbox,\mailbox,\monobox,\bag\}$. The new tool, \newtoolname, allows users to select from one of these architectures when invoking the tool.

Absent any implementation bugs unbeknownst to us, \newtoolname is sound and complete relative to the soundness and completeness of \muval. Since implementability of symbolic protocols is undecidable, \muval may diverge on some of the generated \muCLP instances. 

Our evaluation of \newtoolname aims to answer the following questions:
\begin{enumerate}[label=Q\arabic*]
\item Does \newtoolname correctly decide implementability? \label{q:correctness} 
\item How does \newtoolname's performance (i.e., verification times / termination behavior) change across different network architectures? \label{q:performance} 
\end{enumerate}

\paragraph{Benchmarks}
To help answer both questions, we start from the benchmark suite used in the evaluation of \toolname~\cite{DBLP:conf/cav/LiSWZ25}. These benchmarks subsume various benchmark suites in the multiparty session type literature but also includes new ones for further diversification, in particular, to increase the number of interesting non-implementable benchmarks. Examples include idealized specifications of various web services and distributed protocols. To further help us answer \ref{q:correctness} we added select mini benchmarks that separate each pair of considered network architectures with respect to implementability. These benchmarks are discussed in~\cref{sec:separation-examples}.

\newcommand{\questionsymbol}{\textcolor{gray}{?}}
\newcommand{\questionsymbolB}{\textcolor{black}{?}}
\newcommand{\implYes}{\yessymbol}
\newcommand{\implNo}{\nosymbol}
\newcommand{\implB}{\yessymbol}%{impl.}
\newcommand{\nonimplB}{\nosymbol}%{non-impl.}
\newcommand{\inconclB}{\questionsymbol}%{inconcl.}
\newcommand{\projB}{\yessymbol}%{proj.}
\newcommand{\nonprojB}{\nonimplB}%{non-proj.}
\newcommand{\timeoutB}{T/O} % (30s)}

% Defining a center-aligned column with fixed width 
\newcolumntype{C}[1]{>{\centering\arraybackslash}p{#1}}

\begin{table}[t]\centering
\scriptsize
\begin{tabular}{|C{0.5cm}|l|c|C{0.4cm}r|C{0.4cm}r|C{0.4cm}r|C{0.4cm}r|C{0.4cm}r|}
	\hline
	Source & Example & $\card{\Procs}$ & \ptpbox & Time & \senderbox & Time & \mailbox & Time & \monobox & Time & \bag & Time \\
	\hline \hline
	
	\multirow{10}{*}{\cite{DBLP:journals/pacmpl/00020HNY20}} & 
	Calculator & 2 & \implB & 0.6s & \implB & 0.7s & \implB & 0.7s & \implB & 0.8s & \nonimplB & 16.5s \\
	& Fibonacci & 2 & \implB & 0.5s & \implB & 0.5s & \implB & 0.5s & \implB & 0.5s & \nonimplB & 3.3s \\
	& HigherLower & 3 & \implB & 15.8s & \implB & 13.6s & \nonimplB & 17.0s & \nonimplB & 18.5s & \nonimplB & 45.6s \\
	& HTTP & 2 & \implB & 0.5s & \implB & 0.4s & \implB & 0.4s & \implB & 0.5s & \nonimplB & 24.4s \\
	& Negotiation & 2 & \implB & 1.0s & \implB & 1.0s & \implB & 1.0s & \implB & 1.0s & \nonimplB & 26.1s \\
	& OnlineWallet & 3 & \implB & 17.5s & \implB & 16.6s & \implB & 16.5s & \implB & 20.6s & \nonimplB & 100.0s \\
	& SH & 3 & \implB & 237.1s & \implB & 244.0s & \implB & 256.1s & \implB & 257.2s & \inconclB & \timeoutB \\
	& Ticket & 2 & \implB & 0.6s & \implB & 0.6s & \implB & 0.6s & \implB & 0.6s & \nonimplB & 3.8s \\
	& TravelAgency & 2 & \implB & 9.6s & \implB & 9.7s & \nonimplB & 12.2s & \nonimplB & 11.6s & \nonimplB & 18.3s \\
	& TwoBuyer & 3 & \implB & 4.2s & \implB & 4.0s & \nonimplB & 6.3s & \nonimplB & 6.3s & \implB & 9.3s \\
	\cline{1-1}
	
	\multirow{6}{*}{\cite{DBLP:conf/ecoop/VassorY24}} &
	DoubleBuffering & 3 & \implB & 1.5s & \implB & 1.5s & \nonimplB & 2.4s & \nonimplB & 2.3s & \implB & 2.5s \\
	& OAuth & 3 & \implB & 6.5s & \implB & 6.6s & \implB & 10.5s & \implB & 10.3s & \nonimplB & 18.1s \\
	& PlusMinus & 3 & \implB & 5.5s & \implB & 5.3s & \nonimplB & 7.7s & \nonimplB & 8.6s & \implB & 13.2s \\
	& RingMax & 7 & \implB & 3.7s & \implB & 3.6s & \implB & 5.3s & \implB & 6.8s & \implB & 16.5s \\
	& SimpleAuth & 2 & \implB & 0.5s & \implB & 0.5s & \implB & 0.5s & \implB & 0.5s & \nonimplB & 4.7s \\
	& TravelAgency2 & 2 & \implB & 0.5s & \implB & 0.5s & \implB & 0.5s & \implB & 0.5s & \implB & 3.7s \\
	\cline{1-1}
	
	\multirow{4}{*}{\cite{DBLP:conf/cav/LiSWZ23}} &
	send-validity-yes & 4 & \implB & 1.9s & \implB & 2.7s & \implB & 2.7s & \implB & 2.7s & \implB & 3.6s \\
	& send-validity-no & 4 & \nonimplB & 1.9s & \nonimplB & 1.9s & \nonimplB & 2.7s & \nonimplB & 2.7s & \nonimplB & 3.6s \\
	& receive-validity-yes & 3 & \implB & 5.3s & \implB & 5.3s & \implB & 6.5s & \implB & 5.9s & \implB & 7.8s \\
	& receive-validity-no & 3 & \nonimplB & 3.7s & \nonimplB & 3.8s & \nonimplB & 4.6s & \nonimplB & 4.0s & \nonimplB & 5.4s \\
	\cline{1-1}
	
	\multirow{8}{*}{\cite{Li25OopslaOfficial}} &
	symbolic-two-bidder-yes & 3 & \implB & 24.0s & \implB & 21.6s & \implB & 25.2s & \implB & 22.0s & \implB & 31.7s \\
	& symbolic-two-bidder-no1 & 3 & \nonimplB & 23.9s & \nonimplB & 20.3s & \nonimplB & 22.9s & \nonimplB & 18.3s & \nonimplB & 29.5s \\
	& figure12-yes & 3 & \implB & 2.0s & \implB & 2.0s & \implB & 8.7s & \implB & 8.5s & \implB & 12.8s \\
	& figure12-no & 3 & \nonimplB & 3.0s & \nonimplB & 3.0s & \nonimplB & 9.7s & \nonimplB & 10.7s & \nonimplB & 13.8s \\
	& symbolic-send-validity-yes & 4 & \implB & 6.6s & \implB & 6.5s & \implB & 10.0s & \implB & 10.5s & \implB & 16.0s \\
	& symbolic-send-validity-no & 4 & \nonimplB & 5.6s & \nonimplB & 5.4s & \nonimplB & 8.3s & \nonimplB & 8.8s & \nonimplB & 13.2s \\
	& symbolic-receive-validity-yes & 3 & \implB & 6.4s & \implB & 6.3s & \nonimplB & 8.7s & \nonimplB & 8.4s & \implB & 11.8s \\
	& symbolic-receive-validity-no & 3 & \nonimplB & 8.0s & \nonimplB & 7.8s & \nonimplB & 12.3s & \nonimplB & 12.3s & \nonimplB & 14.5s \\
	\cline{1-1}
	
	\multirow{2}{*}{\cite{DBLP:conf/ictac/Cruz-FilipeGLMP22}} &
	fwd-auth-yes & 3 & \implB & 10.5s & \implB & 10.5s & \implB & 15.6s & \implB & 16.1s & \nonimplB & 27.3s \\
	& fwd-auth-no & 3 & \inconclB & \timeoutB & \inconclB & \timeoutB & \nonimplB & 157.4s & \nonimplB & 163.0s & \nonimplB & 176.1s \\
	\cline{1-1}
	
	\multirow{7}{*}{\cite{DBLP:conf/cav/LiSWZ25}} &
	symbolic-two-bidder-no2 & 3 & \nonimplB & 23.8s & \nonimplB & 29.5s & \nonimplB & 31.1s & \nonimplB & 28.5s & \nonimplB & 40.1s \\
	& higher-lower-ultimate & 3 & \implB & 13.3s & \implB & 13.1s & \nonimplB & 20.5s & \nonimplB & 25.3s & \nonimplB & 30.8s \\
	& higher-lower-winning & 3 & \inconclB & \timeoutB & \projB & 36.3s & \nonimplB & 37.5s & \nonimplB & 45.1s & \implB & 62.2s \\
	& higher-lower-no & 3 & \nonimplB & 13.4s & \nonimplB & 12.9s & \nonimplB & 20.9s & \nonimplB & 29.9s & \nonimplB & 45.3s \\
	& higher-lower-encrypt-yes & 4 & \implB & 10.0s & \implB & 9.8s & \nonimplB & 12.9s & \nonimplB & 14.3s & \implB & 25.0s \\
	& higher-lower-encrypt-no & 4 & \nonimplB & 77.3s & \nonimplB & 88.8s & \nonimplB & 100.0s & \nonimplB & 109.2s & \nonimplB & 126.1s \\
	& higher-lower-mixed & 3 & \nonimplB & 33.6s & \nonimplB & 33.8s & \nonimplB & 36.4s & \nonimplB & 36.2s & \nonimplB & 41.3s \\
	\cline{1-1}
	
	\multirow{5}{*}{new} 
	& bag-no-p2p-yes & 2 & \implB & 0.4s & \implB & 0.4s & \implB & 3.4s & \implB & 0.6s & \nonimplB & 1.8s \\ 
	& mb-no-p2p-yes & 3 & \implB & 0.9s & \implB & 0.9s & \nonimplB & 1.3s & \nonimplB & 1.3s & \implB & 1.3s \\ 
	& p2p-no-sb-yes & 4 & \nonimplB & 4.0s & \implB & 4.1s & \nonimplB & 5.3s & \nonimplB & 5.3s & \nonimplB & 6.7s \\
	& mb-no-monob-no & 3 & \nonimplB & 3.2s & \implB & 3.2s & \nonimplB & 4.1s & \nonimplB & 4.1s & \nonimplB& 5.4s \\ 
	& monob-no-mb-no & 4 & \implB & 1.3s & \implB & 1.3s & \nonimplB & 1.9s & \nonimplB & 1.9s & \implB & 1.9s \\
	\hline \hline 
	 & Average & & & 6.7s & & 5.8s & & 7.6s &  & 7.1s &  & 12.6s \\
	\hline 
\end{tabular}
\caption{\scriptsize Run times and verification results for all five network architectures.  
	For each example, we report the number of participants ($\card{\Procs}$), and for each architecture, the result: 
	\yessymbol\xspace for implementable, 
	\nosymbol\xspace for non-implementable, 
	and \questionsymbol\xspace for inconclusive due to timeout, and run time, with a 15s timeout per \muCLP instance (T/O). The average run times across all benchmark listed in the last row are computed only over those benchmarks for which none of the configurations timed out.
	\label{tab:comparison_final}}
      \vspace*{-4ex}
\end{table}

\paragraph{Experiment and results}
To answer our two questions, we run \newtoolname on the full benchmark suite and all considered network architectures. The experiment is run on a 2024 MacBook Air with an Apple M3 chip and 16GB of RAM, with run times averaged over 10 runs. The results are summarized in \cref{tab:comparison_final}. Run times reported are the sum of GCLTS checking time and implementability checking time, with timeouts for individual \muCLP instances specified separately.

\paragraph{Correctness (\ref{q:correctness})}
We first observe that the implementability results for peer-to-peer box are consistent with those reported in~\cite{DBLP:conf/cav/LiSWZ25} for the common benchmarks. Next, note that the results are also consistent with the implementability relationships between network architectures identified in \cref{lm:implementability-relationships}. For example, the subset of benchmarks identified as implementable under \senderbox is a strict superset of those for \ptpbox. Finally, we did a manual inspection of individual results for our separation mini benchmarks and a random selection of other results. In all cases, the produced results were determined to be correct.

\paragraph{Performance (\ref{q:performance})}
\newtoolname shows similar running times and termination behavior for the new network architectures compared to the previously supported \ptpbox, with \bag behaving the slowest. The slower behavior of \bag is to be expected since the relaxed ordering on message buffers increases the state space to be explored by the message availability check used to implement the GRC condition. In general, the relationships of the average run times appear to track the theoretical relationships between network architectures identified in \cref{sec:impl-relationships}.

%%% Local Variables:
%%% mode: latex
%%% TeX-master: "main"
%%% End:

\section{Related Work} 
\label{sec:related}

\paragraph{Global protocol specifications} 
Global protocol specifications in the form of message sequence charts found early industry adoption by the ITU standard~\cite{ITU-T-Z120-2011} in 1993, was subsequently incorporated into UML~\cite{umlwebsite} in 2005, and is part of the Web Service Choreography Description Language~\cite{w3c-ws-cdl-2005}. 
Global specifications are widely studied in academia in the form of multiparty session types and choreographic programming. 
Multiparty session types (MSTs) have enjoyed widespread implementation in a variety of programming languages 
including Python \cite{DBLP:journals/fmsd/DemangeonHHNY15, DBLP:journals/corr/NeykovaY16, DBLP:journals/fac/NeykovaBY17}, Java \cite{DBLP:conf/fase/HuY16, DBLP:conf/fase/HuY17}, C \cite{DBLP:conf/tools/NgYH12}, Go \cite{DBLP:conf/icse/LangeNTY18, DBLP:journals/pacmpl/CastroHJNY19}, Scala \cite{DBLP:conf/ecoop/Castro-PerezY23}, Rust \cite{DBLP:conf/ppopp/CutnerYV22,DBLP:journals/darts/LagaillardieNY22}, OCaml \cite{DBLP:conf/ecoop/ImaiNYY19} and F\# \cite{DBLP:conf/cc/NeykovaHYA18}. 
Application domains for MSTs include operating systems~\cite{DBLP:conf/eurosys/FahndrichAHHHLL06},
high performance computing~\cite{DBLP:conf/pvm/HondaMMNVY12, DBLP:conf/fpl/NiuNYWYL16,demuijnckhughes_et_al:LIPIcs.ECOOP.2019.6},
cyber-physical systems~\cite{DBLP:conf/ecoop/MajumdarPYZ19, DBLP:journals/pacmpl/MajumdarYZ20}, and
web services~\cite{DBLP:conf/tgc/YoshidaHNN13}.
Choreographic programming frameworks have been implemented in Java~\cite{DBLP:journals/toplas/GiallorenzoMP24}, Haskell~\cite{ DBLP:journals/corr/abs-2303-00924}, Rust~\cite{chorus, kashiwa2023portableefficientpracticallibrarylevel} and applied to distributed architecture~\cite{depalma2024functionasaservicechoreographicprogramminglanguage}, cryptographic security protocols~\cite{DBLP:conf/sp/GancherGSDP23}, and cyber-physical systems~\cite{DBLP:conf/forte/Cruz-FilipeM16}. 
We refer the reader to  \cite{DBLP:books/sp/24/Yoshida24} and \cite{M23} for a comprehensive survey of MST and choreography applicability respectively.

\paragraph{Network-parametric results from concurrency theory} 
\citet{DBLP:conf/concur/BolligGFLLS21} and \citet{DBLP:journals/pacmpl/GiustoFLL23} study the synchronizability problem, which asks whether a communicating finite state machine can be precisely approximated in terms of its ``synchronous'' executions. 
\citet{DBLP:conf/concur/BolligGFLLS21} presents a framework based on MSO logic and special treewidth, and uses it to establish decidability results for existential and universal $k$-boundedness~\cite{DBLP:journals/iandc/GenestKM06}, weak synchrony and weak $k$-synchrony~\cite{DBLP:conf/cav/BouajjaniEJQ18} of \ptpbox and \mailbox communication. 

\citet{DBLP:journals/pacmpl/GiustoFLL23} generalize \citet{DBLP:conf/concur/BolligGFLLS21} to five additional communication models, defined as sets of message sequence charts (MSCs): \senderbox, \monobox, \bag, causally ordered and realizable with synchronous communication. MSCs are partial order graphs on asynchronous events that specify a total order per participant. 
MSCs can thus specify unmatched sends and out-of-order receptions, but not branching behavior or recursion. 
\citet{DBLP:journals/pacmpl/GiustoFLL23} further generalize synchronizability to any notion expressible in terms of a class of MSCs with bounded special treewidth. 
The generalization thus rests on an encoding of communication models and synchronizability properties as MSO formulae, and Courcelle's Theorem, which applies to any graph with bounded special treewidth. 

Non-FIFO (equivalent to \bag) communication is considered alongside FIFO communication in \cite{DBLP:journals/tcs/Lohrey03}, which studies the safe realizability problem of HMSCs with respect to communicating finite state machines. HMSCs generalizee GCLTS along two dimensions: firstly, while specifications are required to satisfy FIFO restrictions, they are not required to be 1-synchronous;
secondly, branching choice is unrestricted, relaxing the sender-driven choice assumption of GCLTS. 
Unlike our setting, HMSCs are not given infinite word semantics. 
\citet{DBLP:journals/tcs/Lohrey03} shows that all (un)decidability results generalize to \bag communication, by exploiting a reduction between FIFO and non-FIFO communication to transfer upper bounds of safe realizability. 

Recently, the realizability of global types for both p2p and rendezvous synchronous communication was studied in \cite{DBLP:journals/corr/abs-2507-17354}. 
Similar to our notion of channel compliance, the work defines communication models as sets of linearizations. 
Rendezvous synchrony is essentially asynchronous communication with a global channel bound of size 1, \textit{\`a la} \texttt{rsc} from \cite{DBLP:journals/fac/ChevrouHQ16} and \cite{DBLP:journals/pacmpl/GiustoFLL23}, and corresponds to  $\SyncToAsync(\rho)$ asynchronous words in our setting.  \cite{DBLP:journals/corr/abs-2507-17354} lifts all structural restrictions on the graph structure, notably including sender-driven choice. 
No decision procedure for realizability of finite global types is given, either in a network-specific or network-agnostic manner. The relationship between p2p FIFO-realizable and synchronously realizable global types is instead established through the definition of realizability directly. 

\citet{DBLP:journals/fac/ChevrouHQ16} study the same seven communication models considered by \cite{DBLP:journals/pacmpl/GiustoFLL23}, and establish a hierarchy that differs from \cite{DBLP:journals/pacmpl/GiustoFLL23} in the placement of \senderbox and \mailbox communication: \cite{DBLP:journals/fac/ChevrouHQ16} judges the two incomparable, whereas \cite{DBLP:journals/pacmpl/GiustoFLL23} judges \senderbox to subsume \mailbox. The discrepancy is due to the difference between a universal and existential quantifier: Chevrou et al. define communication models as sets of linearizations, all of which respect the communication model, whereas Giusto et al. define communication models as sets of MSCs, of which one exhibits a linearization that respects the communication model. The focus is on establishing a taxonomy, and not on specific decision problems concerning one or more communication models. Their results are mechanized in TLA+. 

Our notion of network architecture identifies each channel with a single message buffer, as is common in the context of global protocol specification frameworks. \citet{DBLP:journals/scp/EngelsMR02} consider a more fine-grained model where each channel is associated with an input and and output buffer. This results in a third asynchronous event per message exchange that corresponds to the transmission of a message from the output to the input buffer. They then study, among others, the implementability problem for MSCs under different communication topologies for FIFO channels. However, their notion of implementability differs from ours in that they only check whether all linearizations of the MSC are consistent with the network semantics, but not whether this set can be exactly realized by a CLTS where participants only make local observations. Their results also do not easily generalize to HMCSs and, thus, rule out features like recursion and loops in specifications.

Our axiomatic network model allows for network architectures that exhibit channels with different semantics such as FIFO and bag within a single network. \citet{DBLP:conf/concur/ClementeHS14} study the decidability of the reachability problem for finite-state CLTS with such hybrid networks. Specifically, they provide a complete characterization of the communication topologies for which decidability can be retained. 

\paragraph{Tools}
To our knowledge \newtoolname is the first tool that can check implementability of global protocols for a range of different network architectures. In this work, we build on \toolname~\cite{DBLP:conf/cav/LiSWZ25}, which implements the symbolic algorithm for checking implementability developed in~\cite{Li25OopslaOfficial} for the \ptpbox case. \toolname's closest competitors are \sessionstar~\cite{DBLP:journals/pacmpl/00020HNY20} and Rumpsteak with refinements~\cite{DBLP:conf/ecoop/VassorY24}. We refer to~\cite{DBLP:conf/cav/LiSWZ25} for a detailed comparison.

\paragraph{Synchrony}
We have excluded synchronous communication models from our investigation in this paper. 
Top-down verification for synchronous communication models have been widely studied in their own right, prominently in the form of synchronous multiparty session types~\cite{DBLP:journals/pacmpl/UdomsrirungruangY25, DBLP:conf/lics/PetersY24, DBLP:conf/ppdp/ChenDY24, DBLP:journals/jlp/GhilezanJPSY19}. Zielonka automata~\cite{DBLP:journals/ita/Zielonka87} and their ``realistic'' variant~\cite{DBLP:conf/fsttcs/AkshayDGS13} are a natural starting point for developing a unifying theory of protocols with synchronous or asynchronous semantics. Notably, the Send Coherence condition employed in \cite{Li25OopslaOfficial} and this work can be viewed as a special instance of the third semantical condition, \emph{causally closed}~\cite[Definition 9, (LC3)]{DBLP:conf/fsttcs/AkshayDGS13}, characterizing synchronous realizability of global specifications. This connection points towards the possibility of an implementability characterization that simultaneously handles synchronous and asynchronous communication.

%%% Local Variables:
%%% mode: latex
%%% TeX-master: "main"
%%% End:

\section{Conclusion} 
\label{sec:discussion} 

We presented a network-parametric, precise solution to global protocol implementability and a simple axiomatic network model that characterizes the network architectures to which this solution applies. Equipped with this solution, we derived decidability and complexity results for finite state fragments of global protocols and commonly considered network architectures as well symbolic algorithms for global protocols with potentially infinite state spaces and data domains. We implemented the latter in an existing tool and showed that the resulting new tool, \newtoolname, provides increased applicability across architectures without sacrificing performance. 

One implication of our work
is that the chosen network architecture for the implementability problem does not affect synthesis: for all considered network architectures, the global protocol is implementable if and only if the canonical implementation obtainable by local projection implements it. 
This modularity makes us hopeful that our results can be directly applied to broaden the scope of top-down methodologies. 

%%% Local Variables:
%%% mode: latex
%%% TeX-master: "main"
%%% End:

\subsubsection*{Acknowledgements.} 
This work is supported by %
the National Science Foundation under the grant agreement 2304758. 
	
\clearpage

\bibliography{merged_biblio}

@inproceedings{DBLP:conf/icse/LangeNTY18,
  author       = {Julien Lange and
                  Nicholas Ng and
                  Bernardo Toninho and
                  Nobuko Yoshida},
  editor       = {Michel Chaudron and
                  Ivica Crnkovic and
                  Marsha Chechik and
                  Mark Harman},
  title        = {A static verification framework for message passing in Go using behavioural
                  types},
  booktitle    = {Proceedings of the 40th International Conference on Software Engineering,
                  {ICSE} 2018, Gothenburg, Sweden, May 27 - June 03, 2018},
  pages        = {1137--1148},
  publisher    = {{ACM}},
  year         = {2018},
  url          = {https://doi.org/10.1145/3180155.3180157},
  doi          = {10.1145/3180155.3180157},
  bibsource    = {dblp computer science bibliography, https://dblp.org}
}

@inproceedings{DBLP:conf/cc/NeykovaHYA18,
  author       = {Rumyana Neykova and
                  Raymond Hu and
                  Nobuko Yoshida and
                  Fahd Abdeljallal},
  editor       = {Christophe Dubach and
                  Jingling Xue},
  title        = {A session type provider: compile-time {API} generation of distributed
                  protocols with refinements in F{\#}},
  booktitle    = {Proceedings of the 27th International Conference on Compiler Construction,
                  {CC} 2018, February 24-25, 2018, Vienna, Austria},
  pages        = {128--138},
  publisher    = {{ACM}},
  year         = {2018},
  url          = {https://doi.org/10.1145/3178372.3179495},
  doi          = {10.1145/3178372.3179495},
  bibsource    = {dblp computer science bibliography, https://dblp.org}
}

@inproceedings{DBLP:conf/eurosys/FahndrichAHHHLL06,
  author       = {Manuel F{\"{a}}hndrich and
                  Mark Aiken and
                  Chris Hawblitzel and
                  Orion Hodson and
                  Galen C. Hunt and
                  James R. Larus and
                  Steven Levi},
  editor       = {Yolande Berbers and
                  Willy Zwaenepoel},
  title        = {Language support for fast and reliable message-based communication
                  in singularity {OS}},
  booktitle    = {Proceedings of the 2006 EuroSys Conference, Leuven, Belgium, April
                  18-21, 2006},
  pages        = {177--190},
  publisher    = {{ACM}},
  year         = {2006},
  url          = {https://doi.org/10.1145/1217935.1217953},
  doi          = {10.1145/1217935.1217953},
  bibsource    = {dblp computer science bibliography, https://dblp.org}
}

@inproceedings{DBLP:conf/pvm/HondaMMNVY12,
  author       = {Kohei Honda and
                  Eduardo R. B. Marques and
                  Francisco Martins and
                  Nicholas Ng and
                  Vasco Thudichum Vasconcelos and
                  Nobuko Yoshida},
  editor       = {Jesper Larsson Tr{\"{a}}ff and
                  Siegfried Benkner and
                  Jack J. Dongarra},
  title        = {Verification of {MPI} Programs Using Session Types},
  booktitle    = {Recent Advances in the Message Passing Interface - 19th European {MPI}
                  Users' Group Meeting, EuroMPI 2012, Vienna, Austria, September 23-26,
                  2012. Proceedings},
  series       = {Lecture Notes in Computer Science},
  volume       = {7490},
  pages        = {291--293},
  publisher    = {Springer},
  year         = {2012},
  url          = {https://doi.org/10.1007/978-3-642-33518-1\_37},
  doi          = {10.1007/978-3-642-33518-1\_37},
  bibsource    = {dblp computer science bibliography, https://dblp.org}
}

@inproceedings{DBLP:conf/ecoop/MajumdarPYZ19,
  author       = {Rupak Majumdar and
                  Marcus Pirron and
                  Nobuko Yoshida and
                  Damien Zufferey},
  editor       = {Alastair F. Donaldson},
  title        = {Motion Session Types for Robotic Interactions (Brave New Idea Paper)},
  booktitle    = {33rd European Conference on Object-Oriented Programming, {ECOOP} 2019,
                  July 15-19, 2019, London, United Kingdom},
  series       = {LIPIcs},
  volume       = {134},
  pages        = {28:1--28:27},
  publisher    = {Schloss Dagstuhl - Leibniz-Zentrum f{\"{u}}r Informatik},
  year         = {2019},
  url          = {https://doi.org/10.4230/LIPIcs.ECOOP.2019.28},
  doi          = {10.4230/LIPICS.ECOOP.2019.28},
  bibsource    = {dblp computer science bibliography, https://dblp.org}
}

@inproceedings{DBLP:conf/tgc/YoshidaHNN13,
  author       = {Nobuko Yoshida and
                  Raymond Hu and
                  Rumyana Neykova and
                  Nicholas Ng},
  editor       = {Mart{\'{\i}}n Abadi and
                  Alberto Lluch{-}Lafuente},
  title        = {The Scribble Protocol Language},
  booktitle    = {Trustworthy Global Computing - 8th International Symposium, {TGC}
                  2013, Buenos Aires, Argentina, August 30-31, 2013, Revised Selected
                  Papers},
  series       = {Lecture Notes in Computer Science},
  volume       = {8358},
  pages        = {22--41},
  publisher    = {Springer},
  year         = {2013},
  url          = {https://doi.org/10.1007/978-3-319-05119-2\_3},
  doi          = {10.1007/978-3-319-05119-2\_3},
  bibsource    = {dblp computer science bibliography, https://dblp.org}
}

@inproceedings{DBLP:conf/sdl/MauwR97,
  author       = {Sjouke Mauw and
                  Michel A. Reniers},
  editor       = {Ana R. Cavalli and
                  Amardeo Sarma},
  title        = {High-level message sequence charts},
  booktitle    = {{SDL} '97 Time for Testing, SDL, {MSC} and Trends - 8th International
                  {SDL} Forum, Evry, France, 23-29 September 1997, Proceedings},
  pages        = {291--306},
  publisher    = {Elsevier},
  year         = {1997},
  bibsource    = {dblp computer science bibliography, https://dblp.org}
}

@inproceedings{DBLP:conf/ac/GenestMP03,
  author       = {Blaise Genest and
                  Anca Muscholl and
                  Doron A. Peled},
  editor       = {J{\"{o}}rg Desel and
                  Wolfgang Reisig and
                  Grzegorz Rozenberg},
  title        = {Message Sequence Charts},
  booktitle    = {Lectures on Concurrency and Petri Nets, Advances in Petri Nets [This
                  tutorial volume originates from the 4th Advanced Course on Petri Nets,
                  {ACPN} 2003, held in Eichst{\"{a}}tt, Germany in September 2003. In
                  addition to lectures given at {ACPN} 2003, additional chapters have
                  been commissioned]},
  series       = {Lecture Notes in Computer Science},
  volume       = {3098},
  pages        = {537--558},
  publisher    = {Springer},
  year         = {2003},
  url          = {https://doi.org/10.1007/978-3-540-27755-2\_15},
  doi          = {10.1007/978-3-540-27755-2\_15},
  bibsource    = {dblp computer science bibliography, https://dblp.org}
}

@inproceedings{DBLP:conf/acsd/GenestM05,
  author       = {Blaise Genest and
                  Anca Muscholl},
  title        = {Message Sequence Charts: {A} Survey},
  booktitle    = {Fifth International Conference on Application of Concurrency to System
                  Design {(ACSD} 2005), 6-9 June 2005, St. Malo, France},
  pages        = {2--4},
  publisher    = {{IEEE} Computer Society},
  year         = {2005},
  url          = {https://doi.org/10.1109/ACSD.2005.25},
  doi          = {10.1109/ACSD.2005.25},
  bibsource    = {dblp computer science bibliography, https://dblp.org}
}

@inproceedings{DBLP:conf/concur/GazagnaireGHTY07,
  author       = {Thomas Gazagnaire and
                  Blaise Genest and
                  Lo{\"{\i}}c H{\'{e}}lou{\"{e}}t and
                  P. S. Thiagarajan and
                  Shaofa Yang},
  editor       = {Lu{\'{\i}}s Caires and
                  Vasco Thudichum Vasconcelos},
  title        = {Causal Message Sequence Charts},
  booktitle    = {{CONCUR} 2007 - Concurrency Theory, 18th International Conference,
                  {CONCUR} 2007, Lisbon, Portugal, September 3-8, 2007, Proceedings},
  series       = {Lecture Notes in Computer Science},
  volume       = {4703},
  pages        = {166--180},
  publisher    = {Springer},
  year         = {2007},
  url          = {https://doi.org/10.1007/978-3-540-74407-8\_12},
  doi          = {10.1007/978-3-540-74407-8\_12},
  bibsource    = {dblp computer science bibliography, https://dblp.org}
}

@article{DBLP:journals/tosem/RoychoudhuryGS12,
  author       = {Abhik Roychoudhury and
                  Ankit Goel and
                  Bikram Sengupta},
  title        = {Symbolic Message Sequence Charts},
  journal      = {{ACM} Trans. Softw. Eng. Methodol.},
  volume       = {21},
  number       = {2},
  pages        = {12:1--12:44},
  year         = {2012},
  url          = {https://doi.org/10.1145/2089116.2089122},
  doi          = {10.1145/2089116.2089122},
  bibsource    = {dblp computer science bibliography, https://dblp.org}
}

@article{DBLP:journals/tse/AlurEY03,
  author       = {Rajeev Alur and
                  Kousha Etessami and
                  Mihalis Yannakakis},
  title        = {Inference of Message Sequence Charts},
  journal      = {{IEEE} Trans. Software Eng.},
  volume       = {29},
  number       = {7},
  pages        = {623--633},
  year         = {2003},
  url          = {https://doi.org/10.1109/TSE.2003.1214326},
  doi          = {10.1109/TSE.2003.1214326},
  bibsource    = {dblp computer science bibliography, https://dblp.org}
}

@article{DBLP:journals/tcs/Lohrey03,
  author       = {Markus Lohrey},
  title        = {Realizability of high-level message sequence charts: closing the gaps},
  journal      = {Theor. Comput. Sci.},
  volume       = {309},
  number       = {1-3},
  pages        = {529--554},
  year         = {2003},
  url          = {https://doi.org/10.1016/j.tcs.2003.08.002},
  doi          = {10.1016/J.TCS.2003.08.002},
  bibsource    = {dblp computer science bibliography, https://dblp.org}
}

@inproceedings{DBLP:conf/concur/AlurY99,
  author       = {Rajeev Alur and
                  Mihalis Yannakakis},
  editor       = {Jos C. M. Baeten and
                  Sjouke Mauw},
  title        = {Model Checking of Message Sequence Charts},
  booktitle    = {{CONCUR} '99: Concurrency Theory, 10th International Conference, Eindhoven,
                  The Netherlands, August 24-27, 1999, Proceedings},
  series       = {Lecture Notes in Computer Science},
  volume       = {1664},
  pages        = {114--129},
  publisher    = {Springer},
  year         = {1999},
  url          = {https://doi.org/10.1007/3-540-48320-9\_10},
  doi          = {10.1007/3-540-48320-9\_10},
  bibsource    = {dblp computer science bibliography, https://dblp.org}
}

@inproceedings{DBLP:conf/mfcs/MuschollP99,
  author       = {Anca Muscholl and
                  Doron A. Peled},
  editor       = {Miroslaw Kutylowski and
                  Leszek Pacholski and
                  Tomasz Wierzbicki},
  title        = {Message Sequence Graphs and Decision Problems on Mazurkiewicz Traces},
  booktitle    = {Mathematical Foundations of Computer Science 1999, 24th International
                  Symposium, MFCS'99, Szklarska Poreba, Poland, September 6-10, 1999,
                  Proceedings},
  series       = {Lecture Notes in Computer Science},
  volume       = {1672},
  pages        = {81--91},
  publisher    = {Springer},
  year         = {1999},
  url          = {https://doi.org/10.1007/3-540-48340-3\_8},
  doi          = {10.1007/3-540-48340-3\_8},
  bibsource    = {dblp computer science bibliography, https://dblp.org}
}

@inproceedings{DBLP:conf/stacs/Morin02,
  author       = {R{\'{e}}mi Morin},
  editor       = {Helmut Alt and
                  Afonso Ferreira},
  title        = {Recognizable Sets of Message Sequence Charts},
  booktitle    = {{STACS} 2002, 19th Annual Symposium on Theoretical Aspects of Computer
                  Science, Antibes - Juan les Pins, France, March 14-16, 2002, Proceedings},
  series       = {Lecture Notes in Computer Science},
  volume       = {2285},
  pages        = {523--534},
  publisher    = {Springer},
  year         = {2002},
  url          = {https://doi.org/10.1007/3-540-45841-7\_43},
  doi          = {10.1007/3-540-45841-7\_43},
  bibsource    = {dblp computer science bibliography, https://dblp.org}
}

@article{DBLP:journals/jcss/GenestMSZ06,
  author       = {Blaise Genest and
                  Anca Muscholl and
                  Helmut Seidl and
                  Marc Zeitoun},
  title        = {Infinite-state high-level MSCs: Model-checking and realizability},
  journal      = {J. Comput. Syst. Sci.},
  volume       = {72},
  number       = {4},
  pages        = {617--647},
  year         = {2006},
  url          = {https://doi.org/10.1016/j.jcss.2005.09.007},
  doi          = {10.1016/J.JCSS.2005.09.007},
  bibsource    = {dblp computer science bibliography, https://dblp.org}
}

@inproceedings{DBLP:conf/popl/HondaYC08,
  author       = {Kohei Honda and
                  Nobuko Yoshida and
                  Marco Carbone},
  editor       = {George C. Necula and
                  Philip Wadler},
  title        = {Multiparty asynchronous session types},
  booktitle    = {Proceedings of the 35th {ACM} {SIGPLAN-SIGACT} Symposium on Principles
                  of Programming Languages, {POPL} 2008, San Francisco, California,
                  USA, January 7-12, 2008},
  pages        = {273--284},
  publisher    = {{ACM}},
  year         = {2008},
  url          = {https://doi.org/10.1145/1328438.1328472},
  doi          = {10.1145/1328438.1328472},
  bibsource    = {dblp computer science bibliography, https://dblp.org}
}

@inproceedings{DBLP:conf/concur/BocchiHTY10,
  author       = {Laura Bocchi and
                  Kohei Honda and
                  Emilio Tuosto and
                  Nobuko Yoshida},
  editor       = {Paul Gastin and
                  Fran{\c{c}}ois Laroussinie},
  title        = {A Theory of Design-by-Contract for Distributed Multiparty Interactions},
  booktitle    = {{CONCUR} 2010 - Concurrency Theory, 21th International Conference,
                  {CONCUR} 2010, Paris, France, August 31-September 3, 2010. Proceedings},
  series       = {Lecture Notes in Computer Science},
  volume       = {6269},
  pages        = {162--176},
  publisher    = {Springer},
  year         = {2010},
  url          = {https://doi.org/10.1007/978-3-642-15375-4\_12},
  doi          = {10.1007/978-3-642-15375-4\_12},
  bibsource    = {dblp computer science bibliography, https://dblp.org}
}

@inproceedings{DBLP:conf/tgc/BocchiDY12,
  author       = {Laura Bocchi and
                  Romain Demangeon and
                  Nobuko Yoshida},
  editor       = {Catuscia Palamidessi and
                  Mark Dermot Ryan},
  title        = {A Multiparty Multi-session Logic},
  booktitle    = {Trustworthy Global Computing - 7th International Symposium, {TGC}
                  2012, Newcastle upon Tyne, UK, September 7-8, 2012, Revised Selected
                  Papers},
  series       = {Lecture Notes in Computer Science},
  volume       = {8191},
  pages        = {97--111},
  publisher    = {Springer},
  year         = {2012},
  url          = {https://doi.org/10.1007/978-3-642-41157-1\_7},
  doi          = {10.1007/978-3-642-41157-1\_7},
  bibsource    = {dblp computer science bibliography, https://dblp.org}
}

@article{DBLP:journals/jlp/ToninhoY17,
  author       = {Bernardo Toninho and
                  Nobuko Yoshida},
  title        = {Certifying data in multiparty session types},
  journal      = {J. Log. Algebraic Methods Program.},
  volume       = {90},
  pages        = {61--83},
  year         = {2017},
  url          = {https://doi.org/10.1016/j.jlamp.2016.11.005},
  doi          = {10.1016/J.JLAMP.2016.11.005},
  bibsource    = {dblp computer science bibliography, https://dblp.org}
}

@article{DBLP:journals/pacmpl/00020HNY20,
  author       = {Fangyi Zhou and
                  Francisco Ferreira and
                  Raymond Hu and
                  Rumyana Neykova and
                  Nobuko Yoshida},
  title        = {Statically verified refinements for multiparty protocols},
  journal      = {Proc. {ACM} Program. Lang.},
  volume       = {4},
  number       = {{OOPSLA}},
  pages        = {148:1--148:30},
  year         = {2020},
  url          = {https://doi.org/10.1145/3428216},
  doi          = {10.1145/3428216},
  bibsource    = {dblp computer science bibliography, https://dblp.org}
}

@inproceedings{DBLP:conf/cav/LiSWZ23,
  author       = {Elaine Li and
                  Felix Stutz and
                  Thomas Wies and
                  Damien Zufferey},
  editor       = {Constantin Enea and
                  Akash Lal},
  title        = {Complete Multiparty Session Type Projection with Automata},
  booktitle    = {Computer Aided Verification - 35th International Conference, {CAV}
                  2023, Paris, France, July 17-22, 2023, Proceedings, Part {III}},
  series       = {Lecture Notes in Computer Science},
  volume       = {13966},
  pages        = {350--373},
  publisher    = {Springer},
  year         = {2023},
  url          = {https://doi.org/10.1007/978-3-031-37709-9\_17},
  doi          = {10.1007/978-3-031-37709-9\_17},
  bibsource    = {dblp computer science bibliography, https://dblp.org}
}

@article{DBLP:journals/pacmpl/UnnoTGK23,
  author       = {Hiroshi Unno and
                  Tachio Terauchi and
                  Yu Gu and
                  Eric Koskinen},
  title        = {Modular Primal-Dual Fixpoint Logic Solving for Temporal Verification},
  journal      = {Proc. {ACM} Program. Lang.},
  volume       = {7},
  number       = {{POPL}},
  pages        = {2111--2140},
  year         = {2023},
  url          = {https://doi.org/10.1145/3571265},
  doi          = {10.1145/3571265},
  bibsource    = {dblp computer science bibliography, https://dblp.org}
}

@article{DBLP:journals/jacm/BrandZ83,
  author       = {Daniel Brand and
                  Pitro Zafiropulo},
  title        = {On Communicating Finite-State Machines},
  journal      = {J. {ACM}},
  volume       = {30},
  number       = {2},
  pages        = {323--342},
  year         = {1983},
  url          = {https://doi.org/10.1145/322374.322380},
  doi          = {10.1145/322374.322380},
  bibsource    = {dblp computer science bibliography, https://dblp.org}
}

@article{DBLP:journals/corr/abs-2303-00924,
  author       = {Gan Shen and
                  Shun Kashiwa and
                  Lindsey Kuper},
  title        = {HasChor: Functional Choreographic Programming for All (Functional
                  Pearl)},
  journal      = {CoRR},
  volume       = {abs/2303.00924},
  year         = {2023},
  url          = {https://doi.org/10.48550/arXiv.2303.00924},
  doi          = {10.48550/ARXIV.2303.00924},
  eprinttype    = {arXiv},
  eprint       = {2303.00924},
  bibsource    = {dblp computer science bibliography, https://dblp.org}
}

@article{DBLP:journals/lmcs/FinkelL23,
  author       = {Alain Finkel and
                  {\'{E}}tienne Lozes},
  title        = {Synchronizability of Communicating Finite State Machines is not Decidable},
  journal      = {Log. Methods Comput. Sci.},
  volume       = {19},
  number       = {4},
  year         = {2023},
  url          = {https://doi.org/10.46298/lmcs-19(4:33)2023},
  doi          = {10.46298/LMCS-19(4:33)2023},
  bibsource    = {dblp computer science bibliography, https://dblp.org}
}

@article{DBLP:journals/tcs/Cruz-FilipeM20,
  author       = {Lu{\'{\i}}s Cruz{-}Filipe and
                  Fabrizio Montesi},
  title        = {A core model for choreographic programming},
  journal      = {Theor. Comput. Sci.},
  volume       = {802},
  pages        = {38--66},
  year         = {2020},
  url          = {https://doi.org/10.1016/j.tcs.2019.07.005},
  doi          = {10.1016/J.TCS.2019.07.005},
  bibsource    = {dblp computer science bibliography, https://dblp.org}
}

@inproceedings{DBLP:conf/ecoop/GiallorenzoMPRS21,
  author       = {Saverio Giallorenzo and
                  Fabrizio Montesi and
                  Marco Peressotti and
                  David Richter and
                  Guido Salvaneschi and
                  Pascal Weisenburger},
  editor       = {Anders M{\o}ller and
                  Manu Sridharan},
  title        = {Multiparty Languages: The Choreographic and Multitier Cases (Pearl)},
  booktitle    = {35th European Conference on Object-Oriented Programming, {ECOOP} 2021,
                  July 11-17, 2021, Aarhus, Denmark (Virtual Conference)},
  series       = {LIPIcs},
  volume       = {194},
  pages        = {22:1--22:27},
  publisher    = {Schloss Dagstuhl - Leibniz-Zentrum f{\"{u}}r Informatik},
  year         = {2021},
  url          = {https://doi.org/10.4230/LIPIcs.ECOOP.2021.22},
  doi          = {10.4230/LIPICS.ECOOP.2021.22},
  bibsource    = {dblp computer science bibliography, https://dblp.org}
}

@article{DBLP:journals/pacmpl/HirschG22,
  author       = {Andrew K. Hirsch and
                  Deepak Garg},
  title        = {Pirouette: higher-order typed functional choreographies},
  journal      = {Proc. {ACM} Program. Lang.},
  volume       = {6},
  number       = {{POPL}},
  pages        = {1--27},
  year         = {2022},
  url          = {https://doi.org/10.1145/3498684},
  doi          = {10.1145/3498684},
  bibsource    = {dblp computer science bibliography, https://dblp.org}
}

@article{DBLP:journals/corr/abs-2111-03484,
  author       = {Andrew K. Hirsch and
                  Deepak Garg},
  title        = {Pirouette: Higher-Order Typed Functional Choreographies},
  journal      = {CoRR},
  volume       = {abs/2111.03484},
  year         = {2021},
  url          = {https://arxiv.org/abs/2111.03484},
  eprinttype    = {arXiv},
  eprint       = {2111.03484},
  bibsource    = {dblp computer science bibliography, https://dblp.org}
}

@article{DBLP:journals/pacmpl/HinrichsenBK20,
  author       = {Jonas Kastberg Hinrichsen and
                  Jesper Bengtson and
                  Robbert Krebbers},
  title        = {Actris: session-type based reasoning in separation logic},
  journal      = {Proc. {ACM} Program. Lang.},
  volume       = {4},
  number       = {{POPL}},
  pages        = {6:1--6:30},
  year         = {2020},
  url          = {https://doi.org/10.1145/3371074},
  doi          = {10.1145/3371074},
  bibsource    = {dblp computer science bibliography, https://dblp.org}
}

@article{DBLP:journals/lmcs/HinrichsenBK22,
  author       = {Jonas Kastberg Hinrichsen and
                  Jesper Bengtson and
                  Robbert Krebbers},
  title        = {Actris 2.0: Asynchronous Session-Type Based Reasoning in Separation
                  Logic},
  journal      = {Log. Methods Comput. Sci.},
  volume       = {18},
  number       = {2},
  year         = {2022},
  url          = {https://doi.org/10.46298/lmcs-18(2:16)2022},
  doi          = {10.46298/LMCS-18(2:16)2022},
  bibsource    = {dblp computer science bibliography, https://dblp.org}
}

@article{DBLP:journals/corr/abs-2507-17354,
	author       = {Cinzia Di Giusto and
	{\'{E}}tienne Lozes and
	Pascal Urso},
	title        = {Realisability and Complementability of Multiparty Session Types},
	journal      = {CoRR},
	volume       = {abs/2507.17354},
	year         = {2025},
	url          = {https://doi.org/10.48550/arXiv.2507.17354},
	doi          = {10.48550/ARXIV.2507.17354},
	eprinttype    = {arXiv},
	eprint       = {2507.17354},
	bibsource    = {dblp computer science bibliography, https://dblp.org}
}

@misc{Li25OopslaArxiv,
	title={Characterizing Implementability of Global Protocols with Infinite States and Data}, 
	author={Elaine Li and Felix Stutz and Thomas Wies and Damien Zufferey},
	year={2025},
	eprint={2411.05722},
	archivePrefix={arXiv},
	primaryClass={cs.PL},
	url={https://arxiv.org/abs/2411.05722}, 
}

@article{DBLP:journals/scp/EngelsMR02,
  author       = {Andr{\'{e}} Engels and
                  Sjouke Mauw and
                  Michel A. Reniers},
  title        = {A hierarchy of communication models for Message Sequence Charts},
  journal      = {Sci. Comput. Program.},
  volume       = {44},
  number       = {3},
  pages        = {253--292},
  year         = {2002},
  url          = {https://doi.org/10.1016/S0167-6423(02)00022-9},
  doi          = {10.1016/S0167-6423(02)00022-9},
  bibsource    = {dblp computer science bibliography, https://dblp.org}
}

@inproceedings{DBLP:conf/concur/ClementeHS14,
  author       = {Lorenzo Clemente and
                  Fr{\'{e}}d{\'{e}}ric Herbreteau and
                  Gr{\'{e}}goire Sutre},
  title        = {Decidable Topologies for Communicating Automata with {FIFO} and Bag
                  Channels},
  booktitle    = {{CONCUR} 2014 - Concurrency Theory - 25th International Conference,
                  {CONCUR} 2014, Rome, Italy, September 2-5, 2014. Proceedings},
  pages        = {281--296},
  volume       = {8704},
  publisher    = {Springer},
  year         = {2014},
  url          = {https://doi.org/10.1007/978-3-662-44584-6\_20},
  doi          = {10.1007/978-3-662-44584-6\_20},
  bibsource    = {dblp computer science bibliography, https://dblp.org}
}

@inproceedings{DBLP:conf/cav/LiSWZ25,
	author       = {Elaine Li and
	Felix Stutz and
	Thomas Wies and
	Damien Zufferey},
	editor       = {Ruzica Piskac and
	Zvonimir Rakamaric},
	title        = {Sprout: {A} Verifier for Symbolic Multiparty Protocols},
	booktitle    = {Computer Aided Verification - 37th International Conference, {CAV}
	2025, Zagreb, Croatia, July 23-25, 2025, Proceedings, Part {III}},
	series       = {Lecture Notes in Computer Science},
	volume       = {15933},
	pages        = {304--317},
	publisher    = {Springer},
	year         = {2025},
	url          = {https://doi.org/10.1007/978-3-031-98682-6_16},
	doi          = {10.1007/978-3-031-98682-6_16},
	bibsource    = {dblp computer science bibliography, https://dblp.org}
}

@inproceedings{DBLP:conf/itp/LiW25,
	author       = {Elaine Li and
	Thomas Wies},
	editor       = {Yannick Forster and
	Chantal Keller},
	title        = {Certified Implementability of Global Multiparty Protocols},
	booktitle    = {16th International Conference on Interactive Theorem Proving, {ITP}
	2025, September 28 to October 1, 2025, Reykjavik, Iceland},
	series       = {LIPIcs},
	volume       = {352},
	pages        = {15:1--15:20},
	publisher    = {Schloss Dagstuhl - Leibniz-Zentrum f{\"{u}}r Informatik},
	year         = {2025},
	url          = {https://doi.org/10.4230/LIPIcs.ITP.2025.15},
	doi          = {10.4230/LIPICS.ITP.2025.15},
	bibsource    = {dblp computer science bibliography, https://dblp.org}
}

@article{DBLP:journals/toplas/GallagerHS83,
  author       = {Robert G. Gallager and
                  Pierre A. Humblet and
                  Philip M. Spira},
  title        = {A Distributed Algorithm for Minimum-Weight Spanning Trees},
  journal      = {{ACM} Trans. Program. Lang. Syst.},
  volume       = {5},
  number       = {1},
  pages        = {66--77},
  year         = {1983},
  url          = {https://doi.org/10.1145/357195.357200},
  doi          = {10.1145/357195.357200},
  bibsource    = {dblp computer science bibliography, https://dblp.org}
}

@incollection{DBLP:books/acm/19/Lamport19b,
  author       = {Leslie Lamport},
  title        = {Time, clocks, and the ordering of events in a distributed system},
  booktitle    = {Concurrency: the Works of Leslie Lamport},
  pages        = {179--196},
  year         = {2019},
  crossref     = {DBLP:books/acm/19/2019M},
  url          = {https://doi.org/10.1145/3335772.3335934},
  doi          = {10.1145/3335772.3335934},
  bibsource    = {dblp computer science bibliography, https://dblp.org}
}

@book{DBLP:books/acm/19/2019M,
  editor       = {Dahlia Malkhi},
  title        = {Concurrency: the Works of Leslie Lamport},
  publisher    = {{ACM}},
  year         = {2019},
  url          = {https://doi.org/10.1145/3335772},
  doi          = {10.1145/3335772},
  isbn         = {978-1-4503-7270-1},
  bibsource    = {dblp computer science bibliography, https://dblp.org}
}

@article{DBLP:journals/pacmpl/JacobsHK23,
  author       = {Jules Jacobs and
                  Jonas Kastberg Hinrichsen and
                  Robbert Krebbers},
  title        = {Dependent Session Protocols in Separation Logic from First Principles
                  (Functional Pearl)},
  journal      = {Proc. {ACM} Program. Lang.},
  volume       = {7},
  number       = {{ICFP}},
  pages        = {768--795},
  year         = {2023},
  url          = {https://doi.org/10.1145/3607856},
  doi          = {10.1145/3607856},
  bibsource    = {dblp computer science bibliography, https://dblp.org}
}

@article{DBLP:journals/jlp/GhilezanJPSY19,
	author       = {Silvia Ghilezan and
	Svetlana Jaksic and
	Jovanka Pantovic and
	Alceste Scalas and
	Nobuko Yoshida},
	title        = {Precise subtyping for synchronous multiparty sessions},
	journal      = {J. Log. Algebraic Methods Program.},
	volume       = {104},
	pages        = {127--173},
	year         = {2019},
	url          = {https://doi.org/10.1016/j.jlamp.2018.12.002},
	doi          = {10.1016/J.JLAMP.2018.12.002},
	bibsource    = {dblp computer science bibliography, https://dblp.org}
}

@inproceedings{DBLP:conf/ppdp/ChenDY24,
	author       = {Tzu{-}Chun Chen and
	Mariangiola Dezani{-}Ciancaglini and
	Nobuko Yoshida},
	editor       = {Alessandro Bruni and
	Alberto Momigliano and
	Matteo Pradella and
	Matteo Rossi and
	James Cheney},
	title        = {On the Preciseness of Subtyping in Session Types: 10 Years Later},
	booktitle    = {Proceedings of the 26th International Symposium on Principles and
	Practice of Declarative Programming, {PPDP} 2024, Milano, Italy, September
	9-11, 2024},
	pages        = {2:1--2:3},
	publisher    = {{ACM}},
	year         = {2024},
	url          = {https://doi.org/10.1145/3678232.3678258},
	doi          = {10.1145/3678232.3678258},
	bibsource    = {dblp computer science bibliography, https://dblp.org}
}

@inproceedings{DBLP:conf/lics/PetersY24,
	author       = {Kirstin Peters and
	Nobuko Yoshida},
	editor       = {Pawel Sobocinski and
	Ugo Dal Lago and
	Javier Esparza},
	title        = {Separation and Encodability in Mixed Choice Multiparty Sessions},
	booktitle    = {Proceedings of the 39th Annual {ACM/IEEE} Symposium on Logic in Computer
	Science, {LICS} 2024, Tallinn, Estonia, July 8-11, 2024},
	pages        = {62:1--62:15},
	publisher    = {{ACM}},
	year         = {2024},
	url          = {https://doi.org/10.1145/3661814.3662085},
	doi          = {10.1145/3661814.3662085},
	bibsource    = {dblp computer science bibliography, https://dblp.org}
}

@article{DBLP:journals/pacmpl/UdomsrirungruangY25,
	author       = {Thien Udomsrirungruang and
	Nobuko Yoshida},
	title        = {Top-Down or Bottom-Up? Complexity Analyses of Synchronous Multiparty
	Session Types},
	journal      = {Proc. {ACM} Program. Lang.},
	volume       = {9},
	number       = {{POPL}},
	pages        = {1040--1071},
	year         = {2025},
	url          = {https://doi.org/10.1145/3704872},
	doi          = {10.1145/3704872},
	bibsource    = {dblp computer science bibliography, https://dblp.org}
}

@inproceedings{DBLP:conf/forte/Cruz-FilipeM16,
	author       = {Lu{\'{\i}}s Cruz{-}Filipe and
	Fabrizio Montesi},
	editor       = {Elvira Albert and
	Ivan Lanese},
	title        = {Choreographies in Practice},
	booktitle    = {Formal Techniques for Distributed Objects, Components, and Systems
	- 36th {IFIP} {WG} 6.1 International Conference, {FORTE} 2016, Held
	as Part of the 11th International Federated Conference on Distributed
	Computing Techniques, DisCoTec 2016, Heraklion, Crete, Greece, June
	6-9, 2016, Proceedings},
	series       = {Lecture Notes in Computer Science},
	volume       = {9688},
	pages        = {114--123},
	publisher    = {Springer},
	year         = {2016},
	url          = {https://doi.org/10.1007/978-3-319-39570-8\_8},
	doi          = {10.1007/978-3-319-39570-8\_8},
	bibsource    = {dblp computer science bibliography, https://dblp.org}
}

@misc{depalma2024functionasaservicechoreographicprogramminglanguage,
	title={Towards a Function-as-a-Service Choreographic Programming Language: Examples and Applications}, 
	author={Giuseppe De Palma and Saverio Giallorenzo and Jacopo Mauro and Matteo Trentin and Gianluigi Zavattaro},
	year={2024},
	eprint={2406.09099},
	archivePrefix={arXiv},
	primaryClass={cs.PL},
	url={https://arxiv.org/abs/2406.09099}, 
}

@misc{kashiwa2023portableefficientpracticallibrarylevel,
	title={Portable, Efficient, and Practical Library-Level Choreographic Programming}, 
	author={Shun Kashiwa and Gan Shen and Soroush Zare and Lindsey Kuper},
	year={2023},
	eprint={2311.11472},
	archivePrefix={arXiv},
	primaryClass={cs.PL},
	url={https://arxiv.org/abs/2311.11472}, 
}

@misc{chorus,
	title = {ChoRus: Choreographic Programming in Rust},
	author = {{Languages, Systems, and Data Lab, UC Santa Cruz}},
	howpublished = {\url{https://lsd-ucsc.github.io/ChoRus/introduction.html}},
	note = {Accessed: 2025-07-10}
}

@article{DBLP:journals/toplas/GiallorenzoMP24,
	author       = {Saverio Giallorenzo and
	Fabrizio Montesi and
	Marco Peressotti},
	title        = {Choral: Object-oriented Choreographic Programming},
	journal      = {{ACM} Trans. Program. Lang. Syst.},
	volume       = {46},
	number       = {1},
	pages        = {1:1--1:59},
	year         = {2024},
	url          = {https://doi.org/10.1145/3632398},
	doi          = {10.1145/3632398},
	bibsource    = {dblp computer science bibliography, https://dblp.org}
}

@inproceedings{DBLP:conf/sp/GancherGSDP23,
	author       = {Joshua Gancher and
	Sydney Gibson and
	Pratap Singh and
	Samvid Dharanikota and
	Bryan Parno},
	title        = {Owl: Compositional Verification of Security Protocols via an Information-Flow
	Type System},
	booktitle    = {44th {IEEE} Symposium on Security and Privacy, {SP} 2023, San Francisco,
	CA, USA, May 21-25, 2023},
	pages        = {1130--1147},
	publisher    = {{IEEE}},
	year         = {2023},
	url          = {https://doi.org/10.1109/SP46215.2023.10179477},
	doi          = {10.1109/SP46215.2023.10179477},
	bibsource    = {dblp computer science bibliography, https://dblp.org}
}

@techreport{w3c-ws-cdl-2005,
	title        = {Web Services Choreography Description Language Version 1.0},
	institution  = {World Wide Web Consortium (W3C)},
	author       = {{Web Services Choreography Working Group}},
	year         = {2005},
	month        = {November},
	type         = {W3C Candidate Recommendation},
	note         = {Available at http://www.w3.org/TR/2005/CR-ws-cdl-10-20051109/},
}

@misc{umlwebsite,
	author       = "{Object Management Group}",
	title        = "{Unified Modeling Language (UML) Website}",
	howpublished = "\url{https://www.uml.org/}",
	note         = "Accessed: 2025-07-10"
}

@techreport{ITU-T-Z120-2011,
	author      = "{International Telecommunication Union}",
	title       = "{ITU-T Recommendation Z.120: Message Sequence Chart (MSC)}",
	institution = "{International Telecommunication Union}",
	address     = "Geneva",
	year        = "2011",
	month       = "February",
	type        = "{ITU-T Recommendation}",
	number      = "Z.120",
	url         = "https://www.itu.int/rec/T-REC-Z.120-201102-I/en"
}

@article{DBLP:journals/iandc/GenestKM06,
	author       = {Blaise Genest and
	Dietrich Kuske and
	Anca Muscholl},
	title        = {A Kleene theorem and model checking algorithms for existentially bounded
	communicating automata},
	journal      = {Inf. Comput.},
	volume       = {204},
	number       = {6},
	pages        = {920--956},
	year         = {2006},
	url          = {https://doi.org/10.1016/j.ic.2006.01.005},
	doi          = {10.1016/J.IC.2006.01.005},
	bibsource    = {dblp computer science bibliography, https://dblp.org}
}

@inproceedings{DBLP:conf/cav/BouajjaniEJQ18,
	author       = {Ahmed Bouajjani and
	Constantin Enea and
	Kailiang Ji and
	Shaz Qadeer},
	editor       = {Hana Chockler and
	Georg Weissenbacher},
	title        = {On the Completeness of Verifying Message Passing Programs Under Bounded
	Asynchrony},
	booktitle    = {Computer Aided Verification - 30th International Conference, {CAV}
	2018, Held as Part of the Federated Logic Conference, FloC 2018, Oxford,
	UK, July 14-17, 2018, Proceedings, Part {II}},
	series       = {Lecture Notes in Computer Science},
	volume       = {10982},
	pages        = {372--391},
	publisher    = {Springer},
	year         = {2018},
	url          = {https://doi.org/10.1007/978-3-319-96142-2\_23},
	doi          = {10.1007/978-3-319-96142-2\_23},
	bibsource    = {dblp computer science bibliography, https://dblp.org}
}

@inproceedings{DBLP:conf/concur/BolligGFLLS21,
	author       = {Benedikt Bollig and
	Cinzia Di Giusto and
	Alain Finkel and
	Laetitia Laversa and
	{\'{E}}tienne Lozes and
	Amrita Suresh},
	editor       = {Serge Haddad and
	Daniele Varacca},
	title        = {A Unifying Framework for Deciding Synchronizability},
	booktitle    = {32nd International Conference on Concurrency Theory, {CONCUR} 2021,
	August 24-27, 2021, Virtual Conference},
	series       = {LIPIcs},
	volume       = {203},
	pages        = {14:1--14:18},
	publisher    = {Schloss Dagstuhl - Leibniz-Zentrum f{\"{u}}r Informatik},
	year         = {2021},
	url          = {https://doi.org/10.4230/LIPIcs.CONCUR.2021.14},
	doi          = {10.4230/LIPICS.CONCUR.2021.14},
	bibsource    = {dblp computer science bibliography, https://dblp.org}
}

@inproceedings{DBLP:conf/fsttcs/AkshayDGS13,
	author       = {S. Akshay and
	Ionut Dinca and
	Blaise Genest and
	Alin Stefanescu},
	editor       = {Anil Seth and
	Nisheeth K. Vishnoi},
	title        = {Implementing Realistic Asynchronous Automata},
	booktitle    = {{IARCS} Annual Conference on Foundations of Software Technology and
	Theoretical Computer Science, {FSTTCS} 2013, December 12-14, 2013,
	Guwahati, India},
	series       = {LIPIcs},
	volume       = {24},
	pages        = {213--224},
	publisher    = {Schloss Dagstuhl - Leibniz-Zentrum f{\"{u}}r Informatik},
	year         = {2013},
	url          = {https://doi.org/10.4230/LIPIcs.FSTTCS.2013.213},
	doi          = {10.4230/LIPICS.FSTTCS.2013.213},
	bibsource    = {dblp computer science bibliography, https://dblp.org}
}

@article{Li25OopslaOfficial,
	author    = {Elaine Li and Felix Stutz and Thomas Wies and Damien Zufferey},
	title     = {Characterizing Implementability of Global Protocols with Infinite States and Data},
	year      = {2025},
	publisher = {ACM},
	journal = {PACMPL},
	number = {Object-oriented Programming, Systems, Languages, and Applications (OOPSLA)},
	volume    = {9},
}

@article{DBLP:journals/pacmpl/GiustoFLL23,
	author       = {Cinzia Di Giusto and
	Davide Ferr{\'{e}} and
	Laetitia Laversa and
	{\'{E}}tienne Lozes},
	title        = {A Partial Order View of Message-Passing Communication Models},
	journal      = {Proc. {ACM} Program. Lang.},
	volume       = {7},
	number       = {{POPL}},
	pages        = {1601--1627},
	year         = {2023},
	url          = {https://doi.org/10.1145/3571248},
	doi          = {10.1145/3571248},
	bibsource    = {dblp computer science bibliography, https://dblp.org}
}

@inproceedings{DBLP:conf/concur/DelpyMS24,
	author       = {Romain Delpy and
	Anca Muscholl and
	Gr{\'{e}}goire Sutre},
	editor       = {Rupak Majumdar and
	Alexandra Silva},
	title        = {An Automata-Based Approach for Synchronizable Mailbox Communication},
	booktitle    = {35th International Conference on Concurrency Theory, {CONCUR} 2024,
	September 9-13, 2024, Calgary, Canada},
	series       = {LIPIcs},
	volume       = {311},
	pages        = {22:1--22:19},
	publisher    = {Schloss Dagstuhl - Leibniz-Zentrum f{\"{u}}r Informatik},
	year         = {2024},
	url          = {https://doi.org/10.4230/LIPIcs.CONCUR.2024.22},
	doi          = {10.4230/LIPICS.CONCUR.2024.22},
	bibsource    = {dblp computer science bibliography, https://dblp.org}
}

@inproceedings{DBLP:conf/fm/ChevrouH0Q19,
	author       = {Florent Chevrou and
	Aur{\'{e}}lie Hurault and
	Shin Nakajima and
	Philippe Qu{\'{e}}innec},
	editor       = {Emil Sekerinski and
	Nelma Moreira and
	Jos{\'{e}} N. Oliveira and
	Daniel Ratiu and
	Riccardo Guidotti and
	Marie Farrell and
	Matt Luckcuck and
	Diego Marmsoler and
	Jos{\'{e}} Creissac Campos and
	Troy Astarte and
	Laure Gonnord and
	Antonio Cerone and
	Luis Couto and
	Brijesh Dongol and
	Martin Kutrib and
	Pedro Monteiro and
	David Delmas},
	title        = {A Map of Asynchronous Communication Models},
	booktitle    = {Formal Methods. {FM} 2019 International Workshops - Porto, Portugal,
	October 7-11, 2019, Revised Selected Papers, Part {II}},
	series       = {Lecture Notes in Computer Science},
	volume       = {12233},
	pages        = {307--322},
	publisher    = {Springer},
	year         = {2019},
	url          = {https://doi.org/10.1007/978-3-030-54997-8\_20},
	doi          = {10.1007/978-3-030-54997-8\_20},
	bibsource    = {dblp computer science bibliography, https://dblp.org}
}

@article{DBLP:journals/fac/ChevrouHQ16,
	author       = {Florent Chevrou and
	Aur{\'{e}}lie Hurault and
	Philippe Qu{\'{e}}innec},
	title        = {On the diversity of asynchronous communication},
	journal      = {Formal Aspects Comput.},
	volume       = {28},
	number       = {5},
	pages        = {847--879},
	year         = {2016},
	url          = {https://doi.org/10.1007/s00165-016-0379-x},
	doi          = {10.1007/S00165-016-0379-X},
	bibsource    = {dblp computer science bibliography, https://dblp.org}
}

@article{DBLP:journals/fac/NeykovaBY17,
	author       = {Rumyana Neykova and
	Laura Bocchi and
	Nobuko Yoshida},
	title        = {Timed runtime monitoring for multiparty conversations},
	journal      = {Formal Aspects Comput.},
	volume       = {29},
	number       = {5},
	pages        = {877--910},
	year         = {2017},
	url          = {https://doi.org/10.1007/s00165-017-0420-8},
	doi          = {10.1007/S00165-017-0420-8},
	bibsource    = {dblp computer science bibliography, https://dblp.org}
}

@article{DBLP:journals/corr/NeykovaY16,
	author       = {Rumyana Neykova and
	Nobuko Yoshida},
	title        = {Multiparty Session Actors},
	journal      = {Log. Methods Comput. Sci.},
	volume       = {13},
	number       = {1},
	year         = {2017},
	url          = {https://doi.org/10.23638/LMCS-13(1:17)2017},
	doi          = {10.23638/LMCS-13(1:17)2017},
	bibsource    = {dblp computer science bibliography, https://dblp.org}
}

@InProceedings{demuijnckhughes_et_al:LIPIcs.ECOOP.2019.6,
	author =	{de Muijnck-Hughes, Jan and Vanderbauwhede, Wim},
	title =	{{A Typing Discipline for Hardware Interfaces}},
	booktitle =	{33rd European Conference on Object-Oriented Programming (ECOOP 2019)},
	pages =	{6:1--6:27},
	series =	{Leibniz International Proceedings in Informatics (LIPIcs)},
	ISBN =	{978-3-95977-111-5},
	ISSN =	{1868-8969},
	year =	{2019},
	volume =	{134},
	editor =	{Donaldson, Alastair F.},
	publisher =	{Schloss Dagstuhl -- Leibniz-Zentrum f{\"u}r Informatik},
	address =	{Dagstuhl, Germany},
	URL =		{https://drops.dagstuhl.de/entities/document/10.4230/LIPIcs.ECOOP.2019.6},
	URN =		{urn:nbn:de:0030-drops-107983},
	doi =		{10.4230/LIPIcs.ECOOP.2019.6}
}

@inproceedings{DBLP:conf/fpl/NiuNYWYL16,
	author       = {Xinyu Niu and
	Nicholas Ng and
	Tomofumi Yuki and
	Shaojun Wang and
	Nobuko Yoshida and
	Wayne Luk},
	editor       = {Paolo Ienne and
	Walid A. Najjar and
	Jason Helge Anderson and
	Philip Brisk and
	Walter Stechele},
	title        = {{EURECA} compilation: Automatic optimisation of cycle-reconfigurable
	circuits},
	booktitle    = {26th International Conference on Field Programmable Logic and Applications,
	{FPL} 2016, Lausanne, Switzerland, August 29 - September 2, 2016},
	pages        = {1--4},
	publisher    = {{IEEE}},
	year         = {2016},
	url          = {https://doi.org/10.1109/FPL.2016.7577359},
	doi          = {10.1109/FPL.2016.7577359},
	bibsource    = {dblp computer science bibliography, https://dblp.org}
}

@article{DBLP:journals/pacmpl/MajumdarYZ20,
	author       = {Rupak Majumdar and
	Nobuko Yoshida and
	Damien Zufferey},
	title        = {Multiparty motion coordination: from choreographies to robotics programs},
	journal      = {Proc. {ACM} Program. Lang.},
	volume       = {4},
	number       = {{OOPSLA}},
	pages        = {134:1--134:30},
	year         = {2020},
	url          = {https://doi.org/10.1145/3428202},
	doi          = {10.1145/3428202},
	bibsource    = {dblp computer science bibliography, https://dblp.org}
}

@inproceedings{DBLP:conf/ecoop/Castro-PerezY23,
	author       = {David Castro{-}Perez and
	Nobuko Yoshida},
	editor       = {Karim Ali and
	Guido Salvaneschi},
	title        = {Dynamically Updatable Multiparty Session Protocols: Generating Concurrent
	Go Code from Unbounded Protocols},
	booktitle    = {37th European Conference on Object-Oriented Programming, {ECOOP} 2023,
	July 17-21, 2023, Seattle, Washington, United States},
	series       = {LIPIcs},
	volume       = {263},
	pages        = {6:1--6:30},
	publisher    = {Schloss Dagstuhl - Leibniz-Zentrum f{\"{u}}r Informatik},
	year         = {2023},
	url          = {https://doi.org/10.4230/LIPIcs.ECOOP.2023.6},
	doi          = {10.4230/LIPICS.ECOOP.2023.6},
	bibsource    = {dblp computer science bibliography, https://dblp.org}
}

@article{DBLP:journals/pacmpl/CastroHJNY19,
	author       = {David Castro{-}Perez and
	Raymond Hu and
	Sung{-}Shik Jongmans and
	Nicholas Ng and
	Nobuko Yoshida},
	title        = {Distributed programming using role-parametric session types in go:
	statically-typed endpoint APIs for dynamically-instantiated communication
	structures},
	journal      = {Proc. {ACM} Program. Lang.},
	volume       = {3},
	number       = {{POPL}},
	pages        = {29:1--29:30},
	year         = {2019},
	url          = {https://doi.org/10.1145/3290342},
	doi          = {10.1145/3290342},
	bibsource    = {dblp computer science bibliography, https://dblp.org}
}

@inproceedings{DBLP:conf/ecoop/ImaiNYY19,
	author       = {Keigo Imai and
	Rumyana Neykova and
	Nobuko Yoshida and
	Shoji Yuen},
	editor       = {Robert Hirschfeld and
	Tobias Pape},
	title        = {Multiparty Session Programming With Global Protocol Combinators},
	booktitle    = {34th European Conference on Object-Oriented Programming, {ECOOP} 2020,
	November 15-17, 2020, Berlin, Germany (Virtual Conference)},
	series       = {LIPIcs},
	volume       = {166},
	pages        = {9:1--9:30},
	publisher    = {Schloss Dagstuhl - Leibniz-Zentrum f{\"{u}}r Informatik},
	year         = {2020},
	url          = {https://doi.org/10.4230/LIPIcs.ECOOP.2020.9},
	doi          = {10.4230/LIPICS.ECOOP.2020.9},
	bibsource    = {dblp computer science bibliography, https://dblp.org}
}

@article{DBLP:journals/darts/LagaillardieNY22,
	author       = {Nicolas Lagaillardie and
	Rumyana Neykova and
	Nobuko Yoshida},
	title        = {Stay Safe Under Panic: Affine Rust Programming with Multiparty Session
	Types (Artifact)},
	journal      = {Dagstuhl Artifacts Ser.},
	volume       = {8},
	number       = {2},
	pages        = {09:1--09:16},
	year         = {2022},
	url          = {https://doi.org/10.4230/DARTS.8.2.9},
	doi          = {10.4230/DARTS.8.2.9},
	bibsource    = {dblp computer science bibliography, https://dblp.org}
}

@inproceedings{DBLP:conf/tools/NgYH12,
	author       = {Nicholas Ng and
	Nobuko Yoshida and
	Kohei Honda},
	editor       = {Carlo A. Furia and
	Sebastian Nanz},
	title        = {Multiparty Session {C:} Safe Parallel Programming with Message Optimisation},
	booktitle    = {Objects, Models, Components, Patterns - 50th International Conference,
	{TOOLS} 2012, Prague, Czech Republic, May 29-31, 2012. Proceedings},
	series       = {Lecture Notes in Computer Science},
	volume       = {7304},
	pages        = {202--218},
	publisher    = {Springer},
	year         = {2012},
	url          = {https://doi.org/10.1007/978-3-642-30561-0\_15},
	doi          = {10.1007/978-3-642-30561-0\_15},
	bibsource    = {dblp computer science bibliography, https://dblp.org}
}

@inproceedings{DBLP:conf/fase/HuY17,
	author       = {Raymond Hu and
	Nobuko Yoshida},
	editor       = {Marieke Huisman and
	Julia Rubin},
	title        = {Explicit Connection Actions in Multiparty Session Types},
	booktitle    = {Fundamental Approaches to Software Engineering - 20th International
	Conference, {FASE} 2017, Held as Part of the European Joint Conferences
	on Theory and Practice of Software, {ETAPS} 2017, Uppsala, Sweden,
	April 22-29, 2017, Proceedings},
	series       = {Lecture Notes in Computer Science},
	volume       = {10202},
	pages        = {116--133},
	publisher    = {Springer},
	year         = {2017},
	url          = {https://doi.org/10.1007/978-3-662-54494-5\_7},
	doi          = {10.1007/978-3-662-54494-5\_7},
	bibsource    = {dblp computer science bibliography, https://dblp.org}
}

@inproceedings{DBLP:conf/fase/HuY16,
	author       = {Raymond Hu and
	Nobuko Yoshida},
	editor       = {Perdita Stevens and
	Andrzej Wasowski},
	title        = {Hybrid Session Verification Through Endpoint {API} Generation},
	booktitle    = {Fundamental Approaches to Software Engineering - 19th International
	Conference, {FASE} 2016, Held as Part of the European Joint Conferences
	on Theory and Practice of Software, {ETAPS} 2016, Eindhoven, The Netherlands,
	April 2-8, 2016, Proceedings},
	series       = {Lecture Notes in Computer Science},
	volume       = {9633},
	pages        = {401--418},
	publisher    = {Springer},
	year         = {2016},
	url          = {https://doi.org/10.1007/978-3-662-49665-7\_24},
	doi          = {10.1007/978-3-662-49665-7\_24},
	bibsource    = {dblp computer science bibliography, https://dblp.org}
}

@book{M23,
	author={Montesi, Fabrizio},
	title={Introduction to Choreographies},
	place={Cambridge},
	doi={10.1017/9781108981491},
	publisher={Cambridge University Press},
	year={2023}
}

@incollection{DBLP:books/sp/24/Yoshida24,
	author       = {Nobuko Yoshida},
	editor       = {Frank S. de Boer and
	Ferruccio Damiani and
	Reiner H{\"{a}}hnle and
	Einar Broch Johnsen and
	Eduard Kamburjan},
	title        = {Programming Language Implementations with Multiparty Session Types},
	booktitle    = {Active Object Languages: Current Research Trends},
	series       = {Lecture Notes in Computer Science},
	volume       = {14360},
	pages        = {147--165},
	publisher    = {Springer},
	year         = {2024},
	url          = {https://doi.org/10.1007/978-3-031-51060-1\_6},
	doi          = {10.1007/978-3-031-51060-1\_6},
	bibsource    = {dblp computer science bibliography, https://dblp.org}
}

@inproceedings{DBLP:conf/ppopp/CutnerYV22,
	author       = {Zak Cutner and
	Nobuko Yoshida and
	Martin Vassor},
	editor       = {Jaejin Lee and
	Kunal Agrawal and
	Michael F. Spear},
	title        = {Deadlock-free asynchronous message reordering in rust with multiparty
	session types},
	booktitle    = {PPoPP '22: 27th {ACM} {SIGPLAN} Symposium on Principles and Practice
	of Parallel Programming, Seoul, Republic of Korea, April 2 - 6, 2022},
	pages        = {246--261},
	publisher    = {{ACM}},
	year         = {2022},
	url          = {https://doi.org/10.1145/3503221.3508404},
	doi          = {10.1145/3503221.3508404},
	bibsource    = {dblp computer science bibliography, https://dblp.org}
}

@inproceedings{DBLP:conf/ecoop/VassorY24,
	author       = {Martin Vassor and
	Nobuko Yoshida},
	editor       = {Jonathan Aldrich and
	Guido Salvaneschi},
	title        = {Refinements for Multiparty Message-Passing Protocols: Specification-Agnostic
	Theory and Implementation},
	booktitle    = {38th European Conference on Object-Oriented Programming, {ECOOP} 2024,
	September 16-20, 2024, Vienna, Austria},
	series       = {LIPIcs},
	volume       = {313},
	pages        = {41:1--41:29},
	publisher    = {Schloss Dagstuhl - Leibniz-Zentrum f{\"{u}}r Informatik},
	year         = {2024},
	url          = {https://doi.org/10.4230/LIPIcs.ECOOP.2024.41},
	doi          = {10.4230/LIPICS.ECOOP.2024.41},
	bibsource    = {dblp computer science bibliography, https://dblp.org}
}

@article{DBLP:journals/fmsd/DemangeonHHNY15,
	author       = {Romain Demangeon and
	Kohei Honda and
	Raymond Hu and
	Rumyana Neykova and
	Nobuko Yoshida},
	title        = {Practical interruptible conversations: distributed dynamic verification
	with multiparty session types and Python},
	journal      = {Formal Methods Syst. Des.},
	volume       = {46},
	number       = {3},
	pages        = {197--225},
	year         = {2015},
	url          = {https://doi.org/10.1007/s10703-014-0218-8},
	doi          = {10.1007/S10703-014-0218-8},
	bibsource    = {dblp computer science bibliography, https://dblp.org}
}

@article{DBLP:journals/ita/Zielonka87,
	author       = {Wieslaw Zielonka},
	title        = {Notes on Finite Asynchronous Automata},
	journal      = {{RAIRO} Theor. Informatics Appl.},
	volume       = {21},
	number       = {2},
	pages        = {99--135},
	year         = {1987},
	url          = {https://doi.org/10.1051/ita/1987210200991},
	doi          = {10.1051/ITA/1987210200991},
	bibsource    = {dblp computer science bibliography, https://dblp.org}
}

@unpublished{Multris,
  author       = {Jonas Kastberg Hinrichsen and
                  Jules Jacobs and
                  Robbert Krebbers},
  title        = {Multris: Functional Verification of Multiparty Message
                  Passing in Separation Logic},
  url          = {https://jihgfee.github.io/papers/multris_manuscript.pdf},
  year   = 2024
}

@phdthesis{DBLP:phd/dnb/Stutz24,
  author       = {Felix Stutz},
  title        = {Implementability of Asynchronous Communication Protocols - The Power
                  of Choice},
  school       = {Kaiserslautern University of Technology, Germany},
  year         = {2024},
  url          = {https://kluedo.ub.rptu.de/frontdoor/index/index/docId/8077},
  urn          = {urn:nbn:de:hbz:386-kluedo-80778},
  bibsource    = {dblp computer science bibliography, https://dblp.org}
}

\clearpage

\appendix
\section{Appendix}
\label{appendix}

\subsection{Proofs for \cref{sec:channel-compliance}}
\label{sec:lem:model-entails-ccfacts-proof}

In this section, we prove \cref{lem:model-entails-ccfacts}. Let $\netarch \in \netarchs$ be a network architecture that satisfies our axiomatic network model. We need to prove that $\netarch$ satisfies \ccfact{1} through \ccfact{6}. For convenience, we state these facts again here in one place:

\begin{definition}[Channel compliance facts]
	\label{def:cc-facts} 
	Let $w \in \AlphAsyncSubscript^*$ be channel-compliant. 
	\begin{enumerate}[\ccfact{\arabic*}]
		\item For all $x \in \AlphAsync_!$, $wx$ is channel-compliant. \label{ca:send-enabled}
		\item For all $\procA \neq \procB \in \Procs$ and $\val \in \MsgVals$, $\card {w \wproj_{\snd{\procA}{\procB}{\val}}} \geq \card {w \wproj_{\rcv{\procA}{\procB}{\val}}}$. \label{ca:send-before-receive}
		\item For all $\rho \in \AlphSyncSubscript^*$, $\SyncToAsync(\rho)$ is channel-compliant. \label{ca:sync}
                \item For all $\rho \in \AlphSyncSubscript^*$, if 
                  $w \equiv_\Procs \SyncToAsync(\rho)$, then $w$ is channel-matched. \label{ca:sync-matched}
		\item For all $x \in \AlphAsync_!$, $y \in \AlphAsync_?$, if $wy$ is channel-compliant then $wxy$ is channel-compliant. \label{ca:intro-send}
		\item Let $\alpha, \beta \in \AlphSyncSubscript^*, \procA \neq \procB \in \Procs$ and $\val \in \MsgVals$ such that for all $\procA \in \Procs$, $w \wproj_{\AlphAsync_{\procA}} \leq \SyncToAsync(\alpha\beta) \wproj_{\AlphAsync_{\procA}}$, and $w \wproj_{\AlphAsync_{\procB}} = \SyncToAsync(\alpha)  \wproj_{\AlphAsync_{\procB}}$, and $w \cdot \rcv{\procA}{\procB}{\val}$ is channel-compliant. Then, there exists $w'$ such that $w'$ is compliant with $\beta$, $w' \wproj_{\AlphAsync_{\procB}} = \emptystring$ and $w' \cdot \rcv{\procA}{\procB}{\val}$ is channel-compliant. \label{ca:history-insensitive}
	\end{enumerate}
\end{definition}

Channel states are maps from channels to buffer contents, and thus a word is channel-compliant if and only if its projection onto each channel is defined for the respective buffer. In the proofs below, we thus only need to reason about buffer insert and remove operations, and buffer states. 
	
\begin{lemma}[\ref{ca:send-before-receive}]
	For all $\procA \neq \procB \in \Procs$ and $\val \in \MsgVals$, $\card {w \wproj_{\snd{\procA}{\procB}{\val}}} \geq \card {w \wproj_{\rcv{\procA}{\procB}{\val}}}$. 
\end{lemma}
\begin{proof}
	By strong induction on $\card{w}$. Let $w = w'x$. 
	We case split on whether $x$ is a send or receive event. 
	If $x \in \AlphAsync_!$, the claim immediately follows. If $x \in \AlphAsync_?$, let $x = \rcv{\procC}{\procD}{\val}$. 
	We want to show that $\card {w'x \wproj_{\snd{\procC}{\procD}{\val}}} \geq \card {w'x \wproj_{\rcv{\procC}{\procD}{\val}}}$. 
	From the induction hypothesis, $\card {w' \wproj_{\snd{\procC}{\procD}{\val}}} \geq \card {w' \wproj_{\rcv{\procC}{\procD}{\val}}}$. 
	It suffices to analyze the case when $\card {w' \wproj_{\snd{\procC}{\procD}{\val}}} = \card {w' \wproj_{\rcv{\procC}{\procD}{\val}}}$, i.e. every message $\val$ sent in $w'$ is received. 
	Let $w' = w''y$. 
	Next, we case split on whether $y$ is a send or receive event. 
	If $y \in \AlphAsync_!$, and moreover $y$'s message is not equal to $\val$, then it follows from (\ref{ba:left-com-remove-neq}) that 
	$w' \cdot \rcv{\procC}{\procD}{\val} \cdot y$ 
	Let $\val_y$ be the message value of $y$. 
	If $\val_y \neq \val$, then $w = w' \cdot \snd{\procA}{\procB}{\val} \cdot \rcv{\procC}{\procD}{\val}$, and by \ref{ba:left-com-remove-neq} lifted to words, we obtain that $w'. \rcv{\procC}{\procD}{\val} \cdot \snd{\procA}{\procB}{\val}$ reaches the same channel state as $w$, and we use the induction hypothesis on $w' \cdot \rcv{\procC}{\procD}{\val}$.  
	If $\val_y = \val$, note that because messages are tagged with sender and receiver information, it follows that $y = \snd{\procC}{\procD}{\val}$. We use the induction hypothesis on $w'$ directly. 
	If $y \in \AlphAsync_?$, we first use the induction hypothesis on $w'y$ to exhibit at least one occurrence of $\snd{\procC}{\procD}{\val}$ in $w'$. We find the first occurrence of $\snd{\procC}{\procD}{\val}$ in $w'$: let $w' = u_1 \cdot \snd{\procC}{\procD}{\val} \cdot u_2$ such that $u_1 \wproj_{\snd{\procC}{\procD}{\val}} = \emptystring$. 
	If $u_1 \cdot \rcv{\procC}{\procD}{\val}$ is buffer compliant, we use the induction hypothesis on $u_1$. 
	If not, we find the first occurrence of $\rcv{\procC}{\procD}{\val}$ in $u_2$. 
	If $\rcv{\procC}{\procD}{\val}$ does not occur in $u_2$, we instantiate the induction hypothesis with $w'$. 
	Otherwise, we further split up $u_2 = v_1 \cdot \rcv{\procC}{\procD}{\val} \cdot v_2$ such that $v_1 \wproj_{\rcv{\procC}{\procD}{\val}} = \emptystring$. 
	We then use (\ref{ba:cancel-insert-remove}) on $w = u_1 \cdot \snd{\procC}{\procD}{\val} \cdot v_1 \cdot \rcv{\procC}{\procD}{\val} \cdot v_2$ to conclude that $u_1 v_1 v_2$ reaches the same state as $w$. We use the induction hypothesis on $u_1 v_1 v_2 \cdot \rcv{\procA}{\procB}{\val}$ to prove the claim. 
\end{proof}

\begin{lemma}[\ref{ca:send-enabled}]
	For all $x \in \AlphAsync_!$, $wx$ is channel compliant. 
\end{lemma}
\begin{proof}
	Follows immediately from (\ref{ba:insert-total}) stating the totality of insert. 
\end{proof} 

\begin{lemma}[\ref{ca:intro-send}]
For all $x \in \AlphAsync_!$, $y \in \AlphAsync_?$, if $wy$ is channel compliant then $wxy$ is channel compliant.
\end{lemma}
\begin{proof}
  Follows immediately from (\ref{ba:intro-insert}). 
\end{proof} 

\begin{lemma}[\ref{ca:sync}] 
	For all $\rho \in \AlphSyncSubscript^*$, $\SyncToAsync(\rho)$ is channel compliant. 
\end{lemma}
\begin{proof}
	We prove a stronger property, which is that for all $w \in \AlphAsyncSubscript^*$, $\SyncToAsync(w)$ reaches the empty channel state $\chstate_0$. We prove this by induction on $\card{w}$. The base case is trivial. 
	In the induction step, let $w = w' \cdot \msgFromTo{\procA}{\procB}{\val}$. 
	From the induction hypothesis, $\SyncToAsync(w')$ reaches the empty channel state $\chstate_0$. 
	Let $c = \buffermap(\procA,\procB)$. Then, $\chstate_0(c) = \empbuffer$. 
	By (\ref{ba:emp-insert-remove}), $\revapp{\revapp{\empbuffer}{\bufinsert(\msg)}}{\bufremove(\msg)}$ is defined. 
	By (\ref{ba:cancel-insert-remove}), $\revapp{\revapp{\empbuffer}{\bufinsert(\msg)}}{\bufremove(\msg)} = \empbuffer$. 
	Thus, $\SyncToAsync(w' \cdot \msgFromTo{\procA}{\procB}{\val})$ reaches the empty channel state, and is thus channel-compliant. 
\end{proof}

\begin{lemma}[\ref{ca:sync-matched}]
	Let $\mathcal{T}_{\netarch}$ be a CLTS, $w \in \AlphAsyncSubscript^*$ and $\rho \in \AlphSyncSubscript^*$ such that 
	$w \equiv_\Procs \SyncToAsync(\rho)$, and $w$ is a trace of $\mathcal{T}_{\netarch}$. 
	Then, in the configuration reached on $w$, all channels are empty. 
\end{lemma} 
\begin{proof}
	We prove that every channel-compliant word $\bar w$ that is per-role equal to $\SyncToAsync(\rho)$ for some $\rho \in \AlphSyncSubscript^*$ reaches the empty channel, by induction on $\card{\rho}$. 
	In the base case, $\rho = \emptystring$ and $w = \emptystring$, and the claim holds trivially. 
	In the induction step, let $w = xw'$. 
	We case split on whether $x$ is a send or receive event. 
	If $x \in \AlphAsync_?$, by assumption $xw'$ is channel-compliant, and by prefix closure of channel compliance, $x$ is channel-compliant. 
	By (\ref{ba:emp-remove}), removing from an empty channel is undefined, and we reach a contradiction. 
	Let $x \in \AlphAsync_!$ be $\snd{\procA}{\procB}{\val}$. 
	Because $xw'$ is per-role equal to $\SyncToAsync(\rho)$, there is at least one occurrence of $\rcv{\procA}{\procB}{\val}$ in $w'$. 
	We find the first occurrence, \ie let $w = \snd{\procA}{\procB}{\val} \cdot u_1 \cdot \rcv{\procA}{\procB}{\val} \cdot u_2$ such that $\rcv{\procA}{\procB}{\val}$ does not occur in $u_1$. 
	Furthermore, let $\rho = \rho_1 \cdot \msgFromTo{\procA}{\procB}{\val} \cdot \rho_2$ such that $\msgFromTo{\procA}{\procB}{\val}$ does not occur in $\rho_1$. 
	Thus, $u_1 u_2$ is per-role identical to $\rho_1 \rho_2$, and by the induction hypothesis, $u_1 u_2$ reaches the empty channel. 
	By (\ref{ba:cancel-insert-remove}) lifted to channel states, the channel state reached on $u_1 u_2$ is the same as that reached on $xw'$. We conclude that the channel state reached on $w = xw'$ is the empty channel. 
\end{proof}

\begin{lemma}[\ref{ca:history-insensitive}]
	Let $w \in \AlphAsyncSubscript^*$ be channel-compliant. Let $\alpha, \beta \in \AlphSyncSubscript^*, \procA \neq \procB \in \Procs$ and $\val \in \MsgVals$ such that for all $\procA \in \Procs$, $w \wproj_{\AlphAsync_{\procA}} \leq \SyncToAsync(\alpha\beta) \wproj_{\AlphAsync_{\procA}}$, and $w \wproj_{\AlphAsync_{\procB}} = \SyncToAsync(\alpha)  \wproj_{\AlphAsync_{\procB}}$, and $w \cdot \rcv{\procA}{\procB}{\val}$ is channel-compliant. Then, there exists $w'$ such that $w'$ is compliant with $\beta$, $w \wproj_{\AlphAsync_{\procB}} = \emptystring$ and $w' \cdot \rcv{\procA}{\procB}{\val}$ is channel-compliant.  
\end{lemma} 
\begin{proof}
	We construct $w'$ in two steps. In the first step, we remove matched pairs of events from $w$ and $\alpha$. 
	\paragraph{Claim 1.} There exists $w_1$ and $\alpha_1$ such that $w_1$ is compliant with $\alpha_1\beta$, $w_1$ contains no matched pairs of events from $\alpha$, and moreover reaches the same channel state as $w$. 
	We prove the claim by induction on the number of matched pairs in $w$ from $\alpha$. 
	In the base case, let $w_1 = w$ and $\alpha_1 = \alpha$. 
	In the induction step, let $\alpha = \gamma_1 \cdot \msgFromTo{\procC}{\procD}{\val} \cdot \gamma_2$ such that 
	$\SyncToAsync(\gamma_1 \cdot \msgFromTo{\procC}{\procD}{\val}) \wproj_{\AlphAsync_\procC} \leq w \wproj_{\AlphAsync_\procC}$, and 
	$\SyncToAsync(\gamma_1 \cdot \msgFromTo{\procC}{\procD}{\val}) \wproj_{\AlphAsync_\procD} \leq w \wproj_{\AlphAsync_\procD}$.
	Let $w_\procC$ be the maximal prefix of $w$ with respect to $\gamma_1$ for $\procC$, such that $w = w_\procC \cdot \snd{\procC}{\procD}{\val} \cdot u_\procC$. 
	Let $w_\procD$ be the maximal prefix of $w$ with respect to $\gamma_1$ for $\procD$, such that $w = w_\procD \cdot \rcv{\procC}{\procD}{\val} \cdot u_\procD$. 
	It follows from (\ref{ca:send-before-receive}) that $w_\procC \cdot \snd{\procC}{\procD}{\val} \leq w_\procD$, and we can write 
	$w = w_\procC \cdot \snd{\procC}{\procD}{\val} \cdot u \cdot \rcv{\procC}{\procD}{\val} \cdot v$ such that $\rcv{\procC}{\procD}{\val}$ does not occur in $u$.
	Thus, $w_\procD uv$ is compliant with $\gamma_1\gamma_2$, contains one fewer matched event from $\alpha$, and by (\ref{ba:cancel-insert-remove}), reaches the same channel state as $w$. 
	The claim thus follows from the induction hypothesis. 
	
	Because $w \cdot \rcv{\procA}{\procB}{\val}$ is channel-compliant, and $w_1$ reaches the same channel state as $w$, it follows that $w_1 \cdot \rcv{\procA}{\procB}{\val}$ is channel-compliant. 
	
	In the second step, we remove unmatched send events from $w_1$ and $\alpha_1$. 
	\paragraph{Claim 2.} There exists $w_2$ and $\alpha_2$ such that $w_2$ is compliant with $\alpha_2\beta$, $w_2$ contains no matched pairs of events from $\alpha$ and no unmatched sends from $\alpha$, and moreover $w_2 \cdot \rcv{\procA}{\procB}{\val}$ is channel-compliant. 
	We prove the claim by induction on the number of unmatched sends in $w_1$ from $\alpha_1$. 
	In the base case, let $w_2 = w_1$ and $\alpha_2 = \alpha_1$. 
	In the induction step, let $\alpha_1 = \gamma_3 \cdot \msgFromTo{\procC}{\procD}{\val} \cdot \gamma_4$ such that $\SyncToAsync(\gamma_3 \cdot \msgFromTo{\procC}{\procD}{\val}) \wproj_{\AlphAsync_\procC} \leq w \wproj_{\AlphAsync_\procC}$, and the maximal prefix of $w$ with respect to $\gamma_3$ for $\procC$ is the longest relative to all over unmatched send events in $\alpha_1$. 
	In other words, we find the rightmost occurrence of an unmatched send event from $\alpha_1$ in $w_1$.
	First, we establish that $\msgFromTo{\procC}{\procD}{\val} \neq \msgFromTo{\procA}{\procB}{\val}$. 
	Because $w \wproj_{\AlphAsync_{\procB}} = \SyncToAsync(\alpha)  \wproj_{\AlphAsync_{\procB}}$, and $w$ satisfies (\ref{ca:send-before-receive}), it follows that $\SyncToAsync(\alpha)  \wproj_{\AlphAsync_{\procA}} \leq w \wproj_{\AlphAsync_{\procA}}$. 
	Thus, all send events from $\procA$ to $\procB$ prescribed by $\alpha$ are not unmatched. 
	We can split $w_1 = u_1 \cdot \snd{\procC}{\procD}{\val} \cdot u_2$ such that there are no unmatched sends from $\alpha_1$ in $u_2$. 
	It follows from (\ref{ba:elim-insert}) that $u_1 \cdot u_2 \cdot \rcv{\procA}{\procB}{\val}$ is channel-compliant, contains one fewer unmatched send from $\alpha$, and is compliant with $\gamma_3\gamma_4$. 
	The claim thus follows from the induction hypothesis.  
	
	Finally, let $\alpha_3$ be $\alpha_2$ with all events that do not occur in $w_2$ removed. 
	From the two steps above, it is clear that $\alpha_3 = \emptystring$, and thus $w_2$ is compliant with $\beta$, $w_2 \cdot \rcv{\procA}{\procB}{\val}$ is channel-compliant, and contains no events of $\procB$. 
	This concludes our proof that $w_2$ serves as a witness for $w'$. 
\end{proof}

\subsection{Proofs for \cref{sec:symbolic}}
\label{app:symbolic-proofs}

\StrongPrefixExtensibility*

\begin{proof}
	We prove the claim by induction on the length of $w$. 
	In the base case, let $u = \SyncToAsync(\rho)$. By $(\ccfact{3})$, $u$ is channel-compliant, and by construction, $u \equiv_\Procs \SyncToAsync(\rho)$. 
	In the induction step, let $w = w'x$. 
	From the induction hypothesis, we are given $u'$ such that $w'u'$ is channel-compliant and $w'u' \equiv_\Procs \SyncToAsync(\rho)$. 
	Let $\procA$ be the active participant in $x$, and let $u_1$ be the maximal prefix of $u'$ with respect to $\emptystring$ for $\procA$. 
	By definition, $u_1 \wproj_{\AlphAsync_{\procA}} = \emptystring$, and the next symbol in $u'$ following $u_1$ has $\procA$ as its active participant. 
	Moreover, because both $w'x \wproj_{\AlphAsync_{\procA}}$ and $w'u' \wproj_{\AlphAsync_{\procA}}$ are prefixes of $\rho \wproj_{\AlphAsync_{\procA}}$, the next symbol in $u'$ following $u_1$ equals $x$. 
	Thus, we can write $u' = u_1 y u_2$. 
	Let $u = u_1 u_2$. Clearly, $u$ is per-role equal to $\SyncToAsync(\rho)$. 
	To show that $w'xu_1u_2$ is channel-compliant we require the following facts for the various considered network architectures. Let $v, u_1, u_2 \in \AlphAsyncSubscript^*$ and $z \in \AlphAsyncSubscript$. 
	\paragraph{Claim 1.} If $v u_1 z u_2$ is \ptpbox, \senderbox, or \bag channel-compliant, $z = \snd{\procA}{\procB}{\val}$ and $u_1 \wproj_{\AlphAsync_{\procA}} = \emptystring$, then $v z u_1 u_2$ is respectively \ptpbox, \senderbox or \bag channel-compliant. 
	\paragraph{Claim 2.}  If $v u_1 z u_2$ is \ptpbox or \bag channel-compliant, $z = \rcv{\procA}{\procB}{\val}$ and $u_1 \wproj_{\AlphAsync_{\procB}} = \emptystring$, then $v z u_1 u_2$ is respectively \ptpbox and \bag channel-compliant. 
	
	For \senderbox channel compliance when $z$ is a receive event, we require slightly stronger assumptions on $u_1$. 
	
	\paragraph{Claim 3.} 
	 If $v u_1 z u_2$ is \senderbox channel-compliant, $z = \rcv{\procA}{\procB}{\val}$, $u_1 \wproj_{\AlphAsync_{\procA}} = \emptystring$, $u_1 \wproj_{\snd{\procA}{\hole}{\hole}} = \emptystring$, and $v \wproj_{\rcv{\procA}{\hole}{\hole}} \cdot z \leq v \wproj_{\snd{\procA}{\hole}{\hole}}$, then $v z u_1 u_2$ is \senderbox channel-compliant. 
	 
	 When $z = \rcv{\procA}{\procB}{\val}$, we reason from the fact that $w' \cdot \rcv{\procA}{\procB}{\val}$ is \senderbox channel-compliant, $w' u_1 \cdot \rcv{\procA}{\procB}{\val}$ is \senderbox channel-compliant, and moreover $\procB$ has no events in $u_1$, that $\procB$ does not receive the message from $\procA$ in $u_1$, and thus no other participants receives messages from $\procA$'s \senderbox in $u_1$. Thus, the assumptions required to reorder $z$ ahead of $u_1$ are satisfied.
\end{proof}

\subsection{Complexity} 
\label{app:complexity} 
\Cref{thm:preciseness} immediately yields a decision procedure for implementability of finite protocols for $\netarch \in \set{\ptpbox, \senderbox, \mailbox, \monobox, \bag}$. 
Li et al.~\cite{Li25OopslaOfficial} showed that the implementability problem for finite global protocols on a \ptpbox network is co-NP-complete. 
We first show that this result extends to all homogeneous network architectures under consideration, by examining the complexity of problem in light of relationships between the $\avail$ predicates for each network. 
Given our discussion of how to decide the generalized Coherence Conditions in \cref{sec:instantiations}, it suffices to argue that (a) $\avail$ for \senderbox and \bag are co-NP-complete, and (b) $\avail'$ for \mailbox and \monobox are co-NP-complete. 
It is easy to see that witnesses for the above are verifiable in polynomial time. 

The proof of the co-NP lower bound by Li et al.~\cite{Li25OopslaOfficial} works by a reduction from 3-SAT to implementability. The proof assumes a 3-SAT instance $\varphi = C_1 \land \ldots \land C_k$ with variables $x_1,\dots,x_n$ and literals $L_{ij}$, denoting the $j$th literal of clause~$C_i$, with $1 \leq i \leq k$ and $1 \leq j \leq 3$. From this, it constructs a global protocol $\mathcal{S}_\varphi$ such that $\phi$ is unsatisfiable iff $\mathcal{S}_\varphi$ is implementable. We summarize the construction pictorially in \cref{fig:3-sat-reduction}.
We show that the same lower bound construction with a small modification works for the remaining network architectures.

The construction relies on two gadgets: $\mathcal{S}_X$, a gadget that encodes a variable assignment to variables $x_1, \ldots,x_n$ (\cref{fig:3-sat-assignment-gadget}), and $\mathcal{S}_C$, a gadget that encodes literal selection for clauses $C_1,\ldots,C_k$ (\cref{fig:3-sat-clause-selection-gadget}). 
The highlighted message in \cref{fig:3-sat-reduction} is available for participant $\procB$ in $q_2$ if and only if $\varphi$ is satisfiable. 
Consequently, protocol $\mathcal{S}_\varphi$ is non-implementable if and only if the highlighted message is available for participant $\procB$ in $q_2$, if and only if $\varphi$ is satisfiable. 
In order to show that the construction carries over to bag and senderbox networks, one is only required to analyze the unique message receptions that appear in each gadget. Thus, \cref{fig:3-sat-assignment-gadget} and \cref{fig:3-sat-clause-selection-gadget} each depict the smallest portion of each gadget necessary to establish our new complexity results. 
We refer the reader to \cite{Li25OopslaOfficial} for the full details of the construction. 

\newcommand{\redprocC}{\textcolor{red}{\mathtt{r}}}
\newcommand{\redprocD}{\textcolor{red}{\mathtt{s}}}

\begin{figure}[b]
	\begin{tikzpicture}[sem,
	circ/.style={circle, draw, minimum size=8mm},
	square/.style={rectangle, draw, minimum size=10mm},
	node distance=2cm
	]
	\node[state, initial left=, initial text =](q1){$q_1$};
	% Upper branch
	\node[state, above right=0.8cm and 1.5cm of q1](q4){$q_4$};
	\node[state, accepting, right=1.5cm of q4](q5){$q_5$};
	\path(q1) edge node[sloped, pos=0.5, above] {$\msgFromTo{\procC}{\procA}{\val_2}$} (q4);
	\path(q4) edge node[sloped, pos=0.5, above] {$\msgFromTo{\procA}{\procB}{\val}$} (q5);
	
	% Lower branch 
	\node[state, below right=0.8cm and 1.5cm of q1](q2){$q_2$};
	\node[square, right=1.5cm of q2](q3){$\mathcal{S}_X$}; 
	%\node[state, right=2cm of q3](q34){};
	\node[square, right=3cm of q3](q4){$\mathcal{S}_C$};
	%\node[state, right=2cm of q4](q45){};
	\path(q1) edge node[sloped, pos=0.5, below] {$\msgFromTo{\procC}{\procA}{\val_1}$} (q2);
	\path(q2) edge node[sloped, pos=0.5, below] {$\msgFromTo{\procC}{\procB}{\val}$} (q3);
	\path(q3) edge node[sloped, pos=0.5, below] {$\ldots$} (q4); 
	% \path(q34) edge node[sloped, pos=0.5, below] {} (q4);
	\node[state, accepting, right=1.5cm of q4](qf){$q_f$};
	\path(q4) edge node[sloped, pos=0.5, below] {\colorbox{yellow}{$\msgFromTo{\procA}{\procB}{\val}$}} (qf);
	% Lower branch
%	\node[state, below right=0.3cm and 1.5cm of q0](q5){};
%	\node[state, right= 2.5cm of q5](q6){$q_2$};
%	\node[state, accepting, above right=0.07cm and 2.5cm of q6](q6l){};
%	\node[state, below right=0.07cm and 2.5cm of q6](q6r){};
%	\node[state, accepting, right= 1.5cm of q6r](q7){};
%	
%	% Upper branch
%	\path(q0) edge node[sloped, pos=0.5, above] {$\msgFromTo{\procD}{\procB}{\lblmsgB}$} (q1);
%	\path(q1) edge node[sloped, pos=0.5, above] {$\msgFromTo{\procB}{\procA}{x_1 \cond{x_1 = 4}}$} (q2);
%	\path(q2) edge node[sloped, pos=0.5, above] {$\msgFromTo{\procA}{\procB}{\lblmsgO}$} (q3);
%	\path(q3) edge node[sloped, pos=0.5, above] {$\msgFromTo{\procB}{\procC}{\lblmsgM}$} (q4);
%	
%	% Lower branch
%	\path(q0) edge node[sloped, pos=0.5, below] {$\msgFromTo{\procD}{\procB}{\lblmsgM}$} (q5);
%	\path(q5) edge node[below] {$\msgFromTo{\procB}{\procA}{x_2 \cond{\top}}$} (q6);
%	\path(q6) edge node[sloped, pos=0.5, above] {$\msgFromTo{\procB}{\procC}{\lblmsgB \cond{x_2 \neq 5}}$} (q6l);
%	\path(q6) edge node[sloped, pos=0.5, below] {$\msgFromTo{\procB}{\procC}{\lblmsgM \cond{x_2 = 5}}$} (q6r);
%	\path(q6r) edge node[below] {$\msgFromTo{\procA}{\procB}{\lblmsgO}$} (q7);
\end{tikzpicture}
	\caption{Illustration of 3-SAT reduction for implementability from \cite{Li25OopslaOfficial}.}
	\label{fig:3-sat-reduction}
\end{figure}
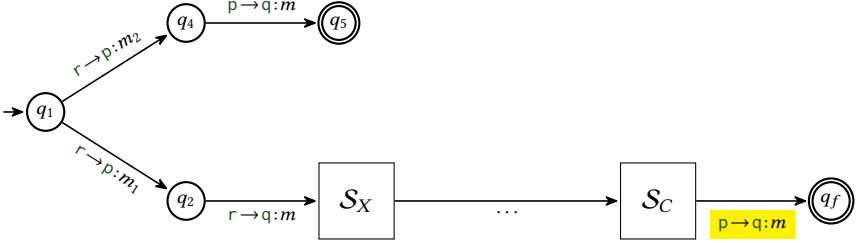

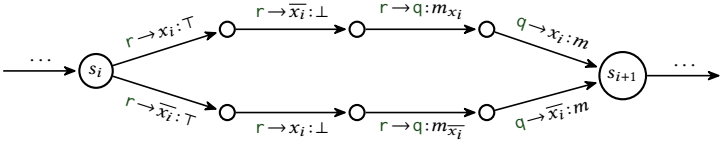
\begin{figure}[b]
	\begin{tikzpicture}[sem] % node distance=1cm and 2cm,>=stealth', line width=0.25mm]
	\node[draw=none,fill=none](q00){};
	\node[state, right=1cm of q00](q0){$s_i$};
	\path(q00) edge node[sloped, pos=0.5, above]{$\ldots$} (q0);
	% Upper branch
	\node[state, above right=0.3cm and 1.5cm of q0](q1){};
	\node[state, right=1.5cm of q1](q2){};
	\node[state, right=1.5cm of q2](q3){};
	\node[state, below right=0.3cm and 1.5cm of q3](q4){$s_{i+1}$};
	\node[draw=none,fill=none, right=1cm of q4](q55){};
	\path(q4) edge node[sloped, pos=0.5, above]{$\ldots$} (q55);

	% Upper branch
	\path(q0) edge node[sloped, pos=0.5, above] {$\msgFromTo{\procC}{x_i}{\top}$} (q1);
	\path(q1) edge node[sloped, pos=0.5, above] {$\msgFromTo{\procC}{\overline{x_i}}{\bot}$} (q2);
	\path(q2) edge node[sloped, pos=0.5, above] {$\msgFromTo{\procC}{\procB}{\val_{x_i}}$} (q3);
	\path(q3) edge node[sloped, pos=0.5, above] {$\msgFromTo{\procB}{x_i}{\val}$} (q4);

	% Lower branch
	\node[state, below right=0.3cm and 1.5cm of q0](q6){};
	\node[state, right=1.5cm of q6](q7){};
	\node[state, right=1.5cm of q7](q8){};
	
	% Lower branch
	\path(q0) edge node[sloped, pos=0.5, below] {$\msgFromTo{\procC}{\overline{x_i}}{\top}$} (q6);
	\path(q6) edge node[sloped, pos=0.5, below] {$\msgFromTo{\procC}{x_i}{\bot}$} (q7);
	\path(q7) edge node[sloped, pos=0.5, below] {$\msgFromTo{\procC}{\procB}{\val_{\overline{x_i}}}$} (q8);
	\path(q8) edge node[sloped, pos=0.5, below] {$\msgFromTo{\procB}{\overline{x_i}}{\val}$} (q4);

\end{tikzpicture}
	\caption{Illustration of modified variable assignment gadget $\mathcal{S}_X$.}
	\label{fig:3-sat-assignment-gadget}
\end{figure}

\begin{figure}[b]
	\begin{tikzpicture}[sem] % node distance=1cm and 2cm,>=stealth', line width=0.25mm]
	% Leftmost states 
	\node[draw=none,fill=none](q00){};
	\node[state, right=1cm of q00](q0){$t_i$};
	\path(q00) edge node[sloped, pos=0.5, above]{$\ldots$} (q0);
	
	% Upper branch
	\node[state, above right=1.5cm and 1.5cm of q0](q1){};
	\node[state, right=1.5cm of q1](q12){};
	\node[state, right=1.5cm of q12](q2){};
	
	% Upper branch
	\node[state, right=1cm and 1.5cm of q0](q3){};
	\node[state, right=1.5cm of q3](q34){};
	\node[state, right=1.5cm of q34](q4){};
	
	% Lower branch
	\node[state, below right=1.5cm and 1.5cm of q0](q5){};
	\node[state, right=1.5cm of q5](q56){};
	\node[state, right=1.5cm of q56](q6){};
	
	% Rightmost states 
	\node[state, right=1.5cm of q4](q7){$t_{i+1}$};
	\node[draw=none,fill=none, right=1cm of q7](q8){};
	\path(q7) edge node[sloped, pos=0.5, above]{$\ldots$} (q8);
	
	% First transitions 
	\path(q0) edge node[sloped, pos=0.5, above] {$\msgFromTo{\procC}{x_a}{\val_1}$} (q1);
	\path(q0) edge node[sloped, pos=0.5, above] {$\msgFromTo{\procC}{x_b}{\val_2}$} (q3);
	\path(q0) edge node[sloped, pos=0.5, below] {$\msgFromTo{\procC}{x_c}{\val_3}$} (q5);
	
	% Second transitions 
	\path(q1) edge node[sloped, pos=0.5, above] {$\msgFromTo{\procC}{\procA}{\textcolor{red}{\val_i}}$} (q12);
	\path(q3) edge node[sloped, pos=0.5, above] {$\msgFromTo{\procC}{\procA}{\textcolor{red}{\val_i}}$} (q34);
	\path(q5) edge node[sloped, pos=0.5, below] {$\msgFromTo{\procC}{\procA}{\textcolor{red}{\val_i}}$} (q56);
	
	% New third transitions 
	\path(q12) edge node[sloped, pos=0.5, above] {\textcolor{red}{$\msgFromTo{\redprocC}{x_a}{\val}$}} (q2);
	\path(q34) edge node[sloped, pos=0.5, above] {\textcolor{red}{$\msgFromTo{\redprocC}{x_b}{\val}$}} (q4);
	\path(q56) edge node[sloped, pos=0.5, below] {\textcolor{red}{$\msgFromTo{\redprocC}{x_c}{\val}$}} (q6);
	
	% Third transitions 
	\path(q2) edge node[sloped, pos=0.5, above] {$\msgFromTo{x_a}{\procA}{\val_1}$} (q7);
	\path(q4) edge node[sloped, pos=0.5, above] {$\msgFromTo{x_b}{\procA}{\val_2}$} (q7);
	\path(q6) edge node[sloped, pos=0.5, below] {$\msgFromTo{x_c}{\procA}{\val_3}$} (q7);
	
\end{tikzpicture}
	\caption{Illustration of modified clause selection gadget $\mathcal{S}_C$. Modifications are highlighted in red.}
	\label{fig:3-sat-clause-selection-gadget}
\end{figure}
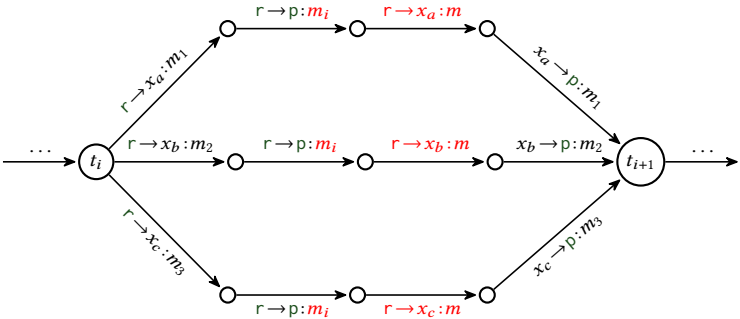

First, we consider the \bag network architecture. 
As established in \cref{sec:instantiations}, any message available in a \ptpbox network is also available in a \bag network. Thus, we only need to show that the rest of $\mathcal{S}_\varphi$ is implementable, which amounts to checking that no other \bag Receive Coherence violations occur. 
Participant $\procC$ does not receive messages, and can thus be ignored. 
Unlike \ptpbox Receive Coherence, \bag Receive Coherence additionally constrains pairs of receptions from the same sender. 
For the variable participants in $\mathcal{S}_X$, each participant receives either a $\bot$ from $\procC$, or a $\top$ followed by an $\val$ message from $\procB$. Thus, \bag Receive Coherence is satisfied. 
Inspecting $\mathcal{S}_C$, each variable participant only receives one kind of message, which is $\val$ from $\procC$, and if so it sends an $\val$ message to $\procA$. Thus, \bag Receive Coherence is satisfied as well. 
Participant $\procA$ is uninvolved in $\mathcal{S}_X$, but in $\mathcal{S}_C$ receives $\val$ messages from $\procC$ which tells it to anticipate a message from some variable participant. The original encoding uses the same message payload from $\procC$ to tell $\procA$ to anticipate a message, but we can modify the construction to let $\procC$ send $\procA$ a message encoding precisely which variable participant to anticipate a message from. This eliminates what would otherwise constitute a \bag Receive Coherence violation for $\procA$, since $\procC$ and the variable participants can overtake $\procA$'s receptions. 
Finally, onto participant $\procB$, who is uninvolved in $\mathcal{S}_C$ and only involved in $\mathcal{S}_X$, $\procB$ receives exactly $n$ messages from $\procC$, that constitute $n$ binary choices between receiving $\val_{x_i}$ and $\val_{\bar{x_i}}$ interrupted by send events from $\procB$. 
Thus, we can conclude that the modified construction is non-implementable iff $\bagavail_{\procA,\procB,\{\procB\}}(\val,q_3)$ holds in $\mathcal{S}_\varphi$ iff $\varphi$ is satisfiable.

Next, we consider the \senderbox network architecture. 
Because \ptpbox implementability implies \senderbox implementability, in this case we only need to independently establish that $\avail_{\procA,\procB,\{\procB\}}(\val,q_3)$ holds for the \senderbox $\avail$ setting. 
As illustrated above, senderbox receivability can be undermined by messages from the same sender to different receivers that are blocked, so we need to check whether any such messages from $\procA$ to other receivers appear in the subprotocols $\mathcal{S}_X$ and $\mathcal{S}_C$. 
It is easy to see that no such messages appear, and thus senderbox implementability still holds. 

For \monobox and \mailbox, we establish that $\avail'_{\procA, \set{\procC,\procA}}(q_2)$ holds if and only if $\varphi$ is satisfiable. 
In gadget $\mathcal{S}_X$, no causally independent receptions to the same receiver occur along the same run. 
In gadget $\mathcal{S}_C$, the addition of the extra message exchange, $\msgFromTo{\procC}{x_a}{\val}$ enforces that the messages from $\procC$ and $x_a$ cannot be independently reordered in $\procA$'s mailbox. 
Thus, the construction maintains that the only possible violation to Generalized Coherence Conditions lies in the availability of the highlighted message from state $q_2$. 

Furthermore, we argue that the construction can be adapted to multiparty session types with directed choice, a syntactically defined fragment of finite global protocols, whose definition is given below: 
\paragraph{Global Multiparty Session Types}
Global types for MSTs 
\cite{DBLP:conf/popl/HondaYC08} are defined by the grammar:
\begin{grammar}
	G \is
	0
	\mid \sum_{i ∈ I} \msgFromTo{\procA}{\procB}{\val_i.G_i}
	\mid \mu t. \; G
	\mid t
\end{grammar}\\[-3ex]
where $\procA, \procB$ range over $\Procs$, $\val_i$ over a finite set $\MsgVals$, and $t$ over a set of recursion variables. 

Note that the top-level choice in \cref{fig:3-sat-reduction} satisfies directed choice, and thus we only require to modify the assignment and clause selection gadgets. 
We depict the gadgets modified to satisfy directed choice in \cref{fig:3-sat-assignment-gadget-mst} and \cref{fig:3-sat-clause-selection-gadget-mst}. 
The encoding from directed choice global protocols to directed choice multiparty session types follows~\cite{Li25OopslaArxiv}, and we refer the reader to their appendix for details. 

\begin{figure}[h]
	\begin{tikzpicture}[sem] % node distance=1cm and 2cm,>=stealth', line width=0.25mm]
	\node[draw=none,fill=none](q00){};
	\node[state, right=1cm of q00](q0){$s_i$};
	\path(q00) edge node[sloped, pos=0.5, above]{$\ldots$} (q0);
	
	% Upper branch
	\node[state, above right=0.3cm and 1.5cm of q0](q1pre){};
	\node[state, right=1.5cm of q1pre](q1){};
	\node[state, right=1.5cm of q1](q2){};
	\node[state, right=1.5cm of q2](q3){};
	\node[state, below right=0.3cm and 1.5cm of q3](q4){$s_{i+1}$};
	\node[draw=none,fill=none, right=1cm of q4](q55){};
	\path(q4) edge node[sloped, pos=0.5, above]{$\ldots$} (q55);
	
	% Upper branch
	\path(q0) edge node[sloped, pos=0.5, above] {\textcolor{red}{$\msgFromTo{\redprocC}{\redprocD}{\top}$}} (q1pre);
	\path(q1pre) edge node[sloped, pos=0.5, above] {$\msgFromTo{\procC}{x_i}{\top}$} (q1);
	\path(q1) edge node[sloped, pos=0.5, above] {$\msgFromTo{\procC}{\overline{x_i}}{\bot}$} (q2);
	\path(q2) edge node[sloped, pos=0.5, above] {$\msgFromTo{\procC}{\procB}{\val_{x_i}}$} (q3);
	\path(q3) edge node[sloped, pos=0.5, above] {$\msgFromTo{\procB}{x_i}{\val}$} (q4);
	
	% Lower branch
	\node[state, below right=0.3cm and 1.5cm of q0](q6pre){};
	\node[state, right=1.5cm of q6pre](q6){};
	\node[state, right=1.5cm of q6](q7){};
	\node[state, right=1.5cm of q7](q8){};
	
	% Lower branch
	\path(q0) edge node[sloped, pos=0.5, below] {\textcolor{red}{$\msgFromTo{\redprocC}{\redprocD}{\bot}$}} (q6pre);
	\path(q6pre) edge node[sloped, pos=0.5, below] {$\msgFromTo{\procC}{\overline{x_i}}{\top}$} (q6);
	\path(q6) edge node[sloped, pos=0.5, below] {$\msgFromTo{\procC}{x_i}{\bot}$} (q7);
	\path(q7) edge node[sloped, pos=0.5, below] {$\msgFromTo{\procC}{\procB}{\val_{\overline{x_i}}}$} (q8);
	\path(q8) edge node[sloped, pos=0.5, below] {$\msgFromTo{\procB}{\overline{x_i}}{\val}$} (q4);
	
\end{tikzpicture}
	\caption{Illustration of directed choice variable assignment gadget $\mathcal{S}_X$. Modifications are highlighted in red. }
	\label{fig:3-sat-assignment-gadget-mst}
\end{figure}
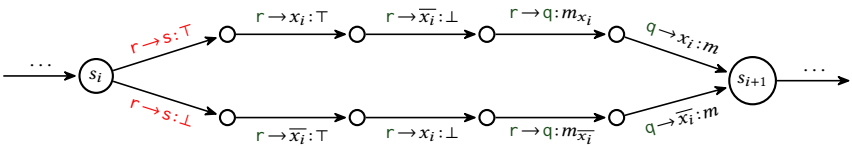

\begin{figure}[h]
	\begin{tikzpicture}[sem] % node distance=1cm and 2cm,>=stealth', line width=0.25mm]
	% Leftmost states 
	\node[draw=none,fill=none](q00){};
	\node[state, right=1cm of q00](q0){$t_i$};
	\path(q00) edge node[sloped, pos=0.5, above]{$\ldots$} (q0);
	
	% Upper branch
	\node[state, above right=1.5cm and 1.5cm of q0](q1pre){};
	\node[state, right=1.5cm of q1pre](q1){};
	\node[state, right=1.5cm of q1](q12){};
	\node[state, right=1.5cm of q12](q2){};
	
	% Upper branch
	\node[state, right=1cm and 1.5cm of q0](q3pre){};
	\node[state, right=1.5cm of q3pre](q3){};
	\node[state, right=1.5cm of q3](q34){};
	\node[state, right=1.5cm of q34](q4){};
	
	% Lower branch
	\node[state, below right=1.5cm and 1.5cm of q0](q5pre){};
	\node[state, right=1.5cm of q5pre](q5){};
	\node[state, right=1.5cm of q5](q56){};
	\node[state, right=1.5cm of q56](q6){};
	
	% Rightmost states 
	\node[state, right=1.5cm of q4](q7){$t_{i+1}$};
	\node[draw=none,fill=none, right=1cm of q7](q8){};
	\path(q7) edge node[sloped, pos=0.5, above]{$\ldots$} (q8);
	
	% New first transitions 
	\path(q0) edge node[sloped, pos=0.5, above] {\textcolor{red}{$\msgFromTo{\redprocC}{\redprocD}{\val_1}$}} (q1pre);
	\path(q0) edge node[sloped, pos=0.5, above] {\textcolor{red}{$\msgFromTo{\redprocC}{\redprocD}{\val_2}$}} (q3pre);
	\path(q0) edge node[sloped, pos=0.5, below] {\textcolor{red}{$\msgFromTo{\redprocC}{\redprocD}{\val_3}$}} (q5pre);
	
	% First transitions 
	\path(q1pre) edge node[sloped, pos=0.5, above] {$\msgFromTo{\procC}{x_a}{\val_1}$} (q1);
	\path(q3pre) edge node[sloped, pos=0.5, above] {$\msgFromTo{\procC}{x_b}{\val_2}$} (q3);
	\path(q5pre) edge node[sloped, pos=0.5, below] {$\msgFromTo{\procC}{x_c}{\val_3}$} (q5);
	
	% Second transitions 
	\path(q1) edge node[sloped, pos=0.5, above] {$\msgFromTo{\procC}{\procA}{\textcolor{red}{\val_i}}$} (q12);
	\path(q3) edge node[sloped, pos=0.5, above] {$\msgFromTo{\procC}{\procA}{\textcolor{red}{\val_i}}$} (q34);
	\path(q5) edge node[sloped, pos=0.5, below] {$\msgFromTo{\procC}{\procA}{\textcolor{red}{\val_i}}$} (q56);
	
	% New third transitions 
	\path(q12) edge node[sloped, pos=0.5, above] {\textcolor{red}{$\msgFromTo{\redprocC}{x_a}{\val}$}} (q2);
	\path(q34) edge node[sloped, pos=0.5, above] {\textcolor{red}{$\msgFromTo{\redprocC}{x_b}{\val}$}} (q4);
	\path(q56) edge node[sloped, pos=0.5, below] {\textcolor{red}{$\msgFromTo{\redprocC}{x_c}{\val}$}} (q6);
	
	% Third transitions 
	\path(q2) edge node[sloped, pos=0.5, above] {$\msgFromTo{x_a}{\procA}{\val_1}$} (q7);
	\path(q4) edge node[sloped, pos=0.5, above] {$\msgFromTo{x_b}{\procA}{\val_2}$} (q7);
	\path(q6) edge node[sloped, pos=0.5, below] {$\msgFromTo{x_c}{\procA}{\val_3}$} (q7);
	
\end{tikzpicture}
	\caption{Illustration of directed choice clause selection gadget $\mathcal{S}_C$. Modifications are highlighted in red. }
	\label{fig:3-sat-clause-selection-gadget-mst}
\end{figure}
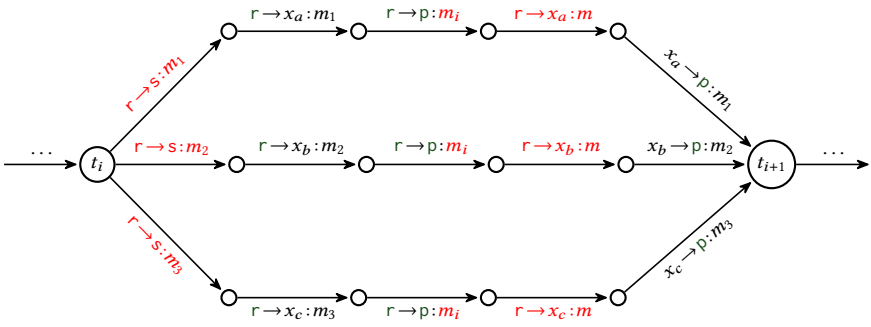

%%% Local Variables:
%%% mode: latex
%%% TeX-master: "main"
%%% End:

	%\section*{Data-Availability Statement}
	%The extended version of this paper containing complete proofs can be found at~\cite{extendedversion}. 
	% \subsubsection{Acknowledgements} Please place your acknowledgments at
	% the end of the paper, preceded by an unnumbered run-in heading (i.e.
	% 3rd-level heading).
%	\begin{comment} 
%	\subsubsection*{Acknowledgements}
%	\end{comment}

	%\bibliographystyle{plain}

%	\clearpage
%	\appendix
%	\input{app-characterization.tex}
%	\input{app-instantiations.tex}

\end{document}